\documentclass[11pt,a4paper,reqno]{amsart}
\usepackage{amsthm,amsfonts,amsmath,amssymb,amsrefs,bookmark,appendix}
\usepackage{amsxtra, mathrsfs}
\usepackage{colordvi}
\usepackage[usenames,dvipsnames]{color}
\usepackage{dsfont}
\usepackage{enumitem}

\usepackage{graphics}
\usepackage{tikz}
\usepackage{float}
\usepackage{hyperref}

\theoremstyle{plain}
\newtheorem{theorem}{Theorem}
\newtheorem{lemma}[theorem]{Lemma}
\newtheorem{corollary}[theorem]{Corollary}
\newtheorem{proposition}[theorem]{Proposition}
\theoremstyle{remark}
\newtheorem{rem}[theorem]{\bf Remark}
\newtheorem{definition}{\bf Definition}
\newtheorem{notation}{\bf Notation}

\numberwithin{equation}{section}

\DeclareMathOperator{\spec}{spec}
\DeclareMathOperator{\supp}{supp}
\DeclareMathOperator{\loc}{loc}
\DeclareMathOperator{\dist}{dist}
\DeclareMathOperator{\dom}{dom}
\DeclareMathOperator{\vol}{vol}
\DeclareMathOperator{\sing}{sing}
\DeclareMathOperator{\ran}{ran}
\DeclareMathOperator{\st}{\textit{st}}
\DeclareMathOperator{\link}{Link}
\DeclareMathOperator{\tr}{tr}
\DeclareMathOperator{\res}{res}
\DeclareMathOperator{\rank}{rank}
\DeclareMathOperator{\Sf}{sf}
\DeclareMathOperator{\sign}{sign}
\DeclareMathOperator{\Hopf}{Hopf}
\DeclareMathOperator{\per}{per}
\DeclareMathOperator{\hol}{hol}
\DeclareMathOperator{\Wr}{Wr}

\renewcommand{\phi}{\varphi}
\newcommand{\eps}{\varepsilon}

\newcommand{\norm}[1]{\lVert #1 \rVert}

\newcommand{\cip}[2]{\langle #1, #2 \rangle}
\newcommand{\wt}[1]{\widetilde{#1}}
\newcommand{\nt}{\widetilde{\boldsymbol{\sigma}}^{(n)}(-i\widetilde{\nabla}^{(n)})}
\newcommand{\nn}{\nonumber}
\newcommand{\Sup}[1]{(#1)_{+}}
\newcommand{\Sdo}[1]{(#1)_{-}}

\newcommand{\R}{\mathbb{R}}
\newcommand{\C}{\mathbb{C}}
\newcommand{\Z}{\mathbb{Z}}
\newcommand{\T}{\mathbb{T}}
\newcommand{\N}{\mathbb{N}}
\renewcommand{\S}{\mathbb{S}}

\newcommand{\cB}{\mathcal{B}}
\newcommand{\cC}{\mathcal{C}}
\newcommand{\cD}{\mathcal{D}}
\newcommand{\cE}{\mathcal{E}}

\newcommand{\cG}{\mathcal{G}}
\newcommand{\cH}{\mathcal{H}}

\newcommand{\cL}{\mathcal{L}}
\newcommand{\cM}{\mathcal{M}}
\newcommand{\cN}{\mathcal{N}}
\newcommand{\cQ}{\mathcal{Q}}
\newcommand{\cT}{\mathcal{T}}

\newcommand{\rT}{\mathrm{T}}
\renewcommand{\d}{\mathrm{d}}
\newcommand{\reg}{\mathrm{reg}}
\newcommand{\crit}{\mathrm{crit}}

\newcommand{\sC}{\mathscr{C}}
\newcommand{\sD}{\mathscr{D}}

\newcommand{\sK}{\mathscr{K}}

\newcommand{\sR}{\mathscr{R}}
\newcommand{\sS}{\mathscr{S}}
\newcommand{\sT}{\mathscr{T}}

\newcommand{\ul}{\underline{\lambda}}
\newcommand{\ua}{\underline{\alpha}}
\newcommand{\ug}{\underline{\gamma}}
\newcommand{\uS}{\underline{S}}
\newcommand{\pxi}[1]{P_{\xi_{#1}}}
\newcommand{\Sd}{\mathcal{S}_{\mathrm{disc}}}
\newcommand{\SO}{\mathbf{SO}}

\newcommand{\bA}{\boldsymbol{A}}
\newcommand{\bB}{\boldsymbol{B}}
\newcommand{\bG}{\boldsymbol{G}}
\newcommand{\bN}{\boldsymbol{N}}
\newcommand{\bS}{\boldsymbol{S}}
\newcommand{\bT}{\boldsymbol{T}}

\newcommand{\bn}{\boldsymbol{n}}
\newcommand{\bx}{\boldsymbol{x}}
\newcommand{\be}{\boldsymbol{e}}
\newcommand{\bv}{\boldsymbol{v}}
\newcommand{\bp}{\boldsymbol{p}}
\newcommand{\bt}{\boldsymbol{t}}
\newcommand{\bs}{\boldsymbol{s}}

\newcommand{\bsigma}{\boldsymbol{\sigma}}
\newcommand{\bal}{\boldsymbol{\alpha}}
\newcommand{\bbet}{\boldsymbol{\beta}}

\newcommand{\Tf}{[0,1]_{\per}}

\newcommand{\bbd}{\mathbb{D}}

\begin{document}

\title[]{Spectral flow for Dirac operators with magnetic links}
\date{}

\author[F. Portmann]{Fabian Portmann}
\address[F. Portmann]{QMATH, Department of Mathematical Sciences, University of Copenhagen
Universitetsparken 5, 2100 Copenhagen, DENMARK\newline
Current Address: IBM Switzerland, Vulkanstrasse 106 Postfach, 8010 Z\"{u}rich, Switzerland} 
\email{f.portmann@bluewin.ch}

\author[J. Sok]{J\'er\'emy Sok}
\address[J. Sok]{QMATH, Department of Mathematical Sciences, University of Copenhagen
Universitetsparken 5, 2100 Copenhagen, DENMARK\newline
Current Address: University of Basel, Department of Mathematics and Computer Science, Spiegelgasse 1, CH-4051 Basel, Switzerland.
} 
\email{jeremyvithya.sok@unibas.ch}

\author[J. P. Solovej]{Jan Philip Solovej}
\address[J. P. Solovej]{Corresponding author. \newline QMATH, Department of Mathematical Sciences, University of Copenhagen
Universitetsparken 5, 2100 Copenhagen, DENMARK
} 
\email{solovej@math.ku.dk}

\thanks{The authors acknowledge support from the ERC grant Nr.\ 321029 ``The mathematics of the structure of matter" and
VILLUM FONDEN through the QMATH Centre of Excellence grant. nr.\ 10059.
This work has been done when all the authors were working at the university of Copenhagen.}

\begin{abstract}
	This paper is devoted to the study of the spectral properties
        of Dirac operators on the three-sphere with singular magnetic fields
        supported on smooth, oriented links. As for Aharonov-Bohm
        solenoids in Euclidean three-space, the flux carried by an
        oriented knot features a $2\pi$-periodicity of the associated
        operator.  For a given link one thus obtains a family of Dirac
        operators indexed by a torus of fluxes. We study the spectral
        flow of paths of such operators corresponding to loops in this
        torus. The spectral flow is in general non-trivial.  In the
        special case of a link of unknots we derive an explicit
        formula for the spectral flow of any loop on the torus of
        fluxes.  It is given in terms of the linking numbers of the
        knots and their writhes.
\end{abstract}


\maketitle
\tableofcontents

\section{Introduction}
In \cite{dirac_s3_paper1} we introduced Dirac operators with magnetic
fields supported on links.  These magnetic fields, also called
magnetic links, are described by a smooth, oriented link in $\S^3$,
together with fluxes $2\pi\alpha_k$ on each component $\gamma_k$.
They can be seen as natural generalizations of the celebrated
Aharonov-Bohm magnetic solenoids, which are magnetic fields supported
on lines in Euclidean three space.  As shown in
\cite{dirac_s3_paper1}, the associated Dirac operators are
self-adjoint, have discrete spectrum and are Fredholm. 
They correspond to specific boundary conditions on the link, a co-dimension $2$ boundary.

These operators possess an inherent $2\pi$-periodicity in the fluxes
$2\pi \alpha_k$, which is manifest when using a singular magnetic
gauge potential (to be elaborated on below).  One is therefore incited
to study the spectral flow, i.e. the net number of eigenvalues
(counting multiplicity) crossing the spectral point $0$ from negative
to positive, when varying such an $\alpha_k$ from $0$ to $1$.  To
ensure that the spectral flow is well-defined, we choose to study the
problem on $\S^3$ rather than $\R^3$, where the corresponding
operators would not have discrete spectrum. 

For a knot $\gamma$ we find that the value of the spectral flow
depends on the {\it writhe} \cite{Fuller71} $\mathrm{Wr}(\gamma)$ of
the knot.  We take the writhe of a closed oriented curve $\gamma$ on
$\S^3$ to be given by the writhe of the image $\gamma_0$ of $\gamma$
through any stereographic projection onto $\R^3$.  The formula of the
writhe for the space curve $\gamma_0$ is:
\[
\mathrm{Wr}(\gamma_0)=\frac{1}{4\pi}\int_{\gamma_0}\int_{\gamma_0}
\left\langle\d \mathbf{r}_1\times \d
\mathbf{r}_2,\frac{\mathbf{r}_1-\mathbf{r}_2}{|\mathbf{r}_1-\mathbf{r}_2|^3}\right\rangle_{\R^3},
\]
which corresponds formally to the Gauss' linking number formula where
instead of integrating over two non-intersecting curves $\gamma_1$ and
$\gamma_2$, we integrate over $\gamma_1=\gamma_0$ and
$\gamma_2=\gamma_0$.  Alternatively we prove in
Appendix~\ref{sec:techn} that the writhe of $\gamma$ is equal to
$-I_\tau(\gamma)/(2\pi)$, where $I_\tau(\gamma)$ is the integral over
$\gamma$ of the relative torsion associated to the Darboux frame of a
Seifert surface for $\gamma$, that is, an oriented compact surface with boundary
$\gamma$ (which exists for all $\gamma$, see
\eqref{eq:eq_darboux_frame}).  In Appendix~\ref{sec:techn} we show,
moreover, that $I_\tau(\gamma)$ is independent of the choice of
Seifert surface, is conformally invariant, isotopy continuous, and
equal to the total torsion modulo $2\pi$.

We show, in particular, that by deforming any given knot the spectral
flow can attain all values.  This stands in strong contrast to the
smooth case, where the spectral flow is always trivial.
This follows from the celebrated Atiyah-Patodi-Singer
Index Theorem: in the smooth setting, the spectral flow of a closed loop
of Dirac operators on a closed manifold with boundary, defined by a smooth
variation of smooth magnetic fields (but otherwise fixed) is always
trivial.
Note however that
on a manifold with boundary that the spectral flow can attain any value,
for instance by choosing an appropriate loop of boundary conditions for a fixed Dirac operator.

\subsection{Overview}
The spectral flow was first introduced by Atiyah, Patodi and Singer in
\cites{ASP1,ASP3} to obtain an index Theorem for elliptic operators on
vector bundles over compact manifolds with boundary. In particular,
they proved the following.  Under appropriate assumptions, the
spectral flow of a first order smooth family $(D_t)_{t\in \S^1}$ of
self-adjoint elliptic operators on a closed, orientable,
odd-dimensional, compact manifold $X$ is equal to the index of the
elliptic operator $\partial_t+D_t$ on $\S^1 \times X$.  The general
formula for the spectral flow of smooth open paths $(D_t)_{t\in
  [0,1]}$ of self-adjoint elliptic operators involves the eta
invariant of the endpoints $\eta_{D_0}(0)$ and $\eta_{D_1}(0)$, and
the dimension of their respective kernels.  We refrain from giving an
overview of the (vast) literature; for more details on the index
theorems and the eta invariant we refer the reader to e.g.
\cites{AS68,ASP1,ASP2,ASP3,getzler,Melrose,Grubb05}.  Let us
nevertheless emphasize that the computation of the eta invariant is a
difficult problem.  It has been computed explicitly in some cases, see
for instance \cite{hitchin}*{Section 3.1}, the survey \cite{Goette12}
and the references therein.

Similarly, it is impossible to give a complete overview on the works
devoted to the spectral flow, and we only mention some works of
importance for the present discussion.  In
\cite{Philips_spectral_flow}, the author gives an equivalent
definition of the spectral flow using the functional calculus. It is
shown that the spectral flow of a continuous path $(D_t)_{t\in [0,1]}$
of bounded, self-adjoint, Fredholm operators on a separable Hilbert
space corresponds (at least in a small interval $[t_1,t_2]$) to the
difference
$$
	\Sf\left[(D_t)_{t\in [t_1,t_2]}\right] = 
        \dim \mathds{1}_{[0,a]}\big(D_{t_2}\big)-\dim\mathds{1}_{[0,a]}\big(D_{t_1}\big),
$$ 
where $a>0$ is a given spectral level for which
$(\mathds{1}_{[-a,a]}(D_t))_{t\in [t_1,t_2]}$ is continuous. The
spectral flow of the whole path is obtained by subdividing $[0,1]$
into a finite family of intervals $[t_{i},t_{i+1}]$ for which we can
apply the above formula and adding all the contributions.  The concept
of spectral flow was then extended to unbounded, self-adjoint,
Fredholm operators in e.g. \cite{BBLP05}.  The restriction of the
spectral flow to loops constitutes a homotopy invariant in what is
called the gap topology.  It gives rise to a group homomorphism from
the fundamental group of the set of unbounded, self-adjoint, Fredholm
operators to $\Z$.

The paths of Dirac operators that we study are not continuous in the
gap topology, because their corresponding eigenfunctions may collapse
(by concentration of its mass on a set of Lebesgue measure $0$).  As
long as this only happens away from the spectral level zero it does
not cause a problem for defining the spectral flow.  Indeed, when
looking at the functional calculus construction of the spectral flow,
one realizes that it suffices to follow the evolution of the
eigenvalues close to zero.  In \cite{Wahl08}, the author introduces a
new topology, intermediate between the topology of strong-resolvent
convergence and the gap topology, for which the spectral flow is still
well defined and is a homotopy invariant.  We will see that this
topology is robust enough to allow some vanishing of eigenfunctions
and will allow us to study the spectral flow of our Dirac operators
introduced in \cite{dirac_s3_paper1}.  Note that other topologies on
the set of unbounded, self-adjoint, Fredholm operators can be
considered, but these topologies are stronger than the gap topology.
We refer the reader to \cites{BBLP05, Nicolaescu07} and the references
therein for their description and their relevance from the point of
view of $\mathrm{K}$-theory.

\subsection{Main results}
We consider a magnetic field, i.e., an exact $2$-form, $\bbet$ on
$\S^3$ which can be decomposed as
$\bbet=\bbet_{\mathrm{r}}+\bbet_{\mathrm{sing}}$ in a regular smooth part
$\bbet_{\mathrm{r}}$ and a singular part $\bbet_{\mathrm{sing}}$
corresponding to a smooth, oriented $K$-link
$$
	\gamma=\bigcup_{k=1}^K\gamma_k \subset \S^3
$$ 
with fluxes $2\pi\alpha_k$ on each $\gamma_k$, $k=1,\ldots,K$ (the
connected components of $\gamma$).  We need to make a gauge choice,
i.e., write $\bbet=d\bal$ to associate a Dirac operator $\cD_{\bal}$ to
such a magnetic field. Writing
$\bal=\bal_{\mathrm{r}}+\bal_{\mathrm{sing}}$ we found it convenient
to choose the singular gauge $\bal_{\mathrm{sing}}$ supported on the
union of Seifert surfaces for the $\gamma_k$. This has the advantage
that the Dirac operator depends periodically on the link fluxes
$\ua=(\alpha_1,\ldots,\alpha_K)$. This construction is explained in
Section~\ref{sec:def_dirac}.  Observe however that the choice of gauge
does not affect the spectra, and thus does not change the spectral
flow.

In Section~\ref{sec:cont_prop} we study the continuity of the Dirac
operator $\cD_{\bal}$ as a function of the link fluxes in the topology
of \cite{Wahl08}, which we call the Wahl topology.  By the
$2\pi$-periodicity, the link fluxes form a $K$-torus, and it turns out
that the family of Dirac operators fails to be continuous only for
certain points where the flux $2\pi\alpha_k$ of a knot approaches
$2\pi$.  In that case several eigenfunctions collapse as they
concentrate around the knot $\gamma_k$ and the discontinuity occurs
when this happens at the spectral level zero.
Theorem~\ref{thm:strg_res_cont} states continuity of the family in the
strong resolvent sense.  Theorem~\ref{thm:bump_cont_bulk} deals with
the Wahl continuity in the ``bulk" case $\ua\in (0,1)^K$, while
Theorem~\ref{thm:bump_cont_bdry_A} treats the Wahl continuity in the
``boundary" case when one or more $\alpha_k$ is $0\sim 1$.

As $\alpha_k \to 1^{-}$, fixing the remaining magnetic field
$\bbet^{(k)}=d\bal^{(k)}$, the limiting eigenvalues of the vanishing
eigenfunctions can be determined
(Proposition~\ref{prop:calcul_spectre}) with the help of an effective
operator which acts on sections of a canonical line bundle on
$\gamma_k$.  The limiting eigenvalues depend on the writhe of
$\gamma_k$ and the flux of the remaining magnetic field $\bbet^{(k)}$
through a Seifert surface $S_k$ for $\gamma_k$, i.e.,
\begin{equation}\label{eq:PHI}
		\Phi_{\bbet}(\gamma_k)=\int_{S_k}\bbet^{(k)}=\int_{\gamma_k}\bal^{(k)}
                =\int_{\gamma_k}\bal_{\mathrm{r}}+2\pi\sum_{j=1\atop j
                  \ne k}^K\alpha_j\link(\gamma_k,\gamma_j),
\end{equation}
where $\link(\gamma_k,\gamma_j)$ denotes the linking number of $\gamma_k$
and $\gamma_j$. The characterization of the limiting eigenvalues allows us to
determine the exact position of the points of discontinuity on the torus, i.e.,
when a limiting eigenvalue is zero.  When removing the set where
continuity fails, we obtain a ``cut" $K$-torus.

The spectral flow then defines a group homomorphism from the
fundamental group of the cut $K$-torus to $\Z$, which is studied in
Section~\ref{sec:spectral_flow}.  Using the fact that we have
identified the critical cuts and shown continuity away from them, we
are able to calculate the spectral flow for loops encircling one of
these critical cuts through a careful analysis close to it, see
Theorem~\ref{thm:spec_flow_gen_l}.

Theorem~\ref{thm:arbitrary_spectral_flow} states how the spectral flow
changes under variation of the fluxes and deformations of the
link. This implies a series of results, in particular, we compute in
Corollary~\ref{cor:sf_link_unknots} the spectral flow of the loop
obtained by tuning the flux carried by an \emph{unknot} in the
presence of a fixed additional magnetic field. The result is stated in the theorem below. We do not know the
generalization to a topologically non-trivial knot.

\begin{theorem}
 Let $\bbet$ be a magnetic field as above and assume that $\gamma_1$
 is a realization of an unknot.  Consider the corresponding closed
 loop of Dirac operators $(\cD_{\bal})_{\alpha_1\in[0,1]}$ obtained by
 tuning the flux $2\pi\alpha_1$ of the knot $\gamma_1$ from $0$ to
 $2\pi$. The spectral flow of the loop is defined only if
 $\frac{1}{2}(1-\Wr[\gamma_1])
 -(2\pi)^{-1}\Phi_{\bbet}(\gamma_1)\notin \Z$, in which case it is
\[
	\Sf((\cD_{\bal})_{\alpha_1\in[0,1]})=\Big\lfloor\frac{1}{2}(1-\Wr[\gamma_1])-\frac1{2\pi}\Phi_{\bbet}(\gamma_1)\Big\rfloor.
\]
\end{theorem}

Throughout the paper we only deal with the difficult part $\bbet_{\sing}\neq 0$, and fix for simplicity $\bbet_{\mathrm{r}}=0$,
but we emphasize that all the results hold for a non-zero smooth part, the proofs become just a little longer.

In the special case of a Hopf $2$-link, that is, two fibers of the Hopf map $\S^3\to \S^2$,
we can give a very detailed picture.
There are two critical points $p_1,p_2$ 
corresponding to the fluxes
$(\alpha_1,\alpha_2)=(\tfrac{1}{2},0)$ and $(\alpha_1,\alpha_2)=(0,\tfrac{1}{2})$,
see Figure~\ref{fig:punctured_torus_intro} below.
The segment $\alpha_1+\alpha_2=\tfrac{3}{2}$ corresponds to a family of fluxes 
for which the corresponding Dirac operator
has a one-dimensional kernel (see Corollary~\ref{cor:dim_ker_hopf_link}), 
and the elements of the kernel vanish in the above sense when the 
fluxes approach $(1^-,\tfrac{1}{2}^+)$ or $(\tfrac{1}{2}^+,1^-)$ along the line 
$\alpha_1+\alpha_2=\tfrac{3}{2}$.
In the general case we only know the location of the critical points, 
but not where the kernel is non-trivial.

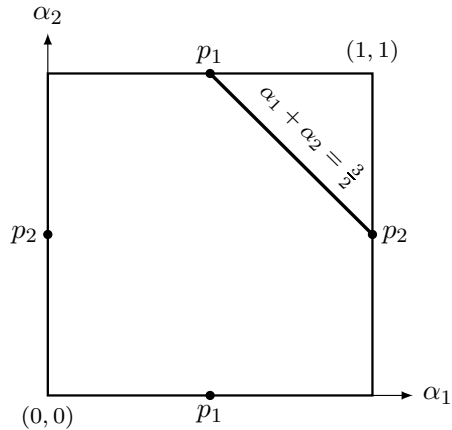
\begin{figure}[!!ht]
	\resizebox{0.5\textwidth}{!}{
	\begin{tikzpicture}
		\draw[thick] (0,0) rectangle (4,4);
		\node[below] at (0,0) {\scriptsize $(0,0)$};
		\node[above] at (4,4) {\scriptsize $(1,1)$};
		\draw[>=latex, ->] (4,0) -- (4.5,0);
		\node[thin, right, scale=0.9] at (4.5,0) {$\alpha_1$};
		\draw[>=latex, ->] (0,4) -- (0,4.5);
		\node[thin, above, scale=0.9] at (0,4.5) {$\alpha_2$};
		
		\node[below, scale = 0.9] at (2,0) {$p_1$};
		\node[above, scale = 0.9] at (2,4) {$p_1$};
		\node[left, scale = 0.9] at (0,2) {$p_2$};
		\node[right, scale = 0.9] at (4,2) {$p_2$};
		\draw [fill] (2,4) circle [radius = 0.05];
		\draw [fill] (4,2) circle [radius = 0.05];
		\draw [fill] (0,2) circle [radius = 0.05];
		\draw [fill] (2,0) circle [radius = 0.05];
		
		\draw[very thick] (2,4) -- (4,2)
		node[sloped, midway, above] {\scriptsize $\alpha_1+\alpha_2 = \tfrac{3}{2}$};
	\end{tikzpicture}
	}
	\caption{The punctured torus for the Hopf 2-link.}
	\label{fig:punctured_torus_intro}
\end{figure}

In \cite{dirac_s3_paper3}, we show that the Dirac
operators with magnetic links are limits (in the Wahl topology and
hence also in strong resolvent sense) of Dirac operators with
smooth magnetic fields. For the smooth approximation we do not have the
$2\pi$-periodicity as in the singular case. So for a single knot, the
spectral flow of the loop of Dirac operators corresponds, after
smoothing, to the spectral flow of an \textit{open} path of Dirac
operators. Nevertheless, the continuity in the Wahl topology allows us
to calculate the spectral flow for these open paths. In particular,
this leads to many new examples of smooth magnetic
fields for which Dirac operators have zero modes. Such examples play
an important role in the study of stability of matter in the presence
of magnetic fields. We refer the reader to 
\cites{Fefferman95,Fefferman96,FroLiebLoss86,LiebLossSol95,LiebLoss86,LossYau86} for more details.

\subsubsection*{\bf{Convention:}}
Throughout this paper, by a link $\gamma \subset \S^3$ we mean a \emph{smooth}, \emph{oriented} 
submanifold of $\S^3$ which is diffeomorphic to finitely many copies of $\S^1$. 
We will often write $\gamma = \cup_{k=1}^K \gamma_k$, where the $\gamma_k$ 
are (oriented) knots (the connected components of $\gamma$),
and we denote the $K$-tuple $(\gamma_1, \dots, \gamma_K)$ by $\ug$.

\subsubsection*{\bf{Acknowledgements:}}
We would like to thank M. Goffeng, G. Grubb and R. Nest for fruitful discussions and helpful comments.
We also thank the referee for various remarks and comments that helped us improving
the paper.

\section{Dirac operators with magnetic fields on links - a summary}\label{sec:def_dirac}
In this section we will recall the most important notions and results from \cite{dirac_s3_paper1},
where the definition of the Dirac operator with magnetic fields on links has been outlined in great detail.
All theorems in this section are given without proof, the interested reader will find them in
\cite{dirac_s3_paper1}. The singular magnetic fields are more easily understood when viewed
as $1$-currents, that is after we used the Hodge duality and the musical isomorphisms to transform
$2$-forms into vectors. Their magnetic potentials become then $2$-currents, 
through the operation which transforms $1$-forms into bivectors.

\smallskip

We see $\S^3$ as the set of unit vectors in $\C^2$. We endow it with its induced metric
and its Levi-Civita connection $\nabla$; its metric tensor is denoted by $g_3$.

\subsection{Magnetic links}
A \textit{magnetic link} $\bB$ is a link $\gamma = \cup_k \gamma_k \subset \S^3$ together with fluxes 
$2\pi\alpha_k$, $0 \leq \alpha_k < 1$.
$\bB$, which we view as a one-current, acts on smooth one-forms $\omega \in \Omega^1(\S^3)$ as
$$
	(\bB;\omega) :=\sum_k  2\pi\alpha_k \int_{\gamma_k}\omega,
$$
so that $(\partial \bB;\phi) = (\bB;\d\phi) = 0$.

To construct a magnetic gauge potential $\bA$ for $\bB$, we use Seifert's theorem \cites{Seifert35, FP30}, 
which states that for any (oriented) knot $\gamma_k \subset \S^3$
there exists a smooth, connected and oriented surface $S_k \subset \S^3$ such that 
$\partial S_k = \gamma_k$. 
The magnetic gauge potential $\bA$ is thus given by
$$
	\bA = 2\pi \sum_{k} \alpha_k [S_k].
$$
For any $\omega \in \Omega^1(\S^3)$, $\bA$ satisfies
$(\partial \bA;\omega) = (\bB; \omega)$ by Stokes' formula.
If $\bA'$ is another magnetic gauge potential for $\bB$, then $\bA$ and $\bA'$ 
should be related by a boundary term. 
We restrict ourselves to one component $\gamma_k$ of $\gamma$;
assuming that $S_k$ and $S_k'$ only intersect on $\gamma_k$,
then we have:
\begin{equation}\label{eq:a_gauge_transf}
	2\pi\alpha_k[S_k] - 2\pi\alpha_k[S_k'] = 2\pi\alpha_k\, \partial [V_k],
\end{equation}
where $[V_k]$ is the volume enclosed by $S_k$ and $S_k'$, seen as a $3$-current.

\subsection{$Spin^c$ spinor bundles on $\S^3$}
In order to be able to give the definition of the Dirac operator, it is necessary to introduce the 
underlying spinor structure on $\S^3$.

A $Spin^c$ spinor bundle $\Psi$ over a Riemannian 3-manifold $\cM$ is a two-dimensional complex vector bundle over 
$\cM$ together with an inner product $\cip{\cdot\,}{\cdot}$ and an isometry $\bsigma: \rT^*\cM \to \Psi^{(2)}$, 
called the Clifford map, where
$$
	\Psi^{(2)} := \{M \in \mathrm{End}(\Psi): M = M^*, \tr(M)=0\},
$$
where the inner product on $\Psi^{(2)}$ is given by $(A,B):=\tfrac{1}{2}\mathrm{Tr}\big(AB\big)$. 
A connection $\nabla^{\Psi}$ on $\Psi$ is a $Spin^c$ connection if for any vector field $X \in \Gamma(\rT\,\cM)$ it satisfies
\begin{enumerate}
	\item $X\cip{\xi}{\eta} = \cip{\nabla_X^{\Psi}\xi}{\eta} + \cip{\xi}{\nabla_X^{\Psi}\eta}$, 
	$\forall \xi,\eta \in \Gamma(\Psi)$.
	
	\item $[\nabla_X^{\Psi},\bsigma(\omega)] = \bsigma(\nabla_X \omega)$, $\forall \omega \in \Omega^1(\cM)$.
\end{enumerate}
Here, $\nabla_X$ is the Levi-Civita connection on $\cM$.

Consider a local frame $(\be_1,\be_2,\be_3)$ of $\rT^*\cM$.
Let $(\lambda_+,\lambda_-)$ be the local trivialization of $\Psi$
given by eigen-sections $\lambda_{\pm}\in\ker(\sigma(\be_3)\mp 1)$.
We require that $(\lambda_+,\lambda_-)$ both have unit length and a fixed relative phase
according to 
$$
	\omega(\be_1) + i\omega(\be_2) = \cip{\lambda_-}{\sigma(\omega)\lambda_+}, 
	\quad \forall \omega \in \Omega^1(\cM).
$$
The connection form of the $Spin^c$ connection $\nabla^{\Psi}$ in the trivialization $(\lambda_+,\lambda_-)$ is
\begin{align}\label{eq:con_form}
	M_{\lambda}(X) := \begin{pmatrix} \cip{\lambda_+}{\nabla_X^{\Psi} \lambda_+} & \cip{\lambda_+}{\nabla_X^{\Psi} \lambda_-} \\
	\cip{\lambda_-}{\nabla_X^{\Psi} \lambda_+} & \cip{\lambda_-}{\nabla_X^{\Psi} \lambda_-}\end{pmatrix},
\end{align}
where $X$ is any vector field on $\cM$.
We then have (as in \cite{MR1860416}*{Prop.~2.9})
\begin{multline}\label{eq:spin_basis}
	M_{\lambda}(X)\\
	= \frac{i}{2}\begin{pmatrix} \cip{\be_1}{\nabla_X \be_2} & -\cip{\be_3}{\nabla_X \be_2} - i\cip{\be_3}{\nabla_X \be_1} \\
	-\cip{\be_3}{\nabla_X \be_2} + i\cip{\be_3}{\nabla_X \be_1} & -\cip{\be_1}{\nabla_X \be_2}\end{pmatrix}\\
	- i\omega_{\lambda}(X)\mathrm{Id}_{\C^2},
\end{multline}
where $\omega_{\lambda}$ is the (local) real one form
\begin{equation}\label{def_omega_lambda}
	\omega_{\lambda}(X) := \frac{i}{2}\left(\cip{\lambda_+}{\nabla_X\lambda_+} 
	+ \cip{\lambda_-}{\nabla_X\lambda_-}\right).
\end{equation}

It is well known that, up to isomorphism, the trivial bundle
$\Psi = \S^3 \times \C^2$ is the unique $Spin^c$ spinor bundle on $\S^3$, 
and that we can choose a global orthonormal frame of $\rT^*\S^3$ ($\S^3$ is a Lie group).
By \cite{MR1860416}*{Prop.~2.11}, the connection form of the induced connection from the Levi-Civita connection on 
$\S^3$ is given by \eqref{eq:spin_basis} with $\omega_{\lambda}=0$.

\begin{rem}
	From now on, whenever $\nabla$ acts on a vector field we mean the Levi-Civita connection,
	whereas if $\nabla$ acts on a spinor, then we mean the trivial (induced) connection on the $Spin^c$ bundle on $\S^3$ 
	(corresponding to $\omega_\lambda=0$).
\end{rem}

\subsection{Dirac operators}
Having defined the spinor structure on $\S^3$, we need to discuss Seifert surfaces in greater detail, as they will 
be one of the main objects through which we define our Dirac operator. 
\begin{rem}
Let us emphasize that the definition given below (Theorem~\ref{thm:def_D_-}) corresponds 
to the convention $\sigma(-i\nabla+\boldsymbol{\alpha})$ for a smooth magnetic potential $\boldsymbol{\alpha}\in \Omega^1(\S^3)$, and \emph{not}
$\sigma(-i\nabla-\boldsymbol{\alpha})$.
\end{rem}

\subsubsection{Seifert frames and local coordinates}\label{sec:seif_fram_loc_coord}
If $S \subset \S^3$ is a Seifert surface for a knot $\gamma$, then the Seifert frame is given by the triple 
$(\bT, \bS, \bN)$, where $\bT$ is the tangent vector to $\partial S = \gamma$, $\bN$ is the normal of the oriented Seifert surface
and $\bS$ is the unit vector such that $(\bT, \bS, \bN)$ is positively oriented (thus along the knot it points towards the
inside of the Seifert surface). The Seifert frame coincides
with the Darboux frame (see \cite{Spivakvol4}*{Chapter~7})
and thus satisfies
\begin{equation}\label{eq:eq_darboux_frame}
	\nabla_{\bT}\begin{pmatrix} \bT\\ \bS\\ \bN\end{pmatrix}
	=\begin{pmatrix}0 & \kappa_g & \kappa_n \\ -\kappa_g & 0 & \tau_r\\ -\kappa_n & -\tau_r & 0 \end{pmatrix}
	\begin{pmatrix} \bT\\ \bS \\ \bN\end{pmatrix},
\end{equation}
where $\kappa_g, \kappa_n$ are the geodesic and normal curvatures and $\tau_r$ is the relative torsion.

\begin{rem}
	Throughout this paper we will drop the subscript $r$ for simplicity; for a knot $\gamma$ with Seifert surface $S$, 
	we will write $\tau_S$ for the corresponding relative torsion.
\end{rem}

We extend the Seifert frame -- which is \emph{a priori} only defined on $\gamma$ --
to a tubular neighborhood of the curve: ($\dist_{g_3}$ denotes the geodesic distance)
$$
	B_{\eps}[\gamma]:=\{\bp \in \S^3: \dist_{g_3}(\bp,\gamma)<\eps\}.
$$
The parameter $\eps>0$ is chosen such that the map
\begin{equation*}
	\exp: 
	\begin{array}{ccl}
		B_{\eps}[\sigma_0;\cN_{\gamma}] &\longrightarrow& B_{\eps}[\gamma]\\
		(\gamma(s),\bv(s)) &\mapsto& \exp_{\gamma(s)}(\bv(s))
	\end{array}
\end{equation*}
is a diffeomorphism, where $\cN_{\gamma}$ is the normal bundle to $\gamma$ in $\S^3$ and $\sigma_0$ its null-section.
We define the frame at a point $\bp = \exp_{\gamma(s)}(t_0\bv_0(s))$, $\norm{\bv_0}=1$, 
by parallel transporting $(\bT,\bS,\bN)$ along the 
geodesic\footnote{The geodesics on $\S^3$ are the great circles.}
\begin{align}\label{eq:geodesic_s3}
	t \in [0,t_0] \mapsto \exp_{\gamma(s)}(t\bv_0(s)) = \cos(t)\gamma(s) + \sin(t)\bv_0(s)\in\S^3\subset \C^2.
\end{align}
The geodesic distance $t_0$ is henceforth denoted by $\rho$.

To describe the self-adjoint extensions of the Dirac operators under study, it is necessary to describe the
behavior of spinors in their domains in the vicinity of the singular magnetic field. For that we need coordinates in such
a neighborhood. Two coordinates are intrinsic to the curve $\gamma$, namely the geodesic distance $\rho$ to it
and the arc length parameter $s$ of the projection. To complete the picture, we need to define a suitable angle function $\theta$ on 
$B_{\eps}[\gamma] \setminus \gamma$, something which can be done with the help of the 
Seifert frame; the angle $\theta$ is defined through
\begin{align}\label{eq:drho}
	\d \rho = \cos(\theta) \bS^{\flat} + \sin(\theta) \bN^{\flat},
\end{align}
giving us a set of orthogonal coordinates $(s,\rho,\theta)$ on $B_{\eps}[\gamma]$.
We recall that $\flat,\sharp$ denote the musical isomorphisms which transform vectors into
1-forms resp. 1-forms into vectors through the metric.
We thus obtain a map
\[
	\mathrm{exp}:(s,\theta,\rho)\mapsto \exp_{\gamma(s)}\big\{\rho(\cos(\theta)\bS+\sin(\theta)\bN) \big\}.
\]
Setting 
\begin{equation}\label{eq:def_h}
	h(\bp) := \cos(\rho) - \sin(\rho)(\kappa_g(s)\cos(\theta) +\kappa_n(s)\sin(\theta)),
\end{equation} 
we get that
\begin{align}\label{def:bG}
	(\bT,(\d\rho)^{\sharp},\bG) :&=(\bT,\exp_{*}(\partial_{\rho}), \sin(\rho)^{-1}\exp_{*}(\partial_{\theta}))\\
	&= (h^{-1}(\exp_{*}(\partial_{s}) - \tau_r \exp_{*}(\partial_{\theta})),\exp_{*}(\partial_{\rho}), \sin(\rho)^{-1}\exp_{*}(\partial_{\theta}))\nonumber
\end{align}
is an orthonormal basis of $\rT\, (B_{\eps}[\gamma] \setminus \gamma)$.
The pullback of the volume form is given by
\begin{equation}\label{eq:pullb_volform}
	\exp^{*}(\vol_{g_3}) = \exp^{*}(\bT^{\flat} \wedge \d\rho \wedge \bG^{\flat}) = h \sin(\rho)\,\d s \wedge \d\rho \wedge \d\theta.
\end{equation}

\subsubsection{Space of knots and Seifert surfaces}
Let $\sK$ be the set of smooth knots in $\S^3$
and $\sS$ the set of Seifert 
surfaces having boundary in $\sK$. 
Furthermore, in $\prod_{k=1}^K \sS$ we define
\begin{multline*}
	\sS^{(K)}\!:=\!\{(S_1,\dots,S_K)\!\in\!\textstyle{\prod_{k=1}^K \sS}\!:
	\quad \cup_{k=1}^K\partial S_k\mbox{ is a $K$-link }\&\\
	\partial S_i\textrm{ and } S_j \textrm{ are transverse in } \S^3, \, i \neq j\}.
\end{multline*}
Two smooth submanifolds $M, N$ of $\S^3$ are said to be transverse if
$$
	\rT_{\bp}M + \rT_{\bp}N = \rT_{\bp} \S^3, \quad \forall \bp \in M \cap N.
$$
Given any $K$-link $\gamma = \cup_{k=1}^K\gamma_k$
we can always find a $K$-tuple of Seifert surfaces 
$\uS \in \sS^{(K)}$ with
$\partial S_k = \gamma_k$
by using a transversality argument, see \cite{dirac_s3_paper1}*{Section~3.3.1}.

\subsubsection{Dirac operators}\label{sec:dirac_op_def}
We are now in the position to define the Dirac operator for a
link\footnote{Up to fixing a base point $\bp_0$, we identify $\gamma$ with its arc length parametrization
$\gamma: \R/\ell\Z \to \S^3$ with $\gamma(0) = \bp_0$.}
$\gamma = \bigcup_k \gamma_k$
together with a $K$-tuple $2\pi$-normalized fluxes $\ua \in \Tf^{K}$, where
$$
	\Tf = \{[0,1]: 0 \sim 1\}.
$$
We choose a $(\uS,\ua) \in \sS^{(K)}\times \Tf^K$ such that $\partial S_k = \gamma_k$
and set
$$
	\bA = \sum_k 2\pi\alpha_k\, [S_k], \quad 0 \leq \alpha_k <1.
$$
Setting 
$$
	\Omega_{\underline{S}} := \S^3\setminus \big(\cup_k S_k \big),
	\quad \wt{S}_k := S_k \cap \big(\cap_{k \neq k'}\complement_{\S^3}S_{k'}\big),
$$
the minimal operator is defined by
\begin{equation}\label{eq:def_dom_min_dirac_op}
\left\{
	\begin{array}{rcl}
	 \dom\big(\cD_{\bA}^{(\min)}\big) 
	&:=& \left\{\psi \in H^1(\Omega_{\uS})^2: \right.\\
	&&\quad\quad
	\left.\left.\psi\right|_{(\wt{S}_k)_+}=e^{-2i\pi\alpha_k}\left.\psi\right|_{(\wt{S}_k)_-}\in H^{1/2}(\wt{S}_k)^2, \forall k\right\},\\
	\cD_{\bA}^{(\min)}\psi&:=&\bsigma(-i\nabla)\psi|_{\Omega_{\uS}}\in L^2(\S^3)^2.
	\end{array}
\right.
\end{equation}

As usual $\cD_{\bA}^{(\max)}$ denotes $(\cD_{\bA}^{(\min)})^*$. In order to characterize 
the self-adjoint extensions of the operator $\cD_{\bA}^{(\min)}$ in the singular gauge $\bA$, 
we need to specify the behavior of the spinors in the neighborhood of each knot $\gamma_k$.
We define two smooth spinors $\xi_{\pm}$ through
\begin{align}\label{eq:def_xi}
	\bsigma(\bT^{\flat})\xi_{+} = \xi_{+}, \quad \bsigma(\bT^{\flat})\xi_{-} = -\xi_{-}.
\end{align}
\emph{A priori} only defined on the link $\gamma$, these spinors can be extended in a tubular neighborhood of the curve $\gamma$
via parallel transport along geodesics.
Furthermore, we fix their relative phase by requiring that 
\begin{align}\label{eq:rel_phase}
	\omega(\bS) + i \omega(\bN) = \langle \xi_-, \bsigma(\omega)\xi_+\rangle, \quad \forall \omega \in \Omega^1(B_{\eps}[\gamma]).
\end{align}

Let $\chi: \R \to \R_+$ be a smooth function with $\supp \chi \in [-1,1]$ and $\chi(x) = 1$ for $x \in [-2^{-1},2^{-1}]$. 
We then define the localization function at level $\delta>0$ (for $\delta$ small enough, see below) by
\begin{equation}\label{eq:chi_loc_curve}
	\chi_{\delta,\gamma}: 
	\begin{array}{ccl}
		B_{\delta}[\gamma] &\longrightarrow& \R_+\\
		\bp = \exp_{\gamma(s)}(\rho_{\gamma}\bv_0(s))
		&\mapsto& \chi\left(\rho_{\gamma}\delta^{-1}\right).
	\end{array}
\end{equation}
Let $\eps>0$ be small enough such that the tubular neighborhoods around each
$\gamma_k$ are mutually disjoint,
$$
	B_{\eps}[\gamma] = \bigcup_{k=1}^K B_{\eps}[\gamma_k].
$$
We recall that for all $k$, $\exp$ defines a diffeomorphism of $B_\eps[\sigma_0]\subset \cN_{\gamma_k}$ onto $B_\eps[\gamma_k]$.
We then choose $\delta>0$ as follows:
$$
	0 < \delta < \eps \min\left\{1, \left(\sup_{k,n}\norm{\kappa_k^{(n)}}_{L^{\infty}} 
	+ \sqrt{\eps + \sup_{k,n}\norm{\kappa_k^{(n)}}_{L^{\infty}}}\right)^{-1}\right\}.
$$

\begin{theorem}\label{thm:def_D_-}[Self-adjointness]
	Fix $K \in \N$ and let $\gamma = \bigcup_{k=1}^K \gamma_k \subset \S^3$ be a link.
	Pick $(\uS,\ua) \in \sS^{(K)} \times \Tf^K$ with $\partial S_k = \gamma_k$ and set
	$$
		\bA = \sum_{k=1}^K2\pi\alpha_k[S_k].
	$$
	We define $\cD_{\bA}$ by
	\begin{equation*}
		\left\{
		\begin{array}{lcl}
		\dom\big(\cD_{\bA}\big) 
		&:=&\big\{\psi \in \dom\big( (\cD_{\bA}^{(\min)})^{*}\big):\\
		&&\quad \langle \xi_{+},\chi_{\delta,\gamma_k}\psi \rangle \xi_{+} 
		\in \dom\big( \cD_{\bA}^{(\min)}\big), 1 \leq k \leq K\big\},\\
		\cD_{\bA}\psi &:=& -i\bsigma(\nabla)\psi|_{\Omega_{\uS}} \in  L^2(\S^3)^2.
		\end{array}
		\right.
	\end{equation*}
	The definition of $\cD_{\bA}$ is independent of the choice of $\chi_{\delta,\gamma}$ and the operator
	is self-adjoint.
\end{theorem}

\begin{rem}\label{rem:d_ext_choice}
	The operator $\cD_{\bA}$ of the above theorem is written $\cD_{\bA}^{(-)}$ in \cite{dirac_s3_paper1}, where
	the exponent $(-)$ refers to the orientation of the elements of the domain on the link $\gamma$.
	\emph{A priori}, we can define other self-adjoint extensions by requiring that the elements of the domain
	satisfy instead
	\[
	\langle \xi_{-e_k},\chi_{\delta,\gamma_k}\psi \rangle \xi_{-e_k} 
		\in \dom\big( \cD_{\bA}^{(\min)}\big),
	\]
	for a given $\underline{e}\in\{+,-\}^K$. 
	As shown in \cite{dirac_s3_paper1}*{Remark 16}, each of these operators coincides
	with some $\cD_{\bA'}$ for another singular gauge $\bA'$ (depending on $\underline{e}$).
\end{rem}
\begin{center}
\textbf{From now on, we remove the superscript $(-)$ on $\cD_{\bA}^{(-)}$ for short.}
\end{center}

\section{Continuity of Dirac operators}
Fix $K \in \N$ and consider the map
\begin{equation}\label{eq:phi_no_top}
\bbd:\begin{array}{ccc}
		\sS^{(K)}\times \Tf^{K}&\to& \Sd(L^2(\S^3)^2)\\
		\big(\underline{S},\underline{\alpha} \big)&\mapsto& \cD_{\bA},
	\end{array}
\end{equation}
where $\Sd(L^2(\S^3)^2)$ denotes the set of self-adjoint operators on $L^2(\S^3)^2$ 
with discrete spectrum (that is: the spectrum of $D\in \Sd(L^2(\S^3)^2)$ 
is a discrete set in $\R$ and consists of eigenvalues with \emph{finite} multiplicities).

We now endow each space with a particular topology and
study the continuity of the resulting map, in order to define the spectral flow
of paths in this space.

Let us emphasize that, up to gauge transformations, the Dirac operators depend
only on the link and the fluxes. In particular the spectrum does not depend on the choice of 
the Seifert surfaces. But as we investigate the continuity properties of the family $\big(\cD_{\bA}\big)_{\uS,\ua}$,
we need to take the gauge into account.

\paragraph{\emph{Gauge transformations}}
Recall \eqref{eq:a_gauge_transf}: given $\gamma\in \sK$
with Seifert surfaces $S,S'\in \sS$ satisfying $S\cap S'=\gamma$, the two gauges $2\pi\alpha [S]$ and $2\pi\alpha[S']$
are related by the gauge transformation $e^{-2i\pi\alpha \mathds{1}_V}$, where $V$
is the volume enclosed by $S$ and $S'$, $\partial [V]=[S]-[S']``=-\nabla \mathds{1}_V"$ (the second homology group $H^2(\S^3)$ is trivial). We have 
\[
	\dom(\cD_{2\pi\alpha [S]})=e^{2i\pi\alpha \mathds{1}_V}\dom(\cD_{2\pi\alpha [S']})\ \&\ \cD_{2\pi\alpha [S]}=	e^{2i\pi\alpha \mathds{1}_V}\cD_{2\pi\alpha [S']}e^{-2i\pi\alpha \mathds{1}_V}.
\]
If $S\cap S'\neq\gamma$, there exists \cite{MR916076} a finite sequence of Seifert surfaces $S_0=S,S_1,\cdots,S_N=S'$
such that $S_j\cap S_{j+1}=\gamma$, the gauge transformation is then defined stepwise. This procedure carries over to magnetic links.

\subsection{Topology on $\sS^{(K)}\times \Tf^{K}$}
In order to establish a notion of convergence of Seifert surfaces, we need to introduce 
a means to compare the distance between oriented, compact sub-manifolds of $\C^2$.

Let $\cM_1,\cM_2$ be two real $d$-dimensional, oriented, compact, smooth sub-manifolds of $\C^2$.
Denote by $N_1,N_2$
their respective Gauss maps which take values in the oriented Grassmannian 
$\widetilde{\mathrm{Gr}}_{\R}(4-d,\C^2)\subset\bigwedge_{\R}^{4-d}\C^2=:\mathbb{E}_{\wedge}$.
By induction over $k\ge 0$, we define the extended $k$-derivatives of $N_i$ as smooth maps $L_kN_i:\cM_i\to \mathcal{L}_k(\C^2,\mathbb{E}_{\wedge})$.
Let $L_0N_i\equiv N_i$ and $L_1N_i(\bp)$ be the canonical extension of the differential $\d N_i$ to the whole space $\C^2$:
\begin{align*}
	L_1N_i(\bp): 
	\begin{array}{ccccl}
		\rT_{\bp}\mathcal{M}_i&\oplus&(\rT_{\bp}\mathcal{M}_i)^{\perp} &\longrightarrow& \textstyle{\bigwedge_{\R}^{4-d}\C^2}=\mathbb{E}_{\wedge}\\
		\bv&+&\bv_{\perp}&\longmapsto& \d N_i(\bp)\bv.
	\end{array}
\end{align*}
Hence $L_1N_i:\mathcal{M}_i\to \mathcal{L}(\C^2,\mathbb{E}_{\wedge})$ is smooth. Having defined $L_{k-1}N_i$,
we define $L_kN_i(\bp)$ as the canonical extension of $\d L_{k-1}N_i(\bp):\rT_{\bp}\mathcal{M}_i\to \mathcal{L}_{k-1}(\C^2,\mathbb{E}_{\wedge})$,
$L_kN_i(\bp)$ is then canonically identified with an element of $\mathcal{L}_k(\C^2,\mathbb{E}_{\wedge})$. 
For a $k$-linear map $M$ in this latter space, we consider the norm:
$$
	\norm{M}_k := \sup_{(\bv_1,\dots,\bv_k), \norm{\bv_j}_{\C^2}^2=1}\left\lVert M[\bv_1,\dots,\bv_k]\right\rVert_{\mathbb{E}_{\wedge}}.
$$
The distance between two points $\bp_1\in \cM_1, \bp_2 \in \cM_2$ is then defined as
$$
	\delta_{\cM_1\times\cM_2}(\bp_1,\bp_2)
	:= |\bp_1-\bp_2| + \sum_{k=0}^\infty 2^{-k}\min\left\{\norm{L_kN_1(\bp_1) - L_kN_2(\bp_2)}_k,1\right\}.
$$
and we set
\begin{align}\label{def:dist_subm}
	&\dist_d(\cM_1,\cM_2) := |\cH_d(\cM_1)-\cH_d(\cM_2)|\\
	&\quad +\max\left(\sup_{\bp_1 \in \cM_1}\inf_{\bp_2 \in \cM_2}\delta_{\cM_1\times\cM_2}(\bp_1,\bp_2),
	\sup_{\bp_2 \in \cM_2}\inf_{\bp_1 \in \cM_1}\delta_{\cM_1\times\cM_2}(\bp_1,\bp_2)\right),\nn
\end{align}
where $\cH_d$ is the $d$-dimensional Hausdorff measure.

The distance to compare different Seifert surfaces is now given by
\begin{align}\label{def:dist_seifert}
	\dist_{\sS}(S_1,S_2) := \dist_2(S_1,S_2) + \dist_1(\partial S_1, \partial S_2).
\end{align}
One important feature of this metric is that the convergence of one Seifert surface to another in $\dist_{\sS}$ ensures
the convergence of the coordinates on their respective boundary curves.

\begin{proposition}[Convergence of coordinates]\label{prop:conv_coord}
	Let $S^{(n)}$ be a sequence of Seifert surfaces converging to $S$. Then
	the geodesic coordinates $(s_n,\rho_n,\theta_n)$ for $S^{(n)}$ converge to the geodesic
	coordinates on $S$ in the $C^\infty$-norm, with $\tfrac{\ell}{\ell_n}s_n \to s$ on $\overline{B}_{\eps}[\partial S]$,
	 $\rho_n \to \rho$ and $\theta_n \to \theta$ on 
	$\overline{B}_{\eps}[\partial S] \setminus B_{\eps'}[\partial S]$ for $0<\eps'<\eps$.
	We also have $\rho_n\to \rho$ in the $C^0$-norm on $\overline{B}_{\eps}[\partial S]$.
\end{proposition}

The above definition then allows us to define a natural 
metric on the cartesian product $\Pi_{k=1}^K\sS\times \Tf^K$; for 
$$
	(\uS,\ua), (\uS',\ua') \in \Pi_{k=1}^K\sS\times \Tf^K
$$ 
we set
$$
	\dist((\uS,\ua),(\uS',\ua'))
	:= \max_{1\leq k\leq K}\left[\dist_{\sS}(S_k,S_k') + \dist_{\Tf}(\alpha_k,\alpha_k')\right].
$$
Observe that $\sS^{(K)}$ is open in $\Pi_{k=1}^K\sS$.
The above metric then induces a topology on $\sS^{(K)} \times \Tf^K$, 
which we denote by $\sT_{\sS}^{(K)}$.

We can now restate an important result \cite[Theorem~22]{dirac_s3_paper1}, 
which concerns the convergence of Dirac operators and will be useful
throughout the rest of the paper. For a (singular) magnetic potential, $\norm{\cdot}_{\bA}$
denotes the graph norm of $\cD_{\bA}$. 

\smallskip

In the theorem below and elsewhere the arrow $\rightharpoonup$
denotes the \emph{weak} convergence in the corresponding Banach space.

\begin{theorem}[Compactness]\label{thm:compactness}
	Fix $K \in \N$. Given
	$(\uS^{(n)},\ua^{(n)})$ and $(\uS,\ua)$ in $\sS^{(K)} \times \Tf^K$,
	such that $(\uS^{(n)},\ua^{(n)}) \to (\uS,\ua)$ in the $\dist$-metric, set
	$$
		\bA^{(n)} = \sum_{k=1}^K 2\pi \alpha_k^{(n)} [S_k^{(n)}],
		\quad \bA = \sum_{k=1}^K 2\pi \alpha_k [S_k].
	$$
	Furthermore, let $(\psi^{(n)})_{n \in \N}$ be a sequence in $L^2(\S^3)^2$ such that:
	\begin{enumerate}
		\item $\psi^{(n)} \in \dom\big(\cD_{\bA^{(n)}}\big)$, $n \in \N$.
		
		\item $\big(\norm{\psi^{(n)}}_{\bA^{(n)}}\big)_{n \in \N}$ is uniformly bounded.
	\end{enumerate}
	Then, up to extraction of a subsequence, $\psi^{(n)} \rightharpoonup \psi \in \dom\big(\cD_{\bA}\big)$
	and
	$$
		\big(\psi^{(n)}, \cD_{\bA^{(n)}}\psi^{(n)}\big)
		\rightharpoonup \big(\psi,\cD_{\bA} \psi\big)
		\textrm{ in } L^2(\S^3)^2 \times L^2(\S^3)^2
	$$
	with
	$$
		\int \big|\cD_{\bA} \psi \big|^2
		\leq \liminf_{n\to\infty} \int \big|\cD_{\bA^{(n)}} \psi^{(n)} \big|^2.
	$$
	In fact, one has strong convergence for $(\psi^{(n)})_n$ 
	except in the cases when $\alpha_k^{(n)} \to 1^{-}$,
	where the loss of $L^2$-mass of $(\psi^{(n)})_n$ can occur only 
	through concentration onto such a knot
	$\gamma_k = \partial S_k$.
\end{theorem}

\subsection{Topology on $\Sd$}\label{sec:sdisc_top}
The usual topology on $\Sd(L^2(\S^3)^2)$, the gap topology, turns out to be too restrictive
for our situation. We recall that the gap topology is the metric induced 
by the embedding of the graphs $\cG_{\cD}$ (viewed as closed vector spaces in $L^2(\S^3)^2\times L^2(\S^3)^2$)
in the set $(\cB(L^2(\S^3)^2\times L^2(\S^3)^2),\norm{\,\cdot\,}_{\cB})$
through the orthogonal projection $\Pi_{\cG_{\cD}}$ onto $\cG_{\cD}$:
$$
	\dist_{\cG_{\cD}}(\cD_1,\cD_2):=\norm{\Pi_{\cG_{D_1}}-\Pi_{\cG_{D_2}}}_{\cB}.
$$
We will instead use topologies introduced in \cite{Wahl08}.

By a \textit{bump function centered at $\lambda\in\R$} we mean a function $\phi\in \sD(\R;\R_+)$ 
such that its translation $\phi(x+\lambda)$ is even with compact support 
$[-a,a]$ and $\phi'(x+\lambda)>0$ for $x\in (-a,0)$. By convention, if 
the center $\lambda$ is not given explicitly, then it is $0$.

For a bump function $\phi\in\sD(\R;\R_+)$ centered around $0$,
we define $\sT_{\phi} \subset 2^{\Sd}$
as the weakest topology for which the maps
\begin{align}\label{eq:def_t_phi}
	\Sd(L^2(\S^3)^2) &\to L^2(\S^3)^2, \cD \mapsto (\cD+i)^{-1}\psi\nn\\
	\Sd(L^2(\S^3)^2) &\to L^2(\S^3)^2, \cD \mapsto (\cD-i)^{-1}\psi\\
	\Sd(L^2(\S^3)^2) &\to \cB(L^2(\S^3)^2), \cD \mapsto \phi(\cD)\nn
\end{align}
are continuous for all $\psi \in L^2(\S^3)^2$. 

If $\wt{\phi}$ is another bump function with $\supp\wt{\phi}\subset (\supp\phi)^{\circ}=(-a,a)$,
then we have the inclusion $\sT_{\wt{\phi}}\subset \sT_{\phi}$. This comes from the factorization
$\wt{\phi}(x)=f\circ \phi(x)$ where $f$ is a continuous function. More precisely, we have
$f(y):=\wt{\phi}(x_\phi(y))$ where $x_\phi$ is the inverse function of $\phi$ that satisfies $x_\phi(\phi(x))=x$ for $x\in[-a,0]$. 
The inclusion follows from functional calculus.

Consider a bump function $\phi_1$ with $\supp\,\phi_1=[-1,1]$, and the family $\phi_n(x):=\phi_1(nx)$, $n\ge 1$,
$(\sT_{\phi_n})_n$ is a decreasing sequence of subsets of $2^{\Sd}$.

For our purpose we define the notion of bump-continuity.

\begin{definition}[Bump-continuity]\label{def:bump_continuity}
	Let $\cL$ be a topological space and let $c: \cL \to \Sd$
	be a map which is continuous in the strong resolvent sense.
	
	For $x\in\cL$, we say that $c$ is bump-continuous at $x$ if $c$ is locally $\sT_{\phi}$-continuous
	at $x$ for some bump function $\phi$, that is if there exist
	an open neighborhood $U\subset \cL$ containing $x$ and such that $c_{|_{U}}$ is $\sT_{\phi}$-continuous.
	
	We say that $c$ is bump-continuous if it is bump-continuous at all points in $\cL$.
\end{definition}

\begin{rem}
 Using the the Heine-Borel property, we can choose the bump function uniformly in $x$ when $\cL$ is compact. 
\end{rem}

\begin{rem}\label{rem:Wahl_top}
 In \cite{Wahl08}, the author defines a topology $\sT_W$ on the set $\mathrm{SF}(\cH)$ of Fredholm self-adjoint operators on a Hilbert space $\cH$. 
 We refer the reader to \cite{Wahl08} for the construction of $\sT_W$ (it uses the inductive limit). 
 
 Let us simply emphasize that for all bump functions $\phi$, the injection map $(\Sd,\sT_{\phi})\to (\mathrm{SF}(\cH),\sT_{W})$ is continuous.
 We can also define the $\sT_{\phi}$-topology directly on $\mathrm{SF}(\cH)$ and $(\Sd,\sT_{\phi})$ is then its induced topology.
 
 For short, we will write $(\mathrm{SF},\sT_{W})$ instead of $(\mathrm{SF}(L^2(\S^3)^2),\sT_{W})$.
\end{rem}

\subsection{Effective operators on the knots}\label{sec:effective_operator}
In order to study the continuity of \eqref{eq:phi_no_top}, we
will make use of some induced operators on the knots $\gamma_k$.

Fix $(\uS,\ua)\in \sS^{(K)}\times \Tf^{K}$.
Recall that in the vicinity of a knot $\gamma_k=\partial S_k$,
the other Seifert surfaces may cut the tubular neighborhood $B_{\delta}[\gamma_k]$.
For all $k'\neq k$, either 
$S_{k'}$ does not intersect $\gamma_k$,
or the intersection is reduced to one point,
or there are $M_{k,k'}$ points of intersection and $B_{\delta}[\gamma_k]\cap \complement S_{k'}$ is split
into $M_{k,k'}$ sections $R_{m}^{(k')}$, $0\le m\le M_{k,k'}-1$. 
Going along $\gamma_k$, we then pass through the other Seifert surfaces at the points of $\gamma_k$ 
with parameters
$$
	0\le s_0<\cdots <s_{J_k-1}<\ell_k,
$$
which then induces a phase jump $e^{ib_{j}}$ across
the points $s_j$. We have $b_{j}=\sum_{k'\neq k}b_{j,k'}$ where $e^{ib_{j,k'}}$ is the phase jump
due to $S_{k'}$ at $s_j$, with $b_{j,k'}=0$ if there is no intersection, else $b_{j,k'}=\pm 2\pi\alpha_{k'}$.
We use the convention $s_{J_k+1}=s_1+\ell_k$ and $s_0 = s_{J_k}-\ell_k$.

On $\R/(\ell_k\Z)$, $D_{k,\bA}$ denotes the self-adjoint extension of $-i\tfrac{\d}{\d s}$ 
whose domain is $H^1$ with the phase jump conditions across the $s_j$'s.
Through the arclength parametrization, we see $D_{k,\bA}$ as
a self-adjoint operator on $L^2(\gamma_k)$.

The $Spin^c$-spinor bundle $\Psi$ restricted to $\gamma_k$ is canonically split into
two different complex line bundles $L_{k}^{(\pm)}$, corresponding to the two different lines $\C\xi_{\pm}$.
We write $\pxi{\pm}$ the (pointwise) projection onto these lines,
$$
	\pxi{\pm}:= \cip{\xi_{\pm}}{\,\cdot\,}\xi_{\pm}.
$$
The canonical $Spin^c$-connection induces the elliptic operators 
$\cT_{k}^{(\pm)}$ on $L_{k}^{(\pm)}$
defined by 
$$
	\cT_{k}^{(\pm)}:=-i\pxi{\pm}\sigma(\bT^{\flat})\nabla_{\bT} \pxi{\pm}.
$$
Taking into account the phase jumps due
to the other $S_{k'}$, we write $\cT_{k,\bA}^{(\pm)}$ for the self-adjoint operator on 
the set of $L^2$-sections of $L_{k}^{(\pm)}$
with domain
$$
	\dom(\cT_{k,\bA}^{(\pm)}):=\big\{ f_{\pm}\xi_{\pm}: f_{\pm}\in\dom(D_{k,\bA})\big\}.
$$
Using the sections $(\xi_+,\xi_-)$ to trivialize $L_{k}^{(\pm)}$, the operators $\cT_{k,\bA}^{(\pm)}$ act
as follows on $L^2(\R/(\ell_k\Z))\simeq L_{k}^{(\pm)}$:
$$
	\cT_{k,\bA}^{(\pm)}\big(f_{\pm}\xi_{\pm}\big)
	=\pm \big((D_{k,\bA}f_{\pm})-i\cip{\xi_{\pm}}{\nabla_{\bT} \xi_{\pm}} f_{\pm} \big)\xi_{\pm}.
$$
Note that changing the Seifert surface $S_k$ for $\gamma_k$ to $S_k'$ induces a gauge transformation
on $\cT_{k,\bA}^{(\pm)}$ seen as operators on $L^2(\R/(\ell_k\Z))$, 
leaving the spectrum of the operator invariant. We now focus on $\cT_{k,\bA}:=\cT_{k,\bA}^{(-)}$.
The operator $\cT_{k,\bA}+\tau_{S_k}$, where $\tau_{S_k}$ denotes the relative torsion of $\gamma_k$ with respect to $S_k$ 
(see \eqref{eq:eq_darboux_frame}), turns out to be an important effective operator in our study. We now calculate its spectrum.

\begin{proposition}[Spectrum of the effective operator]\label{prop:calcul_spectre}
	Let $\ua \in (0,1)^{K_0} \times \{0\}^{K_1}$, and $\uS\in \sS^{(K)}$ with $\gamma_k=\partial S_k$ and
	$\bA:=2\pi\sum_{k=1}^{K_0}\alpha_k[S_k]$.
	Then for $K_0+1 \leq k \leq K$ we have
	\begin{equation*}
		\spec\left(\cT_{k,\bA}+\tau_{S_k}\right)
		= \Bigg\{\frac{1}{\ell_{k}}\Big(2n\pi+\pi(1-\Wr(\gamma_k))
		-\Phi_{\bB}(\gamma_{k})\Big), \quad n \in\Z \Bigg\},
	\end{equation*}
	where $\Wr(\gamma_k)$ denotes the writhe of $\gamma_k$ (see Section~\ref{sec:techn})
	and $\Phi_{\bB}(\gamma_{k})$ denotes the flux of $\bB:=\partial \bA$ through $S_k$
	(or the circulation of $\bA$ along $\gamma_k$):
	\[
	\Phi_{\bB}(\gamma_{k})=2\pi\sum_{k_0=1}^{K_0}\alpha_{k_0}\link(\gamma_k,\gamma_{k_0}).
	\]
	Furthermore, each eigenvalue has multiplicity one.
\end{proposition}
\begin{rem}
	We have defined $\cT_{k,\bA}$ when the last $K_1$ fluxes are zero, but in principle these $K_1$ zeros
	can be located anywhere on $\{1,\dots, K\}$. The Proposition still holds for the corresponding $\cT_{k,\bA}$'s.
\end{rem}

\begin{proof}
	The eigenvalue problem can be solved explicitly. We write down the eigen-equation 
	$(\cT_{k,\bA}+\tau_{S_k})(f\xi_-)=\lambda f\xi_-$, where we assume $f\not\equiv 0$.
	On each subinterval $(s_1,s_2)\subset \T_{\ell_k}$ between two phase jumps $e^{i b_1}$ and $e^{i b_2}$ 
	across $s_1$ and $s_2$, $f$ must satisfy the equation:
	$
		\big(i\partial_s  +i\cip{\xi_-}{\nabla_{\bT} \xi_-}+\tau_{S_k}\big)f=\lambda f.
	$
	Solving the equation on $(s_1,s_2)$, we get
	\[
		f(s_2^{+})=e^{i b_2}f(s_2^{-})=\mathrm{exp}\Big(ib_2-i\int_{s_1}^{s_2}(\lambda-\tau_{S_k} -i\cip{\xi_-}{\nabla_{\bT} \xi_-})\d s\Big)f(s_1^{+}).
	\]
	Going over a full cycle, we obtain the (necessary and sufficient) condition
	\[
		\sum_{j}b_j-\lambda \ell_k+\int_{\gamma_k}(i\cip{\xi_-}{\nabla_{\bT} \xi_-}+\tau_{S_k})\bT^{\flat}\equiv 0\mod 2\pi,
	\]
	and the eigenvalues $\lambda_n$ have multiplicity one. Furthermore for all $k'\neq k$, as $\link(\gamma_{k'},\gamma_k)$ 
	corresponds to the number of algebraic crossing of $\gamma_k$ through $S_{k'}$ \cite{Rolfsen}*{Part D, Chapter 5}, there holds:
	$
	\sum_{j}b_{j,k'}=-2\pi\alpha_{k'}\link(\gamma_{k'},\gamma_k).
	$
	So the eigenvalues $\lambda_n$'s can be rewritten as
	\begin{align*}
		\lambda_n
		=\frac{1}{\ell_{k}}\big( 2n\pi +\!\!\int_{\gamma_k}(i\cip{\xi_-}{\nabla_{\bT} \xi_-}\!+\!\tau_{S_k})\bT^{\flat}
		-2\pi\!\sum_{k_0=1}^{K_0}\!\alpha_{k_0}\link(\gamma_k,\gamma_{k_0})\big),\,n\in\Z.
	\end{align*}
	From \eqref{eq:spin_basis} we get the equality $i\cip{\xi_-}{\nabla_{\bT} \xi_-}+\tau_{S_k} = \omega_{\xi}(\bT)+\tfrac{1}{2}\tau_{S_k}$. That
	the integral $\int_{\gamma_k}\omega_\xi$ equals $\pi\!\! \mod 2\pi$ will be shown in Section~\ref{app:comp_omega_xi}.
	In Section~\ref{sec:techn}, we show that the integrated torsion $\int_{\gamma_k} \tau_{S_k}\bT^{\flat}$
	is independent of the choice of $S_k$ and that it is equal to $-2\pi$ times the writhe of any stereographic projection of $\gamma_k$.
\end{proof}

\subsection{Continuity properties}\label{sec:cont_prop}
We now see the map $\bbd$ defined in \eqref{eq:phi_no_top} as a map
between $(\sS^{(K)}\times \Tf^{K},\sT_{\sS}^{(K)})$, 
and $\Sd(L^2(\S^3)^2)$ endowed with some $\sT_{\phi}$, associating a Dirac operator to Seifert surfaces and fluxes.
We present several theorems describing its continuity properties.

\subsubsection{Strong resolvent continuity}

Since the domain of $\bbd$ is metric, continuity is equivalent to sequential continuity and
the following theorem follows immediately from \cite[Theorem~23]{dirac_s3_paper1}.
\begin{theorem}[Strong resolvent continuity]\label{thm:strg_res_cont}
	The map $\bbd$ is continuous in the strong resolvent sense.
\end{theorem}

\subsubsection{Bump-continuity}

The next theorems (\ref{thm:bump_cont_bulk}-\ref{thm:bump_cont_bdry_A}-\ref{thm:bump_cont_bdry_B}) address the bump continuity of the function
$\bbd$. As their proofs are rather technical, they will be given in Section~\ref{sec:bump_cont_proofs}.

Note that the introduced topologies are only concerned
with the spectral point $\lambda = 0$. However in order to
compute the spectral flow, we will also need to study continuity for arbitrary
spectral points $\lambda \in \R$.

\begin{theorem}[Bump-continuity in the bulk]\label{thm:bump_cont_bulk}
	For any bump function $\phi$ (centered at arbitrary $\lambda \in \R$), the 
	following map is continuous in the bulk $\sS^{(K)}\times (0,1)^{K}$
	\begin{align}\label{eq:phi_bump_cont}
		\phi \circ \bbd:
		\begin{array}{ccc}
			\sS^{(K)}\times \Tf^K
			&\longrightarrow& 
			(\mathcal{B}\big(L^2(\S^3)^2\big),\norm{\cdot}_{\cB})\\
			(\uS,\ua) &\mapsto& \phi\big(\cD_{\bA}\big).
		\end{array}
	\end{align}
\end{theorem}

Theorem~\ref{thm:compactness} suggests that this continuity may fail
at ``boundary'' points $(\uS,\ua)\in\sS^{(K)}\times \rT_b(K)$,
where 
\begin{equation}
	\rT_b(K):=\big\{\ua\in\Tf^K,\ \exists\ 1\le k\le K,\ \alpha_k=0\big\}.
\end{equation}
In fact, we can give an exact characterization of the domain of bump-continuity.
\begin{theorem}[Bump-continuity at the boundary]\label{thm:bump_cont_bdry_A}
	Let $\ua\in \rT_b(K)$, $\uS\in\sS^{(K)}$ and $\lambda \in \R$. Then
	the following two statements are equivalent.
	\begin{enumerate}[leftmargin=*]
		\item There exists a bump function $\phi$ centered at
		$\lambda\in\R $ such that the map $\phi \circ \bbd$ from \eqref{eq:phi_bump_cont}
		is continuous on an open neighborhood of $(\uS,\ua)$.
		
		\item For all $k=1,\ldots,K$ for which $\alpha_{k}=0$
		we have $\lambda\not\in\text{\rm spec}(\cT_{k,\bA}+\tau_{S_k})$.
	\end{enumerate}
\end{theorem}

The points in the spectrum $\spec\big( \cT_{k,\wt{\bA}}+\tau_{\wt{S}_k}\big)$,
which was given explicitly in Proposition~\ref{prop:calcul_spectre}, are the
limits of the eigenvalues of the eigenfunctions that collapse onto the
knot $\wt{\gamma}_{k}$ as the flux $\alpha_k'$ converges to $1^-\in\Tf$ in
the limit $(\uS',\ua')\to (\widetilde{\uS},\widetilde{\ua})$. We now describe this
collapsing process in greater details.

\subsubsection{Decomposition of the ``boundary''}
There is a natural cell decomposition of the boundary
\begin{align*}
	\rT_b(K)&=\bigcup_{\emptyset\ne R\subset\{1,\ldots,K\}}\Gamma(R),\\ 
	\Gamma(R)&:=\{\ua\in\Tf^K,\ 
	\alpha_k=0\text{ if and only if }k\in R\}.
\end{align*}
If $\ua\in\Gamma(R)$ and\footnote{Setting $\ell:=\inf \{\dist_{\Tf}(\alpha_k,0),\ \alpha_k\neq 0\}$, it is sufficient to take $\eps\in (0,\tfrac{\ell}{2})$ or $\eps<\tfrac{1}{2}$ if $\ua=0$.} $0<\varepsilon<\eps(\ua)$, the set 
$B_\varepsilon(\ua)\setminus\rT_b(K)$ has the following $2^{|R|}$ 
connected components indexed by subsets $R'\subset R$:
\begin{equation}\label{eq:Ceps}
	\begin{array}{rl}
		C_\varepsilon(\ua, R')=
		\Big\{\ua'\in B_\varepsilon(\ua)\setminus\rT_b(K),&\text{For all } k\in R,
		1/2<\alpha_k'<1 \\&\text{if and only if }
		k\in R'\Big\}.
	\end{array}
\end{equation}
For points $(\uS',\ua')\in \sS^{(K)}\times
C_\varepsilon(\ua,R')$ approaching $(\uS,\ua)\in
\sS^{(K)}\times\Gamma(R)$, the collapsing eigenfunctions with
eigenvalues tending to some $\lambda \in\R$ will collapse onto
knots $\gamma_k$ for which $k\in R'$ with  
$\lambda\in \spec(\cT_{k,\bA}+\tau_{S_k})$. Moreover, 
there will be exactly one collapsing eigenfunction for each such knot.  

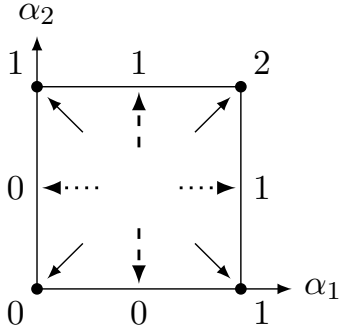
\begin{figure}[!!!ht]
	\resizebox{0.4\textwidth}{!}{
	\begin{tikzpicture}
		\draw (0,0) -- (2,0) -- (2,2) -- (0,2) -- cycle;
	
	\foreach \x in {0,2}
		{
			\foreach \y in {0,2}
			{
				\draw [fill] (\x,\y) circle [radius = 0.05];
			}
		}
	\node[above left, scale=0.9] at (0,2) {$1$};
	\node[below left, scale=0.9] at (0,0) {$0$};
	\node[below right, scale=0.9] at (2,0) {$1$};
	\node[above right, scale=0.9] at (2,2) {$2$};
	\node[left, scale=0.9] at (0,1) {$0$};
	\node[below, scale=0.9] at (1,0) {$0$};
	\node[right, scale=0.9] at (2,1) {$1$};
	\node[above, scale=0.9] at (1,2) {$1$};
	
	\draw[>=latex,->] (0.45,1.55)--(0.1, 1.9);
	\draw[thick,>=latex,->,dotted] (0.6,1)--(0.05, 1);
	\draw[>=latex,->] (0.45,0.45)--(0.1, 0.1);
	\draw [thick,>=latex,->,dashed] (1,0.6) -- (1,0.05);
	
	\draw [>=latex,->] (1.55,0.45) -- (1.9,0.1);
	\draw [thick,>=latex,->,dotted] (1.4,1) -- (1.95,1);
	\draw [>=latex,->] (1.55,1.55) -- (1.9,1.9);
	\draw [thick,>=latex,->,dashed] (1,1.4) -- (1,1.95);
	
	\draw [>=latex,->] (2,0)--(2.5,0);
	\draw [>=latex,->] (0,2)--(0,2.5);
	
	\node[thin, right, scale=0.9] at (2.5,0) {$\alpha_1$};
	\node[thin, above, scale=0.9] at (0,2.5) {$\alpha_2$};
	\end{tikzpicture}
	}
	\caption{Near each of the boundary points $\ua=(1/2,0)$ and
	$\ua=(0,1/2)$ there are two boundary components
	$C_\varepsilon(\ua,R')$, indicated by dashed and
	dotted arrows respectively. Near the point $\ua=(0,0)$ there are four
	boundary components indicated by solid arrows. The numbers
	refer to the size of the set $R'$ in each case. }
	\label{fig:diff_regm}%
\end{figure}

In the next theorem, we introduce the subspace $V$ of eigenfunctions, 
with eigenvalues close to a given spectral level $\lambda$, that``vanishes" 
when reaching a given point on the ``boundary". The collapse can only 
occur in the limit where some fluxes $\alpha_k$'s converge to $1^-$, and the eigenfunctions
concentrates on the corresponding knots $\gamma_k$'s.

\begin{theorem}[Description of the discontinuity at the boundary]\label{thm:bump_cont_bdry_B}
	Consider a non-empty set $R\subset\{1,\ldots,K\}$, $\widetilde{\ua}\in\Gamma(R)$,
	$\widetilde{\uS}\in\sS^{(K)}$, and $\lambda\in\R$. Then there exist
	$0<\eps<2^{-1}$ and $\eta>0$ such that to each point $(\uS,\ua)\in
	B_\eps[\widetilde{\uS}]\times B_\eps[\widetilde{\ua}]$ we can associate a subspace (denoted the
	vanishing subspace)
	\[
		V=V(\bA,\lambda)\subset \ran\,\mathds{1}_{[\lambda-\eta,\lambda+\eta]}
		\big( \cD_{\bA}\big),
	\]
	spanned by eigenfunctions of $\cD_{\bA}$, with the following properties. 
	\begin{enumerate} [leftmargin=*]
		\item The dimension of the vanishing subspace is 
		$$
			\dim V(\bA,\lambda)=\sum_{k\in R,\atop
			1/2<\alpha_k<1}\dim\ker \big(\cT_{k,\wt{\bA}}
			+\tau_{\wt{S}_k}-\lambda\big)=:d(\uS,\ua,\lambda),
		$$
		where for all $k$ we have $\dim\ker \big(\cT_{k,\wt{\bA}}+\tau_{\wt{S}_k}-\lambda\big)\in\{0,1\}$.

		\item The projection-valued map
		$$
			B_\eps[\widetilde{\uS}]\times B_\eps[\widetilde{\ua}]\ni (\uS,\ua)\mapsto
			\mathds{1}_{[\lambda-\eta,\lambda+\eta]}\big( \cD_{\bA}\big)
			-P_{V(\bA,\lambda)}=:P_{W(\bA,\lambda)}
		$$
		is continuous in the norm-topology.
	
		\item The maps 
		\begin{eqnarray*}
			B_\eps[\widetilde{\uS}]\times B_\eps[\widetilde{\ua}]\ni (\uS,\ua)&\mapsto& P_{V(\bA,\lambda)}\\
			B_\eps[\widetilde{\uS}]\times B_\eps[\widetilde{\ua}]\ni (\uS,\ua)&\mapsto& \cD_{\bA}P_{V(\bA,\lambda)}
		\end{eqnarray*}
		are continuous in the strong operator topology. They are continuous in operator norm on the subset 
		$B_{\eps}(\widetilde{\uS})\times \big(B_\eps(\widetilde{\ua})\setminus \rT_b(K)\big)$.
	
		\item The eigenvalues of ${\cD_{\bA}}\big|_{V(\bA,\lambda)}$ are
		continuous as functions of $(\uS,\ua)$ on
		$B_{\eps}(\widetilde{\uS})\times \big(B_\eps(\widetilde{\ua})\setminus
		\rT_b(K)\big)$. From each connected component of the set, i.e., 
		$B_{\eps}(\widetilde{\uS})\times C_\eps(\wt{\ua},R')$ the eigenvalues 
		have well defined limits on 
		$B_{\eps}(\widetilde{\uS})\times (\overline{C_\eps(\wt{\ua},R')}\cap\rT_b(K))$.  In fact, 
		the corresponding set of limiting values at a point $(\uS',\ua')\in
		B_{\eps}(\widetilde{\uS})\times (\overline{C_\eps(\wt{\ua},R')}\cap\rT_b(K))$ is
		$$
			\bigcup_{k\in R',\atop \alpha'_k=0}
			\Big(\spec(\cT_{k,\bA'}+
			\tau_{S_k'})\cap[\lambda-\eta,\lambda+\eta]\Big)\cup \spec(\cD_{\bA'}\big|_{V(\bA',\lambda)}),
		$$	
		where $\bA'$ is the potential defined by $(\uS',\ua')$.
		\item For $\mu>\eta$, up to taking $0<\eps(\mu)\le \eps$ the following map is $\sT_{\phi}$-continuous,
		where $\phi$ is a bump function with support $[\lambda-\eta,\lambda+\eta]$:
		\[
			B_{\eps(\mu)}[\widetilde{\uS}]\times B_{\eps(\mu)}[\widetilde{\ua}]\ni (\uS,\ua)\mapsto
			\cD_{\bA}(1-P_{V(\bA,\lambda)})+ \mu P_{V(\bA,\lambda)}-\lambda =: G(\bA,\lambda).
		\]
	\end{enumerate}
\end{theorem}

\begin{rem}\label{rem:cont_discont}
	\noindent 1. The last point just means that, up to adding $-\lambda\mathrm{Id}$, we can assume that $\lambda=0$, 
	and then up to shifting the eigenvalues of the vanishing subspace away from $0$, we obtain $\sT_{\phi}$-continuity.
		
		\smallskip
		
	\noindent 2. Note that the choice of $\eta$ implicitly means that if $\lambda\in\spec(\cT_{k,\wt{\bA}}+\tau_{\wt{S}_k})$, then
	in point (4) we have according to Proposition~\ref{prop:calcul_spectre}:
	\begin{multline}\label{eq:prec_eig}
		[\lambda-\eta,\lambda+\eta]\cap \spec(\cT_{k,\bA'}+\tau_{S_k'})\\
		=\frac{1}{\ell_k'}\big\{\lambda\wt{\ell}_k-\pi\big(\Wr(\gamma_k')-\Wr(\wt{\gamma}_k)\big)-
		2\pi\sum_{k_1\neq k}(\alpha_{k_1}'-\wt{\alpha}_{k_1})\link(\wt{\gamma}_{k_1},\wt{\gamma}_{k})\big\},
	\end{multline}
	where in the above formula we assume that 
	$0\le \wt{\alpha}_{k_1},\alpha_{k_1}'<1$. And if $\lambda$ is not in the spectrum,
	then the set~\eqref{eq:prec_eig} is empty.
\end{rem}

\subsubsection{Continuity and discontinuity}\label{sec:cont_discont}
	The map $\mathbb{D}$ \eqref{eq:phi_no_top} is defined everywhere, also on the critical set of $\sS^{(K)}\times\Tf^K$. The definition implies that if one of the fluxes $\alpha_k$ vanishes,
	then the corresponding knot is completely absent in the magnetic field. 
	By definition $\mathbb{D}$ is not bump-continuous on the critical points, nevertheless we can still have bump-continuity within some region.
	
	 Indeed, consider Theorem~\ref{thm:bump_cont_bdry_B}: in the region of $B_{\eps}(\widetilde{\uS})\times \big(B_\eps(\widetilde{\ua})\setminus \rT_b(K)\big)$ where $V(\ua,\lambda)$ is trivial,
	(2)-(3)-(5) imply that we have $\lambda$-bump continuity when $\ua$ approaches the boundary $\Gamma(R)\cap B_\eps[\wt{\ua}]$ within this region.

	An important special case is the continuity ``from the right", which we now state for $\lambda=0$. 
	
	Let $\uS\in \sS^{(K)}$ and fix the first $K_1$ fluxes to be non-zero $\ua'\in (0,1)^{K_1}$. 
	By Theorem~\ref{thm:bump_cont_bdry_B} and Proposition~\ref{prop:calcul_spectre} the following holds.
	\begin{enumerate}
	\item For some $\eps'>0$ we have bump-continuity of $\bbd$ within the subset $B_{\eps'}[\uS]\times B_{\eps'}[\ua']\times [0,\eps')^{K-K_1}$.
	\item  Provided that we have
	\[
		\frac{1}{2}(1-\Wr(\gamma_{K}))-\sum_{k=1}^{K_1}\alpha_k\link(\gamma_1,\gamma_K)\notin\Z,
	\]
	 then for some $\eps'\in(0,\eps)$ we have bump-continuity of the map $\mathbb{D}$ within the set 
	 $
	 B_{\eps'}[\uS]\times \times B_{\eps'}[\ua']\times [0,\eps')^{K-K_1-1}\times\Tf.
	 $
	 \end{enumerate}

	As an example, take $\uS\in \sS^{(2)}$ with $\link(\gamma_1,\gamma_2)\neq 0$. Assume $\Wr(\gamma_2)\notin 2\Z+1$: 
	the loop $\alpha_2\mapsto \cD_{2\pi\alpha_2[S_2]}$ is continuous, but the loop $\big((0,\alpha_2)\big)_{\alpha_2\in\Tf}$ does not lie in $\T_{\ug}$.
	However, for $\eps>0$ small enough, the loop $(\eps,\alpha_2)_{\alpha_2\in\Tf}$ lies in $\T_{\ug}$, and due to the continuity from the right -- here for $\alpha_1\to 0^+$ --
	we have the bump-continuous homotopy:
	\[
		(\alpha_2,s)\in \Tf\times[0,\eps]\mapsto \cD_{2\pi(\alpha_2[S_2]+s[S_1])}.
	\]

\section{Spectral flow on the torus of fluxes}\label{sec:spectral_flow}

\subsection{Definition of the spectral flow}
\subsubsection{The spectral flow in $(\Sd,\cT_{\phi})$ and in $(\mathrm{SF},\sT_{W})$}
We define the spectral flow as in \cite[Definition~2.1]{Wahl08} (see also \cite{Philips_spectral_flow}).
Let $\phi\in\sD(\R,\R_+)$ be a bump function, $\supp\phi=[-x_0,x_0]$ 
and let $\sT_{\phi}$ be the corresponding topology on $\Sd$ (see \eqref{eq:def_t_phi}). Recall Remark~\ref{rem:Wahl_top}: $(\mathrm{SF},\sT_{W})$
denotes the set of Fredholm self-adjoint operators with the Wahl topology \cite{Wahl08} and the injection 
$(\Sd,\sT_{\phi})\hookrightarrow (\mathrm{SF},\sT_{W})$ is continuous.

For a $\sT_{\phi}$-continuous path $(\cD_t)_{t \in [a,b]}$ with invertible endpoints $\cD_a$ and $\cD_b$, 
assume that there exists $\mu>0$ such that $\pm \mu \in \res \cD_t$ for all $t \in [a,b]$.
We then define
\begin{align*}
	\Sf\big[(\cD_t)_{t \in [a,b]}\big] 
	:= \dim \ran \mathds{1}_{[0,\mu]}(\cD_b) - \dim \ran \mathds{1}_{[0,\mu]}(\cD_a).
\end{align*}
Note that if such a $\mu$ does not exist for the entire interval $[a,b]$, 
we split the curve into several parts and define the spectral flow piecewise.
The spectral flow of the whole curve is then the sum of the individual contributions.
\begin{rem}
 Let $f_{x_0}$ be the continuous odd function defined by $f_{x_0}(x)=x$ for $|x|\le \tfrac{x_0}{2}$
 and $f_{x_0}(\pm x)=\pm \tfrac{x_0}{2}$ for $x\ge \tfrac{x_0}{2}$, and $g_{x_0}$ be the smooth function defined on $\R$ by
 $g_{x_0}(x):=\mathrm{exp}\big(2i\pi \tfrac{x}{x_0}\big)$. We note $h_{x_0}:=g_{x_0}\circ f_{x_0}$.
 
By \cite[Theorem~VIII.20]{ReedSimon1} and the equality $\phi\circ f_{x_0}=\wt{f}_{x_0}\circ\phi$, with $\wt{f}_{x_0}(y):=\max\big[y,\phi\big(\tfrac{x_0}{2}\big)\big]$,
the family $(f_{x_0}(\cD_t))_t$ is $\sT_{\phi}$-continuous on $\mathrm{SF}(L^2(\S^3)^2)$ (see Remark~\ref{rem:Wahl_top}).
 Hence $(h_{x_0}(\cD_t))_t$ is operator norm-continuous. For any $t$, $h_{x_0}(\cD_t)$ has spectrum embedded in $\mathbb{S}^1$
 with essential spectrum $\{-1\}$ (Figure~\ref{fig:spectru_circle}). 
 
 We can carry over Philips' definition of the spectral flow \cite{Philips_spectral_flow} for $(h_{x_0}(\cD_t))_t$
 seen as the number of eigenvalues crossing $1$ in the positive direction: 
 the spectral flow of $(\cD_t)_{t \in [a,b]}$ coincides with that of $(h_{x_0}(\cD_t))_{t \in [a,b]}$.
\end{rem}

\begin{figure}[!ht]
	\resizebox{0.25\textwidth}{!}{
	\begin{tikzpicture}
		\draw (0,0) circle [radius=2];
		\draw [fill] (-2,0) circle [radius=0.2];
		\node[left] at (-2.2,0) {\large$-1$};
		\draw[thick] (1.8,0) -- (2.2,0);
		\node at (2.4,0) {\large$1$};
		\draw[fill,color=gray] ({-1},{-sqrt(3)}) circle [radius=0.1];
		\draw[fill,color=gray] ({2*cos(25)},{-2*sin(25)}) circle [radius=0.1];
		\draw[fill,color=gray] ({2*cos(80)},{2*sin(80)}) circle [radius=0.1];
		\draw[thick,->] ({2.8*cos(11)},{-2.8*sin(11)}) arc (-11:11:2.8); 
	\end{tikzpicture}
	}
	\caption{The spectrum of $h_{x_0}(\cD_t)$. As $t$ varies, the point $-1$ can be seen as a r\'eservoir from which only finitely many eigenvalues
	can emerge and move along $\S^1$. 
	}
	\label{fig:spectru_circle}
\end{figure}
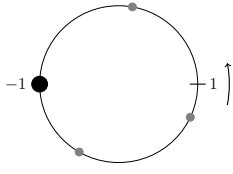

Through the injection 
 $(\Sd,\sT_{\phi})\hookrightarrow (\mathrm{SF},\sT_{W})$, we can see $(\cD_t)_{t \in [a,b]}$
 as a continuous path in $(\mathrm{SF},\sT_{W})$, and replace $(\Sd,\sT_{\phi})$
 by $(\mathrm{SF},\sT_{W})$ in the discussion below.

Let $\phi$ be some bump function around $0$.
We recall several important properties of $\Sf$ from \cite[Section~2]{Wahl08}, which we write for $(\Sd,\sT_{\phi})$:
\begin{enumerate}
	\item The map $\Sf$ is additive with respect to concatenation of paths.
	
	\item Homotopy invariance for open paths: 
	if $(\cD_{(r,t)})_{(r,t) \in [0,1]\times [a,b]}$ is a continuous
	family in $(\Sd,\sT_{\phi})$, such that $\cD_{(r,a)}$ and $\cD_{(r,b)}$ are invertible for all $r \in [0,1]$,
	then we have 
	$
		\Sf\big[(\cD_{(0,t)})_{t \in [a,b]}\big]
		= \Sf\big[(\cD_{(1,t)})_{t \in [a,b]}\big].
	$
	
	\item Homotopy invariance for loops: if $(\cD_{(r,t)})_{(r,t) \in [0,1]\times [a,b]}$ is a continuous
	family in $(\Sd,\sT_{\phi})$, such that $\cD_{(r,a)} = \cD_{(r,b)}$ for all $r \in [0,1]$,
	then
	\begin{align}\label{eq:sf_prop_3}
		\Sf\big[(\cD_{(0,t)})_{t \in [a,b]}\big]
		= \Sf\big[(\cD_{(1,t)})_{t \in [a,b]}\big].
	\end{align}
\end{enumerate}
By the above,  $\Sf$ defines a group homomorphism
$\Sf:\pi_1((\Sd,\sT_{\phi}))\to \Z$, which can be factorized through the first homology group:

\begin{figure}[!!!!!ht]
	\resizebox{0.55\textwidth}{!}{
	\begin{tikzpicture}
	\node[left] at (0.2,0) {$\pi_1((\Sd,\sT_{\phi}))$};
	\node[below] at (3,-1.5) {$H_1((\Sd,\sT_{\phi}))$};
	\node at (3.25,0) {$\mathbb{Z}$};
	
	\draw [->] (0.2,0) -- (2.8,0);
	\draw [->] (3.25,-1.5) -- (3.25,-0.3);
	\draw [->] (-0.4,-0.4) -- (2.2,-1.5);
	
	\node[above,scale=0.8] at (1.4,0) {sf};
	\node[left,scale=0.8] at (3.25,-0.75) {$\widetilde{\mathrm{sf}}$};
	\node[above right,scale=0.8] at (0.8,-0.9) {ab};
	
	\end{tikzpicture}
	}
\end{figure}

Here, $\pi_1((\Sd,\sT_{\phi}))$ denotes the fundamental group, $H_1((\Sd,\sT_{\phi}))$ the first homology group and
 $\mathrm{ab}$ the abelianization.

In accordance with Theorem~\ref{thm:bump_cont_bulk} \&~\ref{thm:bump_cont_bdry_A}
we can characterize a ``regular region" where the map $\bbd$ is bump-continuous (see Def.~\ref{def:bump_continuity}).
Its complement is called the ``critical region", where $\bbd$ fails to be bump-continuous. More precisely the critical region is:
\begin{align}\label{eq:crit_regi}
	\sR_{\crit}^{(K)} :&= \{(\uS',\ua') \in \sS^{(K)} \times \Tf^{K}:\\
	&\qquad \exists\, 1 \leq k \leq K \textrm{ such that } \alpha_k' =0 \textrm{ and }
	0 \in \spec\big(\cT_{k,\bA'}+\tau_{S_k'}\big)\}.
\end{align}
The regular region is:
\begin{align*}
	\sR_{\reg}^{(K)} := \big(\sS^{(K)} \times \Tf^{K}\big) \setminus
	\sR_{\crit}^{(K)},
\end{align*}
Injecting $\Sd$ into $(\mathrm{SF},\sT_{W})$, the map $\bbd$ induces the group homomorphism
\begin{equation}
	\Sf_{\bbd}:= \Sf \circ\, \bbd_{*} : \pi_1\big(\sR_{\reg}^{(K)}\big) \longrightarrow \Z,
\end{equation}
where $\bbd_*$ is the homomorphism $\bbd_*:\pi_1\big((\sR_{\reg}^{(K)},\dist)\big)\to \pi_1( (\mathrm{SF},\sT_{W}))$.

Similarly, when fixing a family $\uS \in \sS^{(K)}$ of Seifert surfaces, by Proposition~\ref{prop:calcul_spectre}
and Theorem~\ref{thm:bump_cont_bdry_A} the regular region is reduced to
$$
	(\sR_{\reg}^{(K)})_{|_{\uS'=\uS}}=\{ \uS\}\times \T_{\ug}\subset \sS^{(K)} \times \Tf^{K},
$$
where $\T_{\ug}$ is the \textit{(cut) torus of fluxes}: 
\begin{multline}\label{def:T_uS}
	\T_{\ug}:=\Tf^{K}\setminus\big\{\ua\in \Tf^K,\ \exists\,1\le k\le K,\ \alpha_{k}=0\  \mathrm{and}\\
	\pi\Wr(\gamma_{k})+2\pi\sum_{k'\neq k}\alpha_{k'}\link(\gamma_{k},\gamma_{k'})=\pi\!\!\mod 2\pi\big\}.
\end{multline}

Recall that $\Wr(\gamma_k)$ corresponds to the writhe of $\gamma_k$ 
(or more precisely that of any stereographic projection of $\gamma_k$) see Section~\ref{sec:techn}).
Furthermore, writing $\bA=\sum_{k'\neq k}2\pi\alpha_{k'}[S_{k'}]$ and $\bB=\partial \bA$, the quantity  
$2\pi\sum_{k'\neq k}\alpha_{k'}\link(\gamma_{k},\gamma_{k'})$ in \eqref{def:T_uS} is the flux
$\Phi_{\bB}(\gamma_k)$ of $\bB$ through the Seifert surface $S_k$.

The continuous map 
\[
	\bbd_{\uS}:
	\begin{array}{ccc}
		\T_{\ug}& \longrightarrow & \Sd\subset (\mathrm{SF},\sT_{W}),\\
		\ua &\mapsto &\cD_{\bA},\ \bA=\sum_{k=1}^K 2\pi\alpha_k[S_k],
	\end{array}
\] 
induces the group homomorphism
\begin{equation}\label{eq:spec_flow_at_fixed_seif}
	\Sf_{\uS}:=\Sf \circ \big(\bbd_{\uS}\big)_*:\pi_1\big(\T_{\ug} \big)\to \Z.
\end{equation}

\subsection{Computing the spectral flow}

\subsubsection{Spectral flow for loop circling the critical set}
Computing the spectral flow for a general loop in $\sR_{\reg}^{(K)}$ can be very difficult.
Theorem~\ref{thm:bump_cont_bdry_A} however allows us to compute the spectral flow for a small
loop encircling a point $(\uS,\ua)$ in $\sR_{\crit}^{(K)}$, for which there exist $k,k_0$ with
$\link(\gamma_k,\gamma_{k_0})\neq 0$, $\alpha_{k_0}\in (0,1)$ and 
the critical conditions
\begin{equation}\label{eq:la_croisiere_s_amuse}
	\alpha_k=0\sim 1\quad\&\quad \pi\Wr(\gamma_k)+2\pi\sum_{k'\neq k}\alpha_{k'}\link(\gamma_k,\gamma_{k'})=\pi\!\!\mod 2\pi.
\end{equation}

We define the homotopy of loops
$$
	(\cD_{(r,t)})_{(r,t) \in (0,1]\times [0,1]} := \Big(\cD_{\bA(r,t)}\Big)_{(r,t) \in (0,1]\times [0,1]}
$$
by freezing $\uS$ and all fluxes $\ua$ \textit{except}
\begin{equation*}
		\left\{
			\begin{array}{rcl}
			\alpha_{k_0}(r,t) &=& \alpha_{k_0} + r\cos(2 \pi t),\\
			\alpha_k(r,t) &=& 1+r\sin(2\pi t)\in \Tf,
			\end{array}
		\right.
\end{equation*}

\begin{theorem}\label{thm:spec_flow_gen_l}
	For $r>0$ small enough, we have
	\begin{align}\label{eq:formula_sf_circle}
		\Sf\big[(\cD_{(r,t)})_{t \in [0,1]}\big] = -\sign \link(\gamma_{k_0},\gamma_k).
	\end{align}
\end{theorem}

\begin{rem}
	Note that a non-zero spectral flow implies the existence of a zero mode for at least one choice 
	of fluxes on the loop.
	
	Furthermore, the non-zero value ensures that the map $\bbd$ is not continuous
	at the point $(\uS,\ua)$ of the critical region enclosed by the loops. If it was, then the spectral flow would be $0$.
	
	The minus sign in Formula~\eqref{eq:formula_sf_circle} is due to our convention of the Dirac operator $\sigma(-i\nabla+\boldsymbol{\alpha})$.
\end{rem}
\begin{proof}[Proof of Theorem~\ref{thm:spec_flow_gen_l}]
	Without loss of generality we may assume that $\alpha_{k'} \neq 0$ for all $k' \neq k$ and $k'\neq k_0$, 
	as loops with vanishing fluxes are simply not present see Section~\ref{sec:cont_discont}.

	By Theorem~\ref{thm:bump_cont_bdry_A} (for $\lambda=0$) and 
	Proposition~\ref{prop:calcul_spectre}, the loop $\big(\cD_{(r,t)}\big)_{t\in[0,1]}$
	is $\sT_W$-continuous for $r$ small enough: there exists a bump function $\phi_r$ such that the loop is $\sT_{\phi_r}$-continuous 
	(see \eqref{eq:def_t_phi} and Remark~\ref{rem:Wahl_top}). 
	In fact, the only possible points of discontinuity are $t=1/2$ and $t=1$, where
	we know that $0$ is not in $\spec(\cT_{k,\bA(r,t)}+\tau_{\gamma_k})$ by Assumption~\eqref{eq:la_croisiere_s_amuse}.
	Hence the spectral flow is well-defined.
	
	We shall now apply Theorem~\ref{thm:bump_cont_bdry_B} at the critical point $(\uS,\ua)$
	for $\lambda=0$. It provides us with a spectral window $[-\eta,\eta]\neq \{0\}$ and a radius $0<\eps<2^{-1}$
	such that for all $(\uS',\ua')\in B_\eps[\uS]\times B_\eps[\ua]=U_\eps$, we have
	$$
		\ran\,\mathds{1}_{[-\eta,\eta]}\big( \cD_{\bA'}\big)=V(\uS',\ua')\overset{\perp}{\oplus}W(\uS',\ua'),
	$$
	satisfying the five points of the theorem. 
	
	As we have the critical Assumption~\eqref{eq:la_croisiere_s_amuse},
	Theorem~\ref{thm:bump_cont_bdry_B}(1) and Proposition~\ref{prop:calcul_spectre} imply that
	for $(\uS',\ua')$ in the subset $\{(\uS'',\ua'')\in U_\eps,\,1-\eps<\alpha_k''<1\}$
	the vanishing subspace $V(\uS',\ua')$ has dimension $1$ and it is spanned by an eigenfunction of $\cD_{\bA'}$.
	Theorem \ref{thm:bump_cont_bdry_B}(4) implies that the corresponding eigenvalue 
	$\lambda(\uS',\ua')$ has a well defined limit in the limit $\alpha_{k}'\to 1^-$
	given by \eqref{eq:prec_eig}.	
	Here the limit value $\wt{\lambda}$ at point $(\wt{\uS},\wt{\ua})$ (with $\wt{\alpha}_k=0$) 
	is the point in $\spec(\cT_{k,\bA(\wt{\uS},\wt{\ua})}+\tau_{\wt{S}_k})$ 
	that connects continuously to $0$ as $(\wt{\uS},\wt{\ua}) \to (\uS,\ua)$.

	Fix $0<r_0<\eps$, and consider the loop $(\cD_{(r_0,t)})_{t\in\R/\Z}$. 
	Let $V(r_0,t)$ and $W(r_0,t)$ denote the decomposition of $\ran\,\mathds{1}_{[-\eta,\eta]}\big(\cD_{(r_0,t)}\big)$,
	and let $\phi_{r_0}$ be a bump function such that the loop $\big(\cD_{(r_0,t)}\big)_{t\in[0,1]}$ is $\sT_{\phi_{r_0}}$-continuous.
	We write $\lambda(r_0,t)$ for the eigenvalue of $\cD_{(r_0,t)}$ corresponding to $V(r_0,t)$ and define
	\[
		G(r_0,t):=\cD_{(r_0,t)}(1-P_{V(r_0,t)})+(2\eta+1)P_{V(r_0,t)}.
	\]
	By Theorem~\ref{thm:bump_cont_bdry_B}(5), the contribution of $G(r_0,t)$ 
	to the spectral flow cancels due to \eqref{eq:sf_prop_3} and 
	the fact that the map $(t,r)\in [0,1]\times [0,r_0]\mapsto G(r,t)$ is a $\sT_{\phi_{r_0}}$-homotopy.
	Hence the main contribution to the spectral flow comes from the vanishing subspace
	$V(r_0,t)$. By \eqref{eq:prec_eig} we have
	\[
		\lim_{t\to (2^{-1})^+}\lambda(r_0,t)=\tfrac{2\pi}{\ell_k}r_0\link(\gamma_{k_0},\gamma_k),
		\quad \lim_{t\to 1^-}\lambda(r_0,t)=-\tfrac{2\pi}{\ell_k}r_0\link(\gamma_{k_0},\gamma_k).
	\]
	
	The loop $t\mapsto \cD_{(r_0,t)}$ is $\sT_{\phi_{r_0}}$-continuous, in particular at points $t=1/2$ and $t=1$.
	Let $0<\mu<\tfrac{2\pi}{\ell_k}r_0\le \eta$, $\mu\notin[\spec(\cD_{(r_0,1)})\cup\spec(\cD_{(r_0,2^{-1})})]$. By Theorem~\ref{thm:bump_cont_bdry_B}(2)
	and the above, there exists $0<\eps_0<4^{-1}$
	such that for $t\in J_1:=[1/2-\eps_0,1/2+\eps_0]$ and $t\in J_2:=[1-\eps_0,1+\eps_0]\subset \R/\Z$,
	$P_{\wt{W}(r_0,t)}:=\mathds{1}_{[-\mu,\mu]}(\cD_{r_0,t})$ is equal to $\mathds{1}_{[-\mu,\mu]}(\cD_{r_0,t}) P_{W(r_0,t)}$ with $\mu\notin\spec(\cD_{(r_0,t)})$.

	By Theorem~\ref{thm:bump_cont_bdry_B}(2)-(3), $t\mapsto \mathds{1}_{[-\eta,\eta]}(\cD_{(r_0,t)})$ is norm-continuous on $I_1:=[\eps_0,1/2-\eps_0]$ and $I_2:=[1/2+\eps_0,1-\eps_0]$.
	Hence for $\eta_1>\eta$ with $(\eta_1-\eta)>0$ small enough, $\eta_1$ is in the resolvent set of $\cD_{(r_0,t)}$ for $t\in I_1\cup I_2$.

	We now use the definition of the spectral flow and split the path into four parts according to the subdivision
	$\Tf=I_1\cup J_1\cup I_2\cup J_2$. We choose the splitting level $\mu$ on $J_1$ and $J_2$ and
	the splitting level $\eta_1$ on $I_1$ and $I_2$, and denote by
	$V_+(r_0,t),W_+(r_0,t),\wt{W}_+(r_0,t)$ the positive parts with respect to $\cD_{(r_0,t)}$ of the corresponding subspaces.
	Let $a=1/2-\eps_0$, $b=1/2+\eps_0$, $c=1-\eps_0$ and $d=1+\eps_0=\eps_0\mod 1$.
	From these considerations it follows that
	\begin{align*}
		&\Sf\big[(\cD_{(r_0,t)})_{t \in [0,1]}\big]\\
		&\quad = \big(\dim \wt{W}_+(r_0,d) - \dim \wt{W}_{+}(r_0,c)\big)\\
		&\qquad +\!\left(\dim\big[W_+(r_0,c) \oplus V_+(r_0,c)\big]\!-\!\dim\big[W_+(r_0,b) \oplus V_+(r_0,b)\big]\right)\\
		&\qquad +\!\big(\dim \wt{W}_+(r_0,b)\!-\!\dim \wt{W}_{+}(r_0,a)\big)\!+\!\left(\dim W_+(r_0,a)\!-\!\dim W_{+}(r_0,d)\right)\\
		&\quad = \dim V_+(r_0,c) - \dim V_+(r_0,b)+\Sf\big[ (G(r_0,t))_{t\in [0,1]}\big]\\
		&\quad = \dim V_+(r_0,c) - \dim V_+(r_0,b)\\
		&\quad = -\sign\link(\gamma_{k_0},\gamma_k).
	\end{align*}
\end{proof}

\subsubsection{Change of the spectral flow under deformations}

We now try to describe the spectral flow defined on 
the fundamental group of a torus $\T_{\ug}$ for a generic $\uS\in\sS^{(K)}$ (with $\partial S_k=\gamma_k$).
It suffices to give the spectral flow on a set of generators. If the critical region of $\T_{\ug}$ is empty,
then the set is the whole torus $\Tf^K$ and its fundamental group is generated by the loop along the edges
\[
	\ell_k:t\in\Tf\mapsto (0,\ldots,0,t,0,\ldots,0)\in \Tf^K.
\]
If the critical region is nonempty, we easily see that $\pi_1(\T_{\ug})$ is generated by the set of loops 
in $\T_{\ug}$ parallel to the $\ell_k$'s:
\begin{equation}\label{eq:def_L_k_alpha}
	\ell_{k,\ua}:t\in\Tf \mapsto (\alpha_1,\ldots,\alpha_{k-1},t,\alpha_{k+1},\ldots,\alpha_K)\in \Tf^{K},
\end{equation}
where $\ua\in (0,1)^{K-1}$. Recall \eqref{def:T_uS}: the loop $\ell_{k,\ua}$
is in $\T_{\ug}$ if and only if
	\begin{equation}\label{eq:cond_cont_loop}
		\pi\Wr(\gamma_{k})+\Phi_{k}(\ua)\neq \pi \!\!\mod 2\pi,
		\quad \Phi_k(\ua):=2\pi\sum_{k'\neq k}\alpha_{k'} \link(\gamma_{k},\gamma_{k'}).
	\end{equation}

Theorem~\eqref{thm:arbitrary_spectral_flow} describes how $\Sf(\ell_{k,\ua})$ changes 
under deformation of the link $\ug$ and the fluxes $\ua$.

\begin{theorem}[Change of the spectral flow under deformations]\label{thm:arbitrary_spectral_flow}
	Let $\ug,\wt{\ug}\in \sK^K$ be two sets of non intersecting knots, defining two isotopic links.
	
	Let $1\le k_0 \le K$, let $\ua^{(k_0)}=(\alpha_k^{(k_0)})_{k\neq k_0}$ 
	and $\wt{\ua}^{(k_0)}=(\wt{\alpha}_k^{(k_0)})_{k\neq k_0}$ in $\Tf^{(K-1)}$,
	and let $\uS,\wt{\uS}\in\sS^{(K)}$ with $\partial S_k=\gamma_k$ and $\partial \wt{S}_k=\wt{\gamma}_k$. 
	
	Let $\Phi_{k_0}(\ua^{(k_0)}),\Phi_{k_0}(\wt{\ua}^{(k_0)})$ be the corresponding fluxes as in \eqref{eq:cond_cont_loop}.
	We assume that neither of the following two numbers are equal to $\pi \!\!\mod 2\pi$:
	\begin{equation}\label{eq:convention}
	\pi\Wr(\gamma_{k_0})+\Phi_{k_0}(\ua^{(k_0)})\le \pi\Wr(\wt{\gamma}_{k_0})+\Phi_{k_0}(\wt{\ua}^{(k_0)}).
	\end{equation}
	Let $(\bA(\alpha))_{0\le \alpha\le 1}$, $(\wt{\bA}(\alpha))_{0\le \alpha\le 1}$ be the two
	paths of magnetic potentials:
	\[
	 \bA(\alpha)=2\pi\alpha [S_{k_0}]+\sum_{k\neq k_0}2\pi\alpha_{k_0}^{(k_0)}[S_k],
	    \wt{\bA}(\alpha)=2\pi\alpha [\wt{S}_{k_0}]+\sum_{k\neq k_0}2\pi\wt{\alpha}_{k_0}^{(k_0)}[\wt{S}_k],
	\]
	defining two loops of Dirac operators. The difference of their spectral flow 
	$\Sf\big( (\cD_{\wt{\bA}(\alpha)})_{\alpha}\big)-\Sf\big( (\cD_{\bA(\alpha)})_{\alpha}\big)$ is equal to:
	\begin{equation*}
	 \big\lfloor \tfrac{1}{2}(1-\Wr[\wt{\gamma}_{k_0}])-\tfrac{1}{2\pi}\Phi_{k_0}(\wt{\ua}^{(k_0)}) \big\rfloor-
	\big\lfloor\tfrac{1}{2}(1-\Wr[\gamma_{k_0}])-\tfrac{1}{2\pi}\Phi_{k_0}(\ua^{(k_0)}) \big\rfloor.
	\end{equation*}
\end{theorem}

\begin{proof}[Proof of Theorem~\ref{thm:arbitrary_spectral_flow}]
	Consider a smooth isotopy $(F_r)_{r \in [0,1]}$ transforming $\ug$ into $\wt{\ug}$, meaning
	$F_0 = \mathrm{id}_{\S^3}$ and $F_1(\gamma_k) = \wt{\gamma}_k$, and set $\gamma_{k,r}:=F_r(\gamma_k)$.
	We can assume that $\wt{S}_k=F_1(S_k)$ by gauge invariance.
	We first deal with the case of one knot $K=1$.

	\begin{figure}[!!ht]
		\begin{tikzpicture}
			\foreach \y in {0,3,6}
			{
				\filldraw[fill=gray!10,draw=black,thick] (0,\y) -- (6,\y) -- (8,\y+2) -- (2,\y+2) -- cycle;
				\draw[thick] (3,\y) -- (5,\y+2)
				node[sloped, pos=0.7, scale=0.8] {$>$};
				\draw [fill] (1,\y+1) circle [radius = 0.05];	
				\draw [fill] (4,\y+1) circle [radius = 0.05];	
				\draw [fill] (7,\y+1) circle [radius = 0.05];	
			}

			\node[left, scale=0.8] at (3.3,0.5) {$M_0$};
			\node[left, scale=0.8] at (3.8,4) {$M_{r_1}$};
			\node[left, scale=0.8] at (3.3,6.5) {$M_{r_2}$};

			\draw[thin,->] (2,2)--(2.3,2.3);
			\node[above right] at (2.2,2.2) {$\alpha_1$};
			
			\draw[thin,->] (6,0)--(6.5,0);
			\node[right] at (6.5,0) {$\alpha_2$};
			
			\draw[thin,->] (0,0)--(0,6.5);
			\node[above] at (0,6.5) {$r$};
			
			\draw [fill] (1,0) circle [radius = 0.05];
			\draw [fill] (4,0) circle [radius = 0.05];
			\draw [fill] (6,2) circle [radius = 0.05];
			\draw [fill] (3,2) circle [radius = 0.05];
			
			\draw [fill] (0,3) circle [radius = 0.05]; 
			\draw [fill] (3,3) circle [radius = 0.05]; 
			\draw [fill] (2,5) circle [radius = 0.05]; 
			\draw [fill] (5,5) circle [radius = 0.05]; 
			\draw [fill] (8,5) circle [radius = 0.05]; 
			
			\draw [fill] (4,8) circle [radius = 0.05]; 
			\draw [fill] (2,6) circle [radius = 0.05];
			\draw [fill] (7,8) circle [radius = 0.05]; 
			\draw [fill] (3,6) circle [radius = 0.05];  
			
			
			\draw[very thick] (3,2) -- (17/3-3,3);
			\draw[very thick, dashed] (17/3-3,3) -- (2,5);
			\draw[very thick]  (6,2) -- (17/3,3);
			\draw[very thick, dashed] (17/3,3) -- (5,5);
			\draw[very thick]  (5,5) -- (14/3,6);
			\draw[very thick, dashed] (14/3,6) -- (4,8);
			\draw[very thick]  (8,5) -- (29/4,29/4);
			\draw[very thick, dashed] (29/4,29/4) -- (7,8);

			\draw[very thick] (4,0) -- (2,6); 
			\draw[very thick] (1,0) -- (0,3); 
			\draw[very thick] (7,1) -- (7,7); 
			\foreach \x in {0,3}		       
			{
				\draw[very thick] (1+\x,1) -- (1+\x,3);  
				\draw[very thick, dashed] (1+\x,3) -- (1+\x,4);
				\draw[very thick] (1+\x,4) -- (1+\x,6);
				\draw[very thick, dashed] (1+\x,6) -- (1+\x,7);
			}
			
			\foreach \y in {0,3,6}
			{
				\draw (4.5,\y) -- (6.5,\y+2)
				node[sloped, pos=0.7, scale=0.8] {$>$};
			}
			
			\node[below right, scale=0.8] at (6.25,1.75){$L_0$};
			\node[below right, scale=0.8] at (6.25,4.75){$L_{r_1}$};
			\node[below right, scale=0.8] at (6.25,7.75){$L_{r_2}$};

		\end{tikzpicture}
		\caption{A $\sT_{W}$-homotopy in $\cup_{r \in [0,r_2]} \bbd_{\uS(r)}(\Tf^2)$. At each height $r$ we have drawn two
		copies of $\Tf^2$ in its universal covering $\R^2$. 
		The bold lines form the set of critical points: the oblique ones for the moving critical $(0,\alpha_2)$ 
		and the vertical ones for the fixed critical $(\alpha_1=\tfrac{1}{2},0)$.}
		\label{fig:t_homotopy}
	\end{figure}
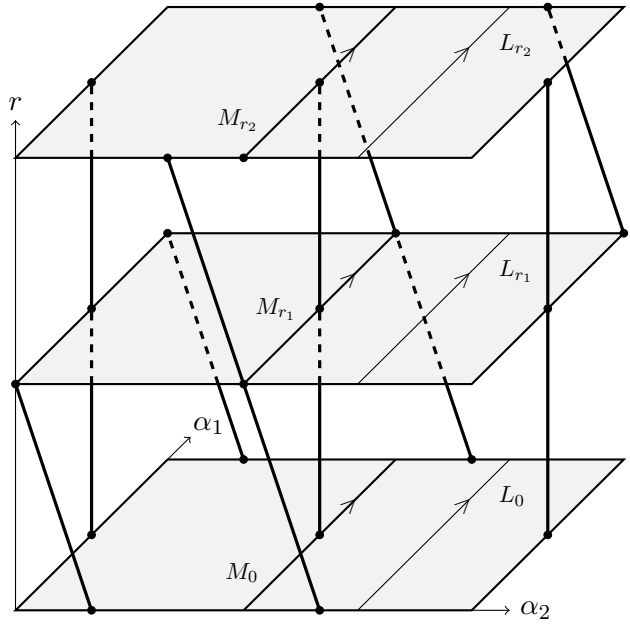
	To the one-parameter family $(\gamma_r)_{r \in [0,1]}$ we can associate a smooth one-parameter family of
	circles $(\cC_r)_{r \in [0,1]}$ in $\S^3$ such that $\link(\gamma_r,\cC_r) = 1$ for all $r \in [0,1]$.
	To each couple of knots $(\gamma_r, \cC_r) =: \Gamma_r$ we assign a couple of Seifert surfaces $\uS(r) \in \sS^{(2)}$
	in such a way that $(\uS(r))_{r \in [0,1]}$ is a smooth family.
	Recall that the critical points on $\partial\T_{\Gamma_r}$ are $(1/2,0)$ and $(0,\alpha_c(r))$, where $\alpha_c(r)$
	satisfies $\pi[1-\Wr(\gamma_r)] - 2\pi\alpha_c(r)\cdot 1 = 0 \!\!\mod 2\pi$ 
	(see Theorem~\ref{thm:bump_cont_bdry_A} and Proposition~\ref{prop:calcul_spectre}).

	We then consider the $\sT_{W}$-homotopy $(L_r)_{r \in [0,r_2]}$ in
	$$
		\bigcup_{r \in [0,r_2]} \bbd_{\uS(r)}(\T_{\Gamma_r}) \subset \bigcup_{r \in [0,r_2]} \bbd_{\uS(r)}(\Tf^2) \subset \Sd,
	$$
	drawn in Figure~\ref{fig:t_homotopy} (which corresponds to freezing the flux of the circle to some $0<\alpha_{\mathrm{aux}}<1$).
	Observe that the loop $M_r$ corresponds to the loop
	$$
		M_r = \big(\cD_{\bA(r, \alpha)}\big)_{\alpha \in [0,1]},
	$$
	where $\bA(r,\alpha) = 2\pi\alpha[S_1(r)]$. By Section~\ref{sec:cont_discont} 
	the loop $M_r$ is bump-continuous as long as $\tfrac{1}{2}[1-\Wr(\gamma_r)]\notin \Z$, and its spectral flow is equal to that of the loop
	associated to $\alpha\mapsto \bA(r,\alpha,\eps):=2\pi\alpha[S_1(r)]+2\pi\eps [S_2(r)]$ for $\eps>0$ small enough.
	From \eqref{eq:sf_prop_3} we immediately obtain that
	$$
		\Sf[L_0] = \Sf[L_{r_2}].
	$$
	Furthermore, $\Sf[L_{r_2}] = \Sf[M_{r_2}]$, yet by Theorem~\ref{thm:spec_flow_gen_l},
	$$
		\Sf[L_0] - \Sf[M_0] = -1,
	$$
	as $L_0-M_0$ is homologically equivalent to a circle around $(\tfrac{1}{2},0)$. We get
	$$
		\Sf[M_{r_2}] = \Sf[M_0] - 1.
	$$
	The jump in the spectral flow occurs precisely when $\tfrac{1}{2}[1-\Wr(\gamma_r)]$ crosses $\Z$, in Figure~\ref{fig:t_homotopy} at $r=r_1$.
	The result of the theorem then follows by successively iterating the above procedure, and the difference of the two spectral flows
	is minus the number of integer points between $\tfrac{1}{2}[1-\Wr(\gamma)]$ and $\tfrac{1}{2}[1-\Wr(\wt{\gamma})]$.
	
	\smallskip
	
	Now let us deal with the case of a link. We first deal with the case 
	$\ua^{(k_0)},\wt{\ua}^{(k_0)}\in (0,1)^K$.

	We can assume that there exists $k_1$ with $\link(\gamma_{k_1},\gamma_{k_0})\neq 0$, otherwise,
	the problem can be reduced to the one-knot case. 
	Indeed, following Section~\ref{sec:cont_discont}, we can lower $\ua^{(k_0)}$ down to $\underline{0}$ along $s\mapsto (1-s)\ua^{(k_0)}.$
	If $\link(\gamma_k,\gamma_{k_0})=0$ for all $k\neq k_0$, this would define a ($\Sd$-valued) bump continuous homotopy.

	As in the one knot case, we consider the smooth path $(\ug)_r$, 
	and we associate to it a smooth path $(\ua_r^{(k_0)})=(\alpha_{k,r}^{(k_0)})_{k\neq k_0}$ connecting $\ua^{(k_0)}$ to
	$\wt{\ua}^{(k_0)}$ within $(0,1)^{K-1}$. The two paths altogether define a (strong-resolvent continuous)-homotopy
	of Dirac operators $(\cD_{\bA(\alpha,r)})_{0\le \alpha,r\le 1}$ where:
	\begin{equation*}
	 \bA(\alpha,r):=2\pi\alpha[F_r(S_{k_0})]+\sum_{k\neq k_0}2\pi\alpha_{k,r}^{(k_0)}[F_r(S_k)].
	\end{equation*}

	This homotopy intersects the critical region $\sR_{\crit}^{(K)}$, (see \ref{eq:crit_regi}, Proposition~\ref{prop:calcul_spectre})
	whenever
	\begin{equation}\label{eq:critical_region_in_proof}
	 \alpha_{k_0}=0\ \&\ f(r):=\frac{1}{2}[1-\Wr(\gamma_{k_0,r})]-\sum_{k\neq k_0}\alpha_{k,r}\link(\gamma_k,\gamma_{k_0})\in \Z.
	\end{equation}
    By Sard's lemma, 
    $e^{-2i\pi\eps}$ is a regular value of $e^{2i\pi f(r)}$ along some sequence $\eps_n\to0^+$.
    So up to adding to $\alpha_{k_1,r}$ a term $\eps \chi(r)$  where $\chi\in C^\infty_c([0,1],[0,1])$ with
    $\chi\restriction (\eta,1-\eta)=1$ and $\chi=0$ around $0$ and $1$,
    we can assume that $f$ crosses $\Z$ finitely many times and that all crossings occur at regular points.

    Up to inserting parts 
    with only $\alpha_{k_1,r}$ varying, we can assume that the intersections occur when everything is frozen but $\alpha_{k_1}$ (after transformation the isotopy and the $\alpha_{k,r}$'s
    are continuous, piecewise $C^{\infty}$, with constant parts for $F_r$ and $\alpha_{k,r}$, $k\neq k_1$).

	\emph{Claim:} across the intersection the spectral flow $\Sf((\cD_{\bA(\alpha,r)})_{\alpha})$ jumps by $1$ (resp. $-1$)
	when the number $ \frac{1}{2}(1-\Wr[\gamma_{k_0,r}])-\sum_{k\neq k_0}\alpha_{k,r}\link(\gamma_k,\gamma_{k_0})$ 
	crosses the lattice $\Z$ \emph{positively} (resp. \emph{negatively}). Let $0<r_1<\cdots<r_M<1$ be the points of intersection. Let $I_0:=[0,r_1),I_M:=(r_{M},1]$ and $I_m:=(r_m,r_{m+1})$, for $1\le m\le M-1$.
	Each restriction $(\cD_{\bA(\alpha,r)})_{\alpha\in\Tf\atop r\in I_m}$ is bump-continuous. 
	For $r=r_m\in (0,1)$ the loop $(\cD_{\bA(\alpha,r_m)})_{\alpha\in\Tf}$ does not have a well-defined spectral flow.
	Furthermore, the spectral flow of the loop with $r=r_m-\eps$ differs from that for $r=r_m+\eps$ by $\pm 1$ since their difference 
	in $H^1(\mathrm{SF},\sT_{W})$
	is homologous to a loop in the bulk $\sS^{(K)}\times (0,1)^K$ and a loop circling the critical region $\sR_{\crit}^{(K)}$ as pictured in Figure~\ref{fig:homologous}.
	Adding up all the contributions, 
	we obtain -- under the assumption \eqref{eq:convention} -- that the difference
	\[
		\Sf\big( (\cD_{\bA(\alpha)})_{\alpha}\big)-\Sf\big( (\cD_{\wt{\bA}(\alpha)})_{\alpha}\big)
	\]
	is equal to the (non-negative) number of integers between $\frac{1}{2}(1-\Wr[\gamma_{k_0}])-\sum_{k\neq k_0}\alpha_{k}\link(\gamma_k,\gamma_{k_0})$
	and $\frac{1}{2}(1-\Wr[\wt{\gamma}_{k_0}])-\sum_{k\neq k_0}\wt{\alpha}_{k}\link(\gamma_k,\gamma_{k_0})$.

	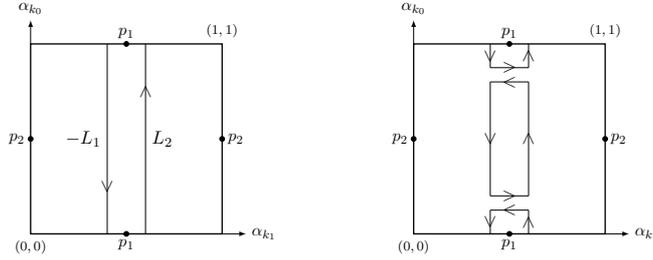
\begin{figure}[!!ht]
	\resizebox{0.7\textwidth}{!}{
	\begin{tikzpicture}
	\foreach \x in {0,8}{
		\draw[thick] ({0+\x},0) rectangle ({4+\x},4);
		\node[below] at ({0+\x},0) {\scriptsize $(0,0)$};
		\node[above] at ({4+\x},4) {\scriptsize $(1,1)$};
		\draw[>=latex, ->] ({4+\x},0) -- ({4.5+\x},0);
		\node[thin, right, scale=0.9] at ({4.5+\x},0) {$\alpha_{k_1}$};
		\draw[>=latex, ->] ({0+\x},4) -- ({0+\x},4.5);
		\node[thin, above, scale=0.9] at ({0+\x},4.5) {$\alpha_{k_0}$};
		
		\node[below, scale = 0.9] at ({2+\x},0) {$p_1$};
		\node[above, scale = 0.9] at ({2+\x},4) {$p_1$};
		\node[left, scale = 0.9] at ({0+\x},2) {$p_2$};
		\node[right, scale = 0.9] at ({4+\x},2) {$p_2$};
		\draw [fill] ({2+\x},4) circle [radius = 0.05];
		\draw [fill] ({4+\x},2) circle [radius = 0.05];
		\draw [fill] ({0+\x},2) circle [radius = 0.05];
		\draw [fill] ({2+\x},0) circle [radius = 0.05];
		}
		\draw (1.6,0)--(1.6,4)
			node[near start,sloped] {$<$}; 
		\node[left] at (1.6,2) {$-L_1$};
		\draw (2.4,0)--(2.4,4)
			node[near end,sloped] {$>$}; 
		\node[right] at (2.4,2) {$L_2$};

		\draw (9.6,0.8)--(9.6,3.2)
			node[midway,sloped] {$<$}; 
		\draw (10.4,0.8)--(10.4,3.2)
			node[midway,sloped] {$>$}; 
		\draw (9.6,0.8)--(10.4,0.8)
			node[midway,sloped] {$>$};
		\draw (9.6,3.2)--(10.4,3.2)
			node[midway,sloped] {$<$};
			
		\draw (10.4,0)--(10.4,0.5)
			node[midway,sloped] {$>$};
		\draw (10.4,0.5)--(9.6,0.5)
			node[midway,sloped] {$<$};
		\draw (9.6,0.5)--(9.6,0)
			node[midway,sloped] {$>$};
			
		\draw (9.6,4)--(9.6,3.5)
			node[midway,sloped] {$>$};
		\draw (9.6,3.5)--(10.4,3.5)
			node[midway,sloped] {$>$};
		\draw (10.4,3.5)--(10.4,4)
			node[midway,sloped] {$>$};
	\end{tikzpicture}
	}
	\caption{Homology of the difference of two parallel loops $L_1,L_2$ separated by the critical region.}
	\label{fig:homologous}
\end{figure}

	Now let us assume that some $\alpha_{k}$ or $\wt{\alpha}_k$ is $0$. By continuity from the right (Section~\ref{sec:cont_discont}), for $\eps>0$
	small enough, we can increase the vanishing $\alpha_{k}$ and $\wt{\alpha}_k$ from $0$ up to some $\eps>0$ so that
	\[
        \left\{
            \begin{array}{rcl}
                \bA(\alpha,s)&=&2\pi\alpha[S_{k_0}]+2\pi\sum_{k\neq k_0}\max(s,\alpha_k)[S_k],\\
                \wt{\bA}(\alpha,s)&=&2\pi\alpha[\wt{S}_{k_0}]+2\pi\sum_{k\neq k_0}\max(s,\wt{\alpha}_k)[\wt{S}_k]
            \end{array}
        \right.
	\]
    define bump-continuous homotopies $(\alpha,s)\in\Tf\times [0,\eps]\mapsto \cD_{\bA(\alpha,s)}$ resp. $\cD_{\wt{\bA}(\alpha,s)}$.
    Then we apply the previous proof to $(\cD_{\bA(\alpha,\eps)})_{\alpha\in\Tf}$ and $(\cD_{\wt{\bA}(\alpha,\eps)})_{\alpha\in\Tf}$.

\end{proof}

We now give two consequences to Theorem~\ref{thm:arbitrary_spectral_flow}, including the computation
of the spectral flow for the unknot.

\subsubsection{Spectral flow for an unknot}\label{subs:circle}
Let $\cC \subset \S^3$ be a circle, that is the intersection of $\S^3$ 
with a $2$-dimensional plane, excluding the case when this set is empty or just a point.
We then have the following result \cite[Theorem~25]{dirac_s3_paper1}.

\begin{theorem}\label{thm:no_zero_modes_for_circles}
	Let $\cC$ be an oriented circle in $\S^3$ with Seifert surface $S$, and set $\bA = 2\pi\alpha[S]$.
	For any flux $0 \leq \alpha < 1$ we have
	$$
		\ker\cD_{\bA}=\{0\}.
	$$
\end{theorem}

This theorem is extremely useful for the computation of the spectral flow, since it will provide
us with a (generally hard-to-obtain) reference point; setting $\bA(\alpha) := 2\pi\alpha [S]$, $\alpha \in \Tf$, 
$\partial S = \cC$, we immediately obtain
$$
	\Sf\big[\big(\cD_{\bA(\alpha)}\big)_{\alpha \in [0,1]}\big] = 0.
$$

We combine this result with Theorem~\ref{thm:arbitrary_spectral_flow}: 
if we take $\wt{\gamma}_{k_0}$ a circle, $\wt{\ua}^{(k_0)}=(\underline{0}^+)$, and $\gamma_{k_0}$
another realization of the unknot we obtain the following result.

\begin{corollary}\label{cor:sf_link_unknots}
	Let $\uS\in \sS^{(K)}$, where $\gamma_{k_0}$ is a realization of the unknot, $1\le k_0\le K$. 
	Let $\ua^{(k_0)}=(\alpha_k^{(k_0)})_{1\le k\le K\atop k\neq k_0}\in \Tf^{(K-1)}$.
	We assume that the number $\frac{1}{2}(1-\Wr[\gamma_{k_0}])-\tfrac{1}{2\pi}\Phi_{k_0}(\ua^{(k_0)})$ is not an integer where:
	 \[
	\tfrac{1}{2\pi} \Phi_{k_0}(\ua^{(k_0)})=\sum_{1\le k\le K\atop k\neq k_0}\alpha_k^{(k_0)}\link(\gamma_k,\gamma_{k_0}).
	 \]
	Let $(\bA(\alpha))_{0\le \alpha\le 1}$ be the family $\bA(\alpha):=2\pi\alpha[S_{k_0}]+\sum_{k\neq k_0}2\pi\alpha_k^{(k_0)}[S_k]$,
	then we have
	\[
		\Sf[(\cD_{\bA(\alpha)})_{0\le \alpha\le 1}] =\lfloor\tfrac{1}{2}(1-\Wr[\gamma_{k_0}]) - \tfrac{1}{2\pi}\Phi_{k_0}(\ua^{(k_0)})\rfloor.
	\]
\end{corollary}

This gives us the full description of the spectral flow for a link of unknots.
\begin{proof}
	Without loss of generality, we can assume $\alpha_k\neq 0$ for $k\neq k_0$(see Section~\ref{sec:cont_discont}).
	By assumption, there exists a smooth isotopy $(F_r)_{r\in [0,1]}$, $F_0=\mathrm{id}_{\S^3}$, transforming $\gamma_{k_0}$ into a circle $\wt{\gamma}_{k_0}:=F_1(\gamma_{k_0})$ 
	(the non-degenerate intersection of a $2$-dimensional plane with $\S^3$). It transforms $\gamma_k$, resp. $S_k$ into $\wt{\gamma}_k:=F_1(\gamma_{k})$ resp. $\wt{S}_k:=F_1(S_k)$. 
	The loop $\wt{\alpha}_{k_0}\mapsto \cD_{2\pi \wt{\alpha}_{k_0}[\wt{S}_{k_0}]}$ has trivial spectral flow by Theorem~\ref{thm:no_zero_modes_for_circles}. 
	
	To end the proof it suffices to apply Theorem~\ref{thm:arbitrary_spectral_flow}
	to $(\uS,\ua^{(k_0)})$ and $(\wt{\uS},\wt{\ua}^{(k_0)})$ with $\wt{\ua}^{(k_0)}\equiv \underline{0}\in \Tf^{(K-1)}$.
	
\end{proof}

\subsubsection{Lower bound on zero modes}
\begin{rem}
	Let $\gamma$ be a knot. In Section~\ref{sec:techn}, we show that its writhe is equal to
	 $-(2\pi)^{-1}$ times the integrated relative torsion of any Seifert frame.
	 Furthermore in Section~\ref{sec:varying_I_tau}, we show that
	 we can smoothly vary its integrated relative torsion without changing its isotopy class.
\end{rem}

Then using Theorem~\ref{thm:arbitrary_spectral_flow} and the above remark, we obtain.

\begin{corollary}[Lower bound on the total number of zero modes]\label{coro:lower_bound}
	Let $m\in\Z$ and let $\gamma_0\subset \S^3$ be a knot. There exists a knot $\gamma\subset \S^3$
	isotopic to $\gamma_0$ with $-\pi\Wr(\gamma)\neq \pi\!\!\mod 2\pi$ such that the 
	following holds. If we pick a Seifert surface $S$ for $\gamma$ and set $\bA(\alpha)=2\pi\alpha[S]$, 
	then we have:
	\[
	\Sf\big( (\cD_{\bA(\alpha)})_{0\le \alpha\le 1}\big)=m.
	\]
	In particular, the total number of zero modes
	$$
		N(\gamma):=\sum_{\alpha\in (0,1)}\dim\,\ker\big(\cD_{\bA(\alpha)}\big)
	$$
	along the path $\big(\cD_{\bA(\alpha)}\big)_{\alpha\in [0,1]}$ is equal or greater than $|m|$.
\end{corollary}

\begin{rem}[Open question]
	We compute in Corollary~\ref{cor:sf_link_unknots} the spectral flow corresponding to the path of fluxes $\alpha\in [0,1]\to 2\pi \alpha$ 
	when $\gamma$ is an an unknot. The case of a general knot remains open.
\end{rem}

We now illustrate these results with several examples for which we know a more detailed description of the zero modes.

\subsection{Hopf links}

\subsubsection{The Erd\H{o}s-Solovej construction}
We quickly recall the construction in \cite{MR1860416}, extended to the singular case in
\cite[Theorem~31]{dirac_s3_paper1}.

Seeing $\S^2 \subset \R^3$ with the metric $\tfrac{1}{4}g_2$, 
where $g_{2}$ is the induced metric from its ambient space, 
we define the Hopf map as
\begin{equation*}
	\Hopf:
	\begin{array}{ccl}
		(\S^3,g_3) &\longrightarrow& (\S^2,\tfrac{1}{4}g_2)\\
		(z_0,z_1) &\mapsto& (|z_0|^2-|z_1|^2,\Re(2z_0\overline{z}_1),\Im(2z_0\overline{z}_1)).
	\end{array}
\end{equation*}

The pre-image of $K$ points $v_1, \dots, v_K \in \S^2$ is given by $K$ interlinking 
circles on $\S^3$, $\gamma^K = \Hopf^{-1}(\{v_1,\dots, v_k\})$, with
$\link(\gamma_i,\gamma_j) = 1$ for any $i \neq j$.
The link $\gamma^K$ is then oriented along the vector field $u_3 = (iz_0,iz_1)$
and we define the \textit{magnetic Hopf link} $\bB^K$ as
\begin{equation}\label{def:magn_hopf_link}
	\bB^K := \sum_{k=1}^K 2\pi \alpha_k \big[\Hopf^{-1}(\{v_k\})\big],
\end{equation}
with (renormalized) flux $\alpha_k \in [0,1)$ on each component $\gamma_k$.

\begin{theorem}\label{thm:erdos-solovej}
	Let $\bB^K$ be a magnetic Hopf link as in \eqref{def:magn_hopf_link} with corresponding singular gauge
	$\bA$ and assume $\alpha_k \in (0,1)$ for all $k$. 
	Let $\cD_{\bA}$ be the Dirac operator on $(\S^3,g_3)$ 
	and define $c \in (-1/2,1/2]$, $m \in \Z$ such that
	$$
		\sum_{k=1}^K \alpha_k =: c + m.
	$$
	Furthermore, set 
	$$
		\beta_{\S^2,\ell} := \left(\,\sum_{k=1}^K 8\pi\alpha_k \delta_{v_k} - 2(c+\ell)\right) \frac{\vol_{g_2}}{4}.
	$$
	As $(2\pi)^{-1}\int_{\S^2} \beta_{\S^2,\ell} = m-\ell$, on the spinor bundle $\Psi_{m-\ell}$ (with Chern number $m-\ell$) there exists a
	two-dimensional Dirac operator $\cD_{\S^2,\ell}$ with magnetic two-form $\beta_{\S^2,\ell}$ (see Remark~\ref{rem:def_D_S_2}). 
	Then:
	\begin{enumerate}
		\item The spectrum of $\cD_{\bA}$ is given by
		$$
			\spec \cD_{\bA} = \bigcup_{\ell \in \Z} \left( \mathcal{Z}_\ell  
			\cup \left\{\pm \sqrt{\lambda^2 + (\ell+c)^2} - \frac{1}{2} : \lambda \in \spec_+\cD_{\S^2,\ell}\right\}\right),
		$$
		where
		\begin{align*}
			\mathcal{Z}_\ell = \left\{\begin{array}{ll}
 				\{\ell+c-1/2\}, 
				&\quad m>\ell,\\
				\emptyset,
				&\quad m=\ell,\\
				\{-\ell-c-1/2\},
				&\quad m<\ell.
				\end{array} \right.
		\end{align*}
		
		\item The multiplicity of an eigenvalue equals the number of ways it can be written as 
		$\sqrt{\lambda^2 + (\ell+c)^2} - \frac{1}{2}$, $\ell \in \Z$ and $\lambda \in \spec_+\cD_{\S^2,\ell}$ counted with
		multiplicity, or as an element in $\mathcal{Z}_\ell$ counted with multiplicity $|m-\ell|$.
		
		\item The eigenspace of $\cD_{\bA}$ with eigenvalue in $\mathcal{Z}_\ell$ contains spinors with definite
		spin value $\mathrm{sgn}(m-\ell)$.

	\end{enumerate}
\end{theorem}
\begin{rem}\label{rem:def_D_S_2}[The operator $\cD_{\S^2,\ell}$]
	Let us precise the definition of $\cD_{\S^2,\ell}$. Up to using the stereographic projection,
	it is enough to define the operator in $\C\hat{=}\R^2$ seen as a chart of $\S^2\setminus\{\mbox{North pole}\}$,
	and in which lie all the A.B. solenoids (then it suffices to use \cite[Appendix~A]{MR1860416}). Seeing $v_k\in \C$,
	and $\tfrac{v-v_k}{|v-v_k|}=:e^{i\theta_k}$ for $v=x+iy\in\C$, we use the scalar gauge of $\beta_{\S^2,\ell}$ given by
	\[
		\boldsymbol{\alpha}_{\ell}=\sum_{k=1}^K\alpha_k\d\theta_k-\frac{c+\ell}{4+|v|^2}(x\d y-y\d x). 
	\]
	In the vicinity of each A.B. solenoid $2\pi\alpha_k\delta_{v_k}$, the spinors in $\dom(\cD_{\S^2,\ell})$
	satisfy the boundary condition of the operator $\cD^{\tau}$ of \cite[Section~2]{Persson06} with $\tau=\pi$ (we refer to \cite[Section~4]{MR1860416}
	for the correspondence between the Dirac operator in the flat metric and that in the $\S^2$-metric).
	That is, writing $\psi_{\pm}$ the spin up and spin down components of $\psi\in \dom(\cD_{\S^2,\ell})$ we have
	 $\psi_+\in H_{\mathrm{loc}}^1(\C)$ and
	\[
	\lim_{v\to v_k\atop v\neq v_k}|v-v_k|^{\alpha_k}\psi_-(v)\mbox{ exists}.
	\]
\end{rem}

In the special case when $c=1/2$, we have precise information about $\ker \cD_{\bA}$.
\begin{corollary}\label{cor:dim_ker_hopf_link}
	If the magnetic Hopf link $\bB^K$ has fluxes $\ua \in (0,1)^K$ with $\sum_{k=1}^K\alpha_k = m+1/2$,
	then
	$$
		\dim \ker \cD_{\bA} = m.
	$$ 
\end{corollary}
Armed with the above, we can to a large extent describe the spectral flow of the operator $\cD_{\bA}$ on
the torus of fluxes, and we begin by a simple example.

\subsubsection{Spectral flow for a magnetic Hopf $2$-link}
We freeze two Seifert surfaces $\uS = (S_1,S_2)$ for $\gamma = \Hopf^{-1}(\{v_1,v_2\})$;
we have that $\tfrac{1}{2}\int_{\gamma_k}\tau_{S_k'}\bT^{\flat} = 0$, since the 
Hopf circles are great circles in $\S^3$.
According to Theorem~\ref{thm:bump_cont_bdry_A} and Proposition~\ref{prop:calcul_spectre},
the map $\bbd(\uS,\cdot)$ is only discontinuous at $p_1 = (\tfrac{1}{2},1)$ and $p_2 = (1,\tfrac{1}{2})$ in $\Tf^2$.
Thus we have 
$$
	\T_{\ug}=\Tf^2 \setminus \{p_1,p_2\}
$$
and the group homomorphism \eqref{eq:spec_flow_at_fixed_seif} is:
$$
	\Sf_{\uS} := \Sf \circ (\bbd_{\uS})_{*} : \pi_1(\T_{\ug}) \to \Z.
$$
The general picture for the spectral flow is illustrated in Figure~\ref{fig:punctured_torus}.
\begin{figure}[!ht]
	\resizebox{0.5\textwidth}{!}{
	\begin{tikzpicture}
		\draw[thick] (0,0) rectangle (4,4);
		\node[below] at (0,0) {\scriptsize $(0,0)$};
		\node[above] at (4,4) {\scriptsize $(1,1)$};
		\draw[>=latex, ->] (4,0) -- (4.5,0);
		\node[thin, right, scale=0.9] at (4.5,0) {$\alpha_1$};
		\draw[>=latex, ->] (0,4) -- (0,4.5);
		\node[thin, above, scale=0.9] at (0,4.5) {$\alpha_2$};
		
		\node[above] at (1.5,0.35) {\scriptsize $L_1$};
		\node[right] at (0.35,1.5) {\scriptsize $L_2$};
		\node[below] at (1.72,3.8) {\scriptsize $L_3$};
		
		\node[below, scale = 0.9] at (2,0) {$p_1$};
		\node[above, scale = 0.9] at (2,4) {$p_1$};
		\node[left, scale = 0.9] at (0,2) {$p_2$};
		\node[right, scale = 0.9] at (4,2) {$p_2$};
		\draw [fill] (2,4) circle [radius = 0.05];
		\draw [fill] (4,2) circle [radius = 0.05];
		\draw [fill] (0,2) circle [radius = 0.05];
		\draw [fill] (2,0) circle [radius = 0.05];
		
		\draw[very thick] (2,4) -- (4,2)
		node[sloped, midway, above] {\scriptsize $\alpha_1+\alpha_2 = \tfrac{3}{2}$};
		\draw[very thin, dashed] (1.5,2) circle [radius = 0.5];
		\draw[very thin, dashed] (3.5,1) circle [radius = 0.3];
		
		\draw (0.4,0) -- (0.4,4)
		node[pos=0.7, scale=0.6] {$\wedge$};
		\draw(0,0.4) -- (4,0.4)
		node[pos=0.7, scale=0.6] {$>$};
		\draw(1.7,4) arc (180:360:0.3)
		node[sloped, midway, scale=0.6, rotate=5] {$>$};
		\draw (2.3,0) arc (0:180:0.3)
		node[sloped, midway, scale=0.6, rotate=5] {$<$};
	\end{tikzpicture}
	}
	\caption{The punctured torus $\T_{\ug}$ for the Hopf 2-link.}
	\label{fig:punctured_torus}
\end{figure}

The family $(L_1,L_2,L_3)$ generates the first homology group $H_1(\T_{\ug})$.
The bold segment $[p_1,p_2]$ of equation $\alpha_1+\alpha_2=3/2$ corresponds to fluxes
for which we know 
that $\dim \ker\cD_{\bA}=1$, see Corollary~\ref{cor:dim_ker_hopf_link}.

Recall that $\wt{\Sf}_{\uS}\circ \textrm{ab}$ is the factorization of $\Sf_{\uS}$.
We have:
\[
  	\left\{
		\begin{array}{rcl}
			\wt{\Sf}_{\uS}(L_1)&=&0,\\
			\wt{\Sf}_{\uS}(L_2)&=&0,\\
			\wt{\Sf}_{\uS}(L_3)&=&-1.
		\end{array}
	\right.
\]
Indeed, the result $\wt{\Sf}_{\uS}(L_3) = -1$ follows directly from Theorem~\ref{thm:spec_flow_gen_l}
and Corollary~\ref{cor:sf_link_unknots} covers the other two.

\begin{rem}
	Note that the computation of the spectral flow on $\pi_1$ (or the homology group $H_1$) 
	does not enable us to exclude the occurence of zero modes outside the range $\alpha_1+\alpha_2=\tfrac{3}{2}$.
	Indeed if there existed closed curves in $\Tf^2$ in which all points $(\alpha_1,\alpha_2)$
	gave rise to Dirac operators with non-trivial kernel, it would not be detected by the spectral flow
	(these potential curves are represented by dashed circles in Figure~\ref{fig:punctured_torus}).
	
	The existence of these curves could however be disproved
	if one could determine the positive spectrum of the operators $\cD_{\S^2,\ell}$
	(and thus the full spectrum of $\cD_{\bA}$).
\end{rem}

\subsubsection{Spectral flow for a magnetic Hopf $3$-link}

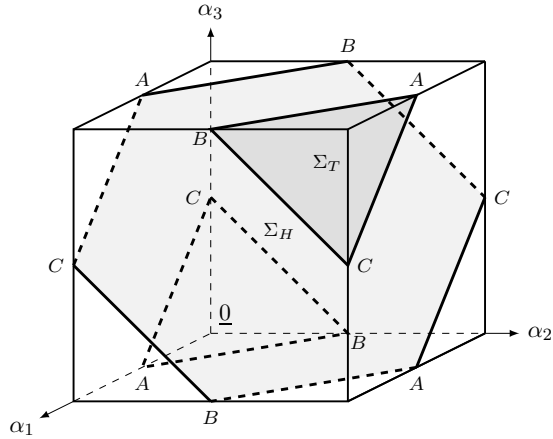
\begin{figure}[!ht]
	\resizebox{0.6\textwidth}{!}{
	\begin{tikzpicture}
		\fill[color=gray!10]
		(2,0) -- (0,2) -- (1,4.5) -- (4,5) -- (6,3) -- (5,0.5) -- cycle; 
		
		\filldraw[fill=gray!25, draw=black,very thick]
		(2,4) -- (5,4.5) -- (4,2) -- cycle;
		
		\draw[thick] (0,0) rectangle (4,4);
		\draw[dashed] (0,0) -- (2,1);
		\draw[dashed] (2,1) -- (6,1);
		\draw[dashed] (2,1) -- (2,5);
		
		\draw[thick] (0,4) -- (2,5);
		\draw[thick] (4,0) -- (6,1);
		\draw[thick] (4,0) -- (6,1);
		\draw[thick] (4,4) -- (6,5);
		\draw[thick] (6,1) -- (6,5);
		\draw[thick] (2,5) -- (6,5);
		
		\draw[>=latex, ->] (0,0) -- (-0.5,-0.25);
		\node[thin, scale=0.9] at (-0.75,-0.4) {$\alpha_1$};
		\draw[>=latex, ->] (6,1) -- (6.5,1);
		\node[thin, right, scale=0.9] at (6.5,1) {$\alpha_2$};
		\draw[>=latex, ->] (2,5) -- (2,5.5);
		\node[thin, above, scale=0.9] at (2,5.5) {$\alpha_3$};
		\node[thin, above right, scale=0.9] at (2,1) {$\underline{0}$};
		
		\draw[very thick] (0,2) -- (2,0);
		\draw[very thick, dashed] (0,2) -- (1,4.5);
		\draw[very thick] (1,4.5) -- (4,5);
		\draw[very thick, dashed] (4,5) -- (6,3);
		\draw[very thick] (6,3) -- (5,0.5);
		\draw[very thick, dashed] (5,0.5) -- (2,0);
		
		\node  at (3,2.5) {\scriptsize $\Sigma_H$};
		
		\draw[very thick, dashed] (2,3) -- (4,1) -- (1,0.5) -- cycle;
		\node  at (3.7,3.5) {\scriptsize $\Sigma_T$};
		
		\node[above] at (1,4.5) {\scriptsize $A$};
		\node[above] at (4,5) {\scriptsize $B$};
		\node[right] at (6,3) {\scriptsize $C$};
		
		\node[below] at (5,0.5) {\scriptsize $A$};
		\node[below] at (2,0) {\scriptsize $B$};
		\node[left] at (0,2) {\scriptsize $C$};
		
		\node[below] at (1,0.5) {\scriptsize $A$};
		\node[below right] at (3.9,1.1) {\scriptsize $B$};
		\node[left] at (2,3) {\scriptsize $C$};
		
		\node[above] at (5,4.5) {\scriptsize $A$};
		\node[below left] at (2.1,4.1) {\scriptsize $B$};
		\node[right] at (4,2) {\scriptsize $C$};
	\end{tikzpicture}
	}
	\caption{The cut torus $\T_{\ug}$ for the Hopf 3-link.}
	\label{fig:sf_cube}
\end{figure}

We again freeze the triplet $\uS=(S_1,S_2,S_3)$ for $\gamma = \Hopf^{-1}(\{v_1,v_2,v_3\})$.
Since $\tfrac{1}{2}\int_{\gamma_k}\tau_{S_k}\bT^{\flat} = 0$, 
the map $\bbd(\uS,\cdot)$ is only discontinuous on the loops
$$
	\alpha_i = 1, \quad \alpha_j + \alpha_k = \frac{1}{2}\!\!\!\! \mod 1, \quad \{i,j,k\} = \{1,2,3\},
$$
which are represented by the different (dashed or non-dashed) segments between the points $A,B,C$ in Figure~\ref{fig:sf_cube}.
The shaded surfaces $\Sigma_H, \Sigma_T$ correspond to fluxes for which the associated Dirac operator has non-trivial kernel:
$$
	\Sigma_H: \left\{ \begin{array}{l}
		\alpha_1+\alpha_2+\alpha_3 = 3/2\\
		0<\alpha_1,\alpha_2,\alpha_3<1
		\end{array}\right.,
	\quad\Sigma_T:\left\{ \begin{array}{l}
		\alpha_1+\alpha_2+\alpha_3 = 5/2\\
		0<\alpha_1,\alpha_2,\alpha_3<1
		\end{array}\right..
$$
According to Corollary~\ref{cor:dim_ker_hopf_link}, for $\ua \in \Sigma_H$ the kernel is one-dimensional, 
whereas for $\ua \in \Sigma_T$ it is two-dimensional.

To discuss the homotopy of $\T_{\ug}$, it is easier to analyze the situation in the universal covering $\R^3$ 
of $\Tf^3$.
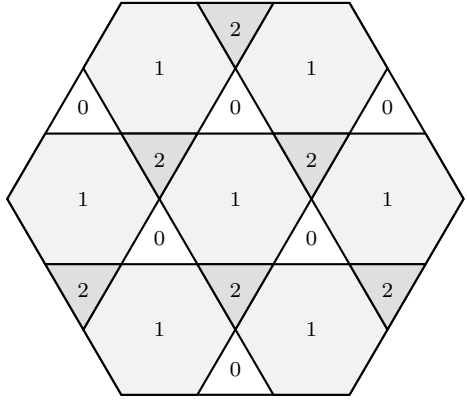
\begin{figure}[!ht]
	\begin{tikzpicture}
		\draw[thick] (-1,0) -- (0.5,3*sin{60}) -- (3.5,3*sin{60}) -- (5,0) -- (3.5,-3*sin{60}) -- (0.5,-3*sin{60}) -- cycle;
		
		\foreach \x in {0,2,4}
		{
			\filldraw[fill=gray!10, draw=black, thick]
			(\x+1,0) -- (\x+0.5,0+sin{60}) -- (\x-0.5,0+sin{60}) -- (\x-1,0+sin{0.0}) -- (\x-0.5,0-sin{60}) -- (\x+0.5,0-sin{60})-- cycle;
			\node at (\x,0) {\scriptsize $1$};
		}
		\foreach \x in {1,3}
		{
			\foreach \y in {-2*sin{60},2*sin{60}}
			{
				\filldraw[fill=gray!10, draw=black, thick]
				(\x+1,\y) -- (\x+0.5,\y+sin{60}) -- (\x-0.5,\y+sin{60}) 
				-- (\x-1,\y+sin{0.0}) -- (\x-0.5,\y-sin{60}) -- (\x+0.5,\y-sin{60})-- cycle;
				\node at (\x,\y) {\scriptsize $1$};
			}
		}
		
		\foreach \x in {0,2,4}
			\node at (\x,1.4*sin{60}) {\scriptsize $0$};
		\foreach \x in {1,3}
			\node at (\x,-0.6*sin{60}) {\scriptsize $0$};
		\node at (2,-2.6*sin{60}) {\scriptsize $0$};
		
		\filldraw[fill=gray!25, draw=black, thick]
		(2,2*sin{60}) -- (2+1/2, 2*sin{60}+sin{60}) -- (2-1/2,2*sin{60}+sin{60}) -- cycle;
		\node at (2,2*sin{60} + 0.6*sin{60}) {\scriptsize $2$};
		
		\foreach \x in {1,3}
		{
			\filldraw[fill=gray!25, draw=black, thick]
			(\x,0) -- (\x+1/2,	0+sin{60}) -- (\x-1/2,0+sin{60}) -- cycle;
			\node at (\x,0.6*sin{60}) {\scriptsize $2$};
		}
		\foreach \x in {0,2,4}
		{
			\filldraw[fill=gray!25, draw=black, thick]
			(\x,-2*sin{60}) -- (\x+1/2, -2*sin{60}+sin{60}) -- (\x-1/2,-2*sin{60}+sin{60}) -- cycle;
			\node at (\x,-2*sin{60} + 0.6*sin{60}) {\scriptsize $2$};
		}
	\end{tikzpicture}
	\caption{The tiling of the surface $\alpha_1+\alpha_2+\alpha_3=1/2$ in the universal covering. The number in
	each tile denotes the number of zero modes of the corresponding Dirac operator for such an 
	$\ua = (\alpha_1,\alpha_2,\alpha_3)$.}
	\label{fig:sf_david}
\end{figure}

As in the case of the Hopf 2-link, for $\eps>0$ small enough and $t \in \Tf$, the loops
$$
	\left\{
	\begin{array}{l}
		L_1: t \mapsto (t,\eps,\eps),\\
		L_2: t \mapsto (\eps,t,\eps),\\
		L_3: t \mapsto (\eps,\eps,t),
	\end{array}
	\right.
$$
have trivial spectral flow.

Any loop enclosing an edge in the tiling in
Figure~\ref{fig:sf_david} is homotopic to a loop discussed in Theorem~\ref{thm:spec_flow_gen_l}.
A case which is not directly covered is when the loop encircles a vertex, for example when going from a $0$-tile to
a $2$-tile as in Figure~\ref{fig:special_loop}.
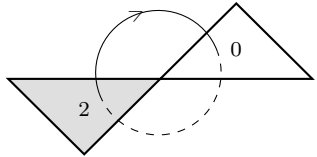
\begin{figure}[!ht]
	\begin{tikzpicture}
		\filldraw[fill=gray!25, draw=black, thick]
		(0,0) -- (-2,0) -- (-1,-1) -- cycle;
		\node at (-1,-0.4) {\scriptsize $2$};
		\draw[thick] (0,0) -- (2,0) -- (1,1) -- cycle;
		\node at (1,0.4) {\scriptsize $0$};
		
		\draw (-0.8,-0.2) arc (200:20:0.825)
		node[sloped, midway, scale=0.7, rotate=5] {$>$};
		\draw[dashed] (-0.8,-0.2) arc (-160:20:0.825);
	\end{tikzpicture}
	\caption{A loop in the universal covering encircling a vertex.}
	\label{fig:special_loop}
\end{figure}
Such a loop is however homologous to two loops encircling an edge;
both tiles share two $1$-tiles on their boundary and the loop
can be decomposed into the loop encircling the edge bordering the $2$-tile
and one of the $1$-tile and the loop encircling the edge shared by this same $1$-tile and the $0$-tile.
This leaves us with $6$ loops $L_4, \dots, L_9$, one for each side of the hexagon, 
each of them 
circling an edge so that they intersect the surface $\sum_{i=1}^3\alpha_i=\tfrac{1}{2}$ -- oriented by $(1,1,1)$ -- 
exactly twice and transversally: positively on the $(n+1)$-tile and negatively on the $n$-tile with $n=0$ or $n=1$.
Summarizing, we have
$$
	\wt{\Sf}_{\uS}(L_1) = \ldots = \wt{\Sf}_{\uS}(L_3) = 0, \quad \wt{\Sf}_{\uS}(L_4) = \ldots = \wt{\Sf}_{\uS}(L_9) = -1,
$$
and $L_1,\ldots, L_9$ generate $H_1(\T_{\ug})$.

\subsubsection*{\bf{Conjecture}}
	In the case of the Hopf 2-link (Figure~\ref{fig:punctured_torus}),
	the two critical points $p_1,p_2$ are linked by a line of solutions to the zero mode problem, namely
	$\alpha_1+\alpha_2=\tfrac{3}{2}$. If we now deform for example the second Hopf circle and increase 
	its integrated relative torsion, then the critical point $p_1$ moves to the left. 
	We conjecture that the line of solutions remains, and having this picture in mind,
	the jump of the spectral flow can be easily explained: it corresponds to the 
	crossing of the line across the vertical line $\alpha_1=0\in \Tf$.
	
	More generally, we conjecture that for any system of two linked knots, this line of solutions linking the two critical points
	exist, and that it is a $C^1$-curve. Furthermore, adding other knots in the picture, we think that the critical region bounds
	hypersurfaces of solutions when the knots have non-trivial linking numbers
	(see Figure~\ref{fig:sf_cube} for the Hopf 3-link in which the critical region bounds the triangle $\Sigma_{T}$ and
	the hexagon $\Sigma_{H}$). There may be degenerate cases for specific values of the integrated relative torsions.

\section{Proofs of Theorems~\ref{thm:bump_cont_bulk},~\ref{thm:bump_cont_bdry_A}~\&~\ref{thm:bump_cont_bdry_B}}\label{sec:bump_cont_proofs}
The following Lemma from \cite[Lemma~23]{dirac_s3_paper1} will be useful when proving the main theorems below.
\begin{lemma}\label{lem:char_sres_conv}
	Let $(\cD_n)_n$ be a sequence of (unbounded) self-adjoint operators on a separable 
	Hilbert space $\cH$. 
	Then the following propositions are equivalent.
	\begin{enumerate}
		\item $\cD_n$ converges to $\cD$ in the strong resolvent sense.
    	
		\item For any $(f,\cD f)\in \cG_{\cD}$, 
    		there exists a sequence $(f_n,\cD_n f_n)\in \cG_{\cD_n}$ converging to 
		$(f,\cD f)$ in $\cH \times \cH$.
    	
		\item The orthogonal projection $P_n$ onto $\cG_{\cD_n}$ 
		converges in the strong operator topology to $P$, the orthogonal projector onto $\cG_{\cD}$.
		\end{enumerate}
\end{lemma}

\subsection{Proof of Theorem~\ref{thm:bump_cont_bulk}}
	As $\sS^{(K)}\times \Tf^{K}$ is metric, it suffices to verify sequential continuity.
	Let $(\uS,\ua)\in  \sS^{(K)}\times (0,1)^{K}$, and 
	let $(\uS^{(n)},\ua^{(n)})_n$ be a sequence converging to $(\uS,\ua)$.
	We set 
	$$
		M:=\rank\,\mathds{1}_{\supp\,\phi}(\cD_{\bA}).
	$$
	Due to the strong resolvent continuity of $\bbd$ 
	(characterization (2) of Lem.~\ref{lem:char_sres_conv}), 
	there exists a sequence of $M$-dimensional planes $W^{(n)}\subset\ran\,\cD_{\bA^{(n)}}$ such that the projection
	$P_{W^{(n)}}$ converges to $\mathds{1}_{\supp\,\phi}(\cD_{\bA})$. As the operators have discrete spectrum, 
	we can assume that $W^{(n)}$ is generated by eigenfunctions of $\cD_{\bA^{(n)}}$. 
	By Theorem~\ref{thm:compactness}, any sequence $(\psi^{(n)})_n$ with 
	$$
		\psi^{(n)} \in \mathds{1}_{\supp\,\phi}(\cD_{\bA^{(n)}}), \quad \norm{\psi^{(n)}}_{L^2}=1, 
	$$
	converges (up to extraction) to an element
	in $\ran \mathds{1}_{\supp\,\phi}(\cD_{\bA})$.
	This shows that along the sequence $(\uS^{(n)},\ua^{(n)})_n$ 
	there is no vanishing of eigenfunctions in the spectral region $(\supp\,\phi)^\circ$, and thus
	$$
		\lim_{n\to+\infty}\norm{\phi(\cD_{\bA^{(n)}})-\phi(\cD_{\bA})}_{\cB}=0.
	$$
\qed

\subsection{Proof of Theorems~\ref{thm:bump_cont_bdry_A}~\&~\ref{thm:bump_cont_bdry_B}}

\subsubsection{Description of the proof of Theorem~\ref{thm:bump_cont_bdry_A}}
We split the proof into several steps. In the first three steps we
recall some technical tools from \cite{dirac_s3_paper1} used to
analyze the operator $\cD_{\bA}$ and its domain in the vicinity of a
component $\gamma_k$ for a given magnetic link represented by
$(\uS,\ua)\in \sS^{(K)}\times(0,1)^K$.

{\bf Step 1 \ref{sec:remov_phase}:} Using a local gauge transformation 
in a tubular neighborhood around $\gamma_k$ (see Proposition~\ref{prop:sa_link}
below) we replace the singular vector potentials coming from all
the other knots by a smooth vector potential. 

{\bf Step 2 \ref{sec:localization}:} We introduce a partition of unity
enabling us to localize wave functions around $\gamma_k$.

{\bf Step 3 \ref{sec:model_case}:} In the vicinity of a $\gamma_k$ we
express the Dirac operator in the local coordinates $(s,\rho,\theta)$ and the spinor 
basis $(\xi_+,\xi_-)$, defined in
Section~\ref{sec:seif_fram_loc_coord}.  Up to sub-principal terms
depending on the geometry of $\gamma_k$ and the smooth vector
potential from Step 1 the Dirac operator acts like a cylindrical Dirac
operator $\cD_{\T_{\ell_k},\alpha_k}$ which has a phase jump of
$e^{-2i\pi\alpha_k}$ across $\theta=0$. In fact, as there is a small
discrepancy between the surface $\theta=0$ and the Seifert surface
$S_k$ we need a singular gauge transformation to move the phase jump
from $S_k$ to $\theta=0$. We have an explicit orthogonal decomposition of the domain of 
$\cD_{\T_{\ell_k},\alpha_k}$ in a regular and singular part (see \eqref{eq:form_domain_D_T}).

{\bf Step 4 \ref{sec:start_proof}:} This is the main part of the proof. We prove both
implications in Theorem~\ref{thm:bump_cont_bdry_A} 
by contraposition, i.e., we prove the equivalence of the statements  
\begin{itemize}
		\item[(S1)] For any bump-functions $\phi$ centered at
                  $\lambda$, the map $\bbd$ \eqref{eq:phi_no_top}
                  is not locally $\sT_{\phi}$-continuous at $(\uS,\ua)$.
\item[(S3)] There exists a $k \in \{1,\ldots,K\}$ such that $\alpha_k=0$ and
			$$ \lambda \in \spec \left(\cT_{k,\bA} +
  \tau_{S_k}\right).
			$$
	\end{itemize}
We show the equivalence of theses two statements 
by showing that they are both equivalent to
\begin{itemize}
	\item[(S2)] There exists a sequence $(\uS^{(n)},\ua^{(n)}) \to
          (\uS,\ua)$ for which there exists a sequence
          $(\psi^{(n)})_n$ with
		\begin{equation}\label{def:w-l_approx_seq}
			\left\{
			\begin{array}{l}
			\psi^{(n)} \in \dom\big(\cD_{\bA^{(n)}}\big),
                        \quad \norm{\psi^{(n)}}_{L^2} = 1,\\ \lim_{n
                          \to
                          \infty}\int_{\S^3}\big|\big(\cD_{\bA^{(n)}}-\lambda\big)\psi^{(n)}\big|^2
                        = 0,\\ \psi^{(n)} \rightharpoonup 0.
			\end{array}
			\right.
		\end{equation}
		We say that there exists a vanishing sequence
                $(\psi^{(n)})_n$ (at $\lambda$) if the above holds.
\end{itemize}
The most difficult part is to show that (S2) is equivalent to (S3). By
Theorem~\ref{thm:compactness} a vanishing sequence must concentrate
onto knots $\gamma_k$ for which $\alpha_k\to1^-$.  We localize the
vanishing sequence $(\psi^{(n)})_n$ around each of these $\gamma_k$'s
and express (S2) in local coordinates.  We decompose the localized
functions into regular and singular parts. To prove that (S2) implies
(S3) we show that the singular part of the vanishing sequence
collapses to an eigenfunction of $\cT_{k,\bA} +
  \tau_{S_k}$.
We show that (S3) implies (S2) by explicitly constructing a vanishing
sequence. 

\begin{notation}
	Along the proof we will consider converging sequences of
        tuples $(\uS^{(n)},\ua^{(n)})$. After localization we will
        concentrate our study around $\gamma_k^{(n)}$ for a given
        $k$. To avoid the burden of notation the index $k$ will often
        be omitted.
	
	For $\ell>0$, we will use the notations
        $\T_{\ell}:\R/(\ell\Z)$ and
        $\T_{\ell}^*:=\frac{2\pi}{\ell}\Z$.  As an example: in the
        proof and after localization, $\T_n$ denotes
        $\T_{\ell_k^{(n)}}$ where $\ell_k^{(n)}$ is the length of
        $\gamma_k^{(n)}$.
        
        We recall that the local coordinates $(s,\rho,\theta)$ around a knot $\gamma$ 
        are defined in Section~\ref{sec:seif_fram_loc_coord}.
\end{notation}

\subsubsection{Removing the phase jumps}\label{sec:remov_phase}
In this part, we introduce the function $E_{k}$
explained in \cite[Section~3.4]{dirac_s3_paper1}, which removes the
phase jumps along the curve $\gamma=\gamma_k$.

The following result can be found in \cite[Section~3.4]{dirac_s3_paper1}. 
The construction of $E_k$ is reexplained below.

\begin{proposition}\label{prop:sa_link}
	Let $(\uS,\ua)\in \sS^{(K)}\times(0,1)^K$ and let $(\,\cdot\,)$ denote $(\max)$, $(-)$ or $(\min)$.
	Let $\eps>0$ be small enough such that for all $1\le k\neq k'\le K$,
	 $B_{\eps}[\gamma_k]\cap S_{k'}$ and $\gamma_k\cap S_{k'}$ have the same
	number of connected components.
	
	For any $1\le k\le K$, there exists a map (defined below in \eqref{def:curve_phasej_rem})
	\[
	  E_k:B_\eps[\gamma_k]\cap\Omega_{\uS}\mapsto \S^1
	\]
	which maps the set 
	$
		\{\psi \in \dom(\cD_{\bA}^{(\,\cdot\,)}): \supp\psi \in B_{\eps}[\gamma_k] \cap \Omega_{\uS}\}
	$
	onto the set 
	$
		\{\psi \in \dom(\cD_{\bA_k}^{(\,\cdot\,)}): \supp\psi
                \in B_{\eps}[\gamma_k] \cap \Omega_{\uS}\},
	$
	where $\bA_k=2\pi\alpha_k[S_k]$ and we have the correspondence
	\begin{equation}\label{eq:gauge_E_k}
	 \cD_{\bA}^{(\,\cdot\,)}\psi=E_k\left(\cD_{\bA_k}^{(\,\cdot\,)}
         +\frac{c_k}{h_k}\sigma(\bT_k^{\flat})\right)\overline{E}_k\psi,
	\end{equation}
	where $c_k$ is the number:
	\begin{equation}\label{eq:pjump_linkn}
		c_k:=\frac{1}{\ell_k}
			\sum_{k'\neq k}2\pi\alpha_{k'}\link(\gamma_{k'},\gamma_{k})=\frac{1}{\ell_k}\Phi_{k}[(\alpha_{k'})_{k'\neq k}],
	\end{equation}
	Furthermore the function $E_k$ has a bounded derivative in $B_{\eps}[\gamma_k] \cap \Omega_{\uS}$, and for any vector field $X$,
	$\overline{E_k}X(E_k)$ coincides with $ic_k X(s_k(\cdot))$ on $B_{\eps}[\gamma_k] \cap \Omega_{\uS}$
	and $\overline{E_k}X(E_k)$ can be extended to an element in $C^1(B_{\eps}[\gamma_k])$. 
	The choice of $E_k$ is unique up to a constant phase. 
\end{proposition}

In the proposition, $s_k(\cdot)$ denotes the coordinate map that associates to a point $\bp\in B_\eps[\gamma_k]$
	the arclength parameter $s_k(p)\in\R/(\ell_k\Z)$ of its projection onto $\gamma_k$.

The phase jump function $E_k$ is defined as the product of the $E_{k,k'}$'s,
where $E_{k,k'}$ is the function describing the phase jump on $B_\eps[\gamma_k]$ due to $S_{k'}$. 
For $\bp \in B_{\eps}[\gamma_k] \cap \Omega_{\uS}$ we set:
\begin{equation}\label{def:curve_phasej_rem}
	E_k(\bp) :=
	\prod_{k'\neq k}E_{k,k'}(\bp).
\end{equation}

\paragraph{\textit{Construction of $E_{k,k'}$}}

Let $1\le k\neq k'\le K$. 

If $\gamma_k\cap S_{k'}=\emptyset$, we set $E_{k,k'}\equiv 1$,
else we proceed as follows. The curve $\gamma_k$ intersects $S_{k'}$ at the points 
$0\le s_1<s_2<\ldots<s_{M}<\ell_{k}$. Call $C_1,\cdots,C_{M}$ the corresponding connected
components of the intersection $B_{\eps}[\gamma_k]\cap S_{k'}$.
For $1\le m\le M$, $S_{k'}$ induces a phase jump $e^{ib_{m}}$ across the cut $C_m$,
where $b_{m}=\pm 2\pi\alpha_{k'}$.

\medskip

\subparagraph{\textit{Case 1: $M=1$}}
		The surface $S_{k'}$ cuts $\gamma_k$ only once and the cut neighborhood
		$B_\eps[\gamma_k]\setminus S_{k'}=:B(\eps,k,k')$ is contractible. 
		We can thus lift the coordinate map $s_k(\cdot)$ on this subset which gives a smooth function
		\[
		    s_{k,k'}:B_\eps[\gamma_k]\setminus S_{k'}\mapsto \mathbb{R},
		\]
		satisfying for all $\bp\in B_\eps[\gamma_k]\setminus S_{k'}$:
		\[
		    \mathrm{exp}\Big(\frac{2i\pi}{\ell_k}s_{k,k'}(\bp)\Big)=\mathrm{exp}\Big(\frac{2i\pi}{\ell_k}s_{k}(\bp)\Big).
		\]
		We then define for all $\bp\in B_\eps[\gamma_k]\setminus S_{k'}\supset B_\eps[\gamma_k]\cap\Omega_{\uS}$
		\begin{equation}\label{Eq:E_k_k'_1_intersection}
			E_{k,k'}(\bp):=\exp\Big(-i b_{1}\frac{s_{k,k'}(\bp)}{\ell_k}\Big).
		\end{equation}
\subparagraph{\textit{Case 2: $M\ge 2$}} 
		In this case the cuts $C_m$'s split $B_{\eps}[\gamma_k]$ into $M$ sections:
		\[
			B_{\eps}[\gamma_k] \cap \complement S_{k'}
			=: R_{12} \cup R_{23} \cup \dots \cup R_{(M-1)M}\cup R_{M(M+1)},
		\]
		When passing from $R_{(m-1)m}$ to $R_{m(m+1)}$, we are going through $C_{m}$
		which then induces the phase jump $e^{ib_{m}}$ (with $C_{M+1}=C_1$).
		As in the first case, we pick a smooth lift $s_{k,k'}$ of $s_k$
		defined on $B_\eps[\gamma_k]\setminus C_1$.
		
		Writing for $2\le m\le M$
		$$
			R_{m(M+1)}:=R_{m(m+1)} \cup R_{(m+1)(m+2)} \cup \dots \cup R_{M(M+1)},
		$$
		we set for all $\bp\in B_\eps[\gamma_k]\cap\Omega_{\uS}$:
		\begin{equation}\label{Eq:E_k_k'_many_intersections}
			E_{k,k'}(\bp):=\exp\Big(-i\sum_{m=1}^M b_m\frac{s_{k,k'}(\bp)}{\ell_k}
				+i \sum_{m=2}^{M}b_m \mathds{1}_{R_{m(M+1)}}(\bp)\Big).
		\end{equation}
		
	Note that
	    \[
		\sum_{m=1}^{M}b_{m}=-2\pi\alpha_{k'}\link(\gamma_k,\gamma_{k'}),
	     \]
	as the linking number $\link(\gamma_k,\gamma_{k'})$ \cite{Rolfsen}*{Part D, Chapter 5} 
	corresponds to the number of algebraic crossing of $\gamma_{k}$ through the Seifert surface $S_{k'}$ for $\gamma_{k'}$.
	Thus $E_{k,k'}$ is a $\S^1$-valued function with a fixed slope and the correct phase jump across the cuts $C_m$'s: it is a unique
	up to a constant phase.

\subsubsection{Localization}\label{sec:localization}
Let $(\uS^{(n)},\ua^{(n)})_n$ be a sequence converging to $(\uS,\ua)$.
We pick $\eps>0$ such that 1. $B_{\eps}[\gamma_k] \cap B_{\eps}[\gamma_{k'}] = \emptyset$ for $k \neq k'$,
2. the coordinates $(s_k,\rho_k,\theta_k)$ are well-defined on $B_{\eps}[\gamma_k]$ and 3. Proposition~\ref{prop:sa_link}
applies. Note that by the continuity of the coordinates (see Proposition~\ref{prop:conv_coord}), 
the same will be true for for the links $\gamma^{(n)}$ for $n$ large enough.

We introduce a partition of unity (depending on $n$) that localizes around the regions around the knots $\gamma_k^{(n)}$'s,
their Seifert surfaces and away from them.

For $1 \leq k \leq K$, recall the function $\chi_{\delta,\gamma_k^{(n)}}$ 
as in \eqref{eq:chi_loc_curve}, with 
$$
	0<\delta< \eps \min\left\{1, \left(\sup_{k,n}\norm{\kappa_k^{(n)}}_{L^{\infty}} 
	+ \sqrt{\eps + \sup_{k,n}\norm{\kappa_k^{(n)}}_{L^{\infty}}}\right)^{-1}\right\},
$$
where for $s\in\R/(\ell_k^{(n)}\Z)$ we set:
\[
\kappa_k^{(n)}(s):=\sup_{\theta}|\kappa_{g,k}^{(n)}(s)\cos(\theta)+\kappa_{n,k}^{(n)}(s)\sin(\theta)|.
\]
This condition on $\delta$ ensures that $\cos(\rho_k^{(n)})-\sin(\rho_k^{(n)})\kappa_k^{(n)}(s)\ge 1-\eps$
uniformly in $k$ and $n$ (this function is a lower bound of \eqref{eq:def_h} that appears in the 
pullback of the volume form $\vol_{\S^3}$). Furthermore, set
$$
	\chi_{\delta,S_k^{(n)}}(\bp) 
	:= \chi\left(4\tfrac{\dist_{g_3}(\bp,S_k^{(n)})}{\delta}\right)\left(1-\chi_{\delta,\gamma_{k}^{(n)}}(\bp)\right),
$$
which has support close to the Seifert surface $S_k^{(n)}$, but away from the knot $\gamma_k^{(n)}$.
The remainder is then defined as
$$
	\chi_{\delta,R_k^{(n)}}(\bp) := 1-\chi_{\delta,S_k^{(n)}}(\bp) - \chi_{\delta,\gamma_k^{(n)}}(\bp),
$$
and the three $\chi$'s provide us with a partition of unity subordinate to $S_k^{(n)}$.
The partition of unity for the entire link is then given by
\begin{align}\label{def:part_unity}
	1 &= \prod_{k=1}^K\left(\chi_{\delta,\gamma_k^{(n)}}(\bp) + \chi_{\delta,S_k^{(n)}}(\bp)
	+\chi_{\delta,R_k^{(n)}}(\bp) \right)\nn\\
	&= \sum_{\underline{a} \in \{1,2,3\}^K}\chi_{\delta,\underline{a}}^{(n)}(\bp)
	=\sum_{k=1}^K\chi_{\delta,\gamma_k^{(n)}}(\bp)+\sum_{\underline{a} \in \{2,3\}^K}\chi_{\delta,\underline{a}}^{(n)}(\bp),
\end{align}
where $\chi_{\delta,\underline{a}}^{(n)}$ is the product $\prod_{k=1}^K\chi_{\delta,X_k^{(n)}}$, with 
$$
	X_k^{(n)} = \left\{\begin{array}{lr}
	\gamma_k^{(n)}, &a_k = 1,\\
	S_{k}^{(n)}, &a_k = 2,\\
	R_k^{(n)}, &a_k = 3.
	\end{array}\right.
$$

\subsubsection{The model case}\label{sec:model_case}
The results of this section are explained in full details in \cite{dirac_s3_paper1}*{Section 3.2.2}.
Consider the operator $\cD_{\bA}$ associated to $(\uS,\ua)\in \sS^{(K)}\times(0,1)^K$: we now
study its action close to one of the $\gamma_k$'s.  Close to such a
knot $\gamma_k$, we write
$$
	\chi_{\delta,\gamma}\psi = E_k(f_+\xi_+ + f_-\xi_-),
$$ where $E_k$ denotes the phase function
        \eqref{def:curve_phasej_rem}.
The local coordinates manifold
$$
	\T_{\ell} \times \R^2 := (\R / (\ell \Z)) \times \R^2,
$$ 
is endowed with the flat metric and coordinates $(s,u_1,u_2)$
defined by $u_1+iu_2=\rho e^{i\theta}$. The
$Spin^c$ spinor bundle is $\T_{\ell} \times \R^2 \times \C^2$,
endowed with the Clifford map $\wt{\bsigma}$ defined through
$$ 
\wt{\bsigma}(\d s) = \sigma_3, \quad \wt{\bsigma}(\d u_1) =
\sigma_1, \quad \wt{\bsigma}(\d u_2) = \sigma_2.
$$ 
Furthermore, the symbol $\wt{\nabla}$ denotes the canonical
        connection.

Outside the Seifert surfaces $S_{k'}$'s, across which a
$\psi\in\dom(\cD_{\bA}^{(\max)})$ has the phase jumps
\eqref{eq:def_dom_min_dirac_op} inherited from $\cD_{\bA}^{(\min)}$,
$\cD_{\bA}^{(\max)}$ acts like the free Dirac operator
$\sigma(-i\nabla)$.  Written in the coordinates $(s,\rho,\theta)$ the
action of $\cD_{\bA}^{(\max)}$ is given by
\begin{multline}\label{eq:D_max_close_knot}
	\forall\,\bp\in B_{\delta}[\gamma_k]\cap \Omega_{\uS},\\
	\quad \overline{E}_k\begin{pmatrix}\cip{\xi_+}{\cD_{\bA}^{(\max)}\psi} \\ \cip{\xi_-}{\cD_{\bA}^{(\max)}\psi}\end{pmatrix}(\bp)
	=  \big(\wt{\bsigma}(-i\wt{\nabla})+ \mathcal{E}_1 + \mathcal{E}_0+\frac{c_k}{h_k}\sigma_3 \big) f(s(\bp),\rho(\bp),\theta(\bp)),
\end{multline}
where the operators of the last line (which act \emph{away} from the phase jump surfaces in \eqref{eq:D_max_close_knot})
are
\begin{equation}\label{eq:D_gamma_in_T}
	\left\{
	\begin{array}{ccl}
	\wt{\bsigma}(-i\wt{\nabla})	 &=& -i\begin{pmatrix} \partial_s & e^{-i\theta}(\partial_{\rho} - \frac{i}{\rho}\partial_{\theta})\\
		e^{i\theta}(\partial_{\rho} + \frac{i}{\rho}\partial_{\theta}) & -\partial_s\end{pmatrix},\\
		\cE_1 &=& -i \begin{pmatrix} (h_k^{-1}-1)\partial_s & \frac{-ie^{-i\theta}(\rho-\sin\rho)}{\rho \sin\rho}\partial_{\theta}\\
		\frac{ie^{i\theta}(\rho-\sin\rho)}{\rho \sin\rho}\partial_{\theta} & -(h_k^{-1}-1)\partial_s\end{pmatrix}
	 	+ i\sigma_3 \tfrac{\tau}{h} \partial_\theta,\\
		\cE_0 &=& -i(\sigma_3 M_{\xi}(\bT) + \sigma_1M_{\xi}(\bS) + \sigma_2M_{\xi}(\bN)).
	\end{array}\right.
\end{equation}
We recall that $M_\xi$ is the connection form of the canonical
connection \eqref{eq:con_form} in the $(\xi_+,\xi_-)$-trivialization
(relatively to $\gamma_k$), and $h_k$ is defined in
\eqref{eq:def_h}. The term $\frac{c_k}{h_k}\sigma_3$ corresponds to
the action on the phase jump function $E_k$. On $B_\delta[\gamma_k]$,
the matrix $\cE_0$ is bounded, and we will show that the other
sub-principal term $\cE_1$ can be controlled by the leading term. 

Consider the magnetic field supported on $\T_{\ell} \times\{0\}\subset
\T_{\ell} \times\R^2$ with flux $0\le \alpha<1$; we use the singular
gauge $2\pi\alpha[\wt{S}]$ where $\wt{S}$ denotes the surface
$\T_\ell\times\{(u_1,u_2)\in\R^2,\ \theta=0\}$ oriented by
$\partial_{u_2}$.\footnote{In \cite{dirac_s3_paper1} we use the
  ``smooth" gauge $\alpha\d\theta$: it is linked to the singular gauge
  $2\pi\alpha[\wt{S}]$ through the gauge transformation $e^{i\alpha
    \theta}$, where $\theta$ has branch cut along $\theta=0$.}

The minimal operator $\cD_{\T_{\ell},\alpha}^{(\min)}$ with domain 
$$
	\dom(\cD_{\T_{\ell},\alpha}^{(\min)}) 
	:= \big\{ \psi\in H^1(\T_\ell\times\R^2\setminus \wt{S})^2,\ \psi_{|_{\wt{S}_+}}=e^{-2i\pi\alpha}\psi_{|_{\wt{S}_-}}\big\},
$$ acts like the free Dirac operator $\wt{\bsigma}(-i\wt{\nabla})$
        outside $\wt{S}$.  An important property, which can be easily
        obtained by Stokes' formula, is the fact that for all $f_0\in
        \dom\big(\cD_{\T_{\ell},\alpha}^{(\min)}\big)$ the Dirac
        energy coincides with the Dirichlet energy on $\{ (s,u)\in
        \T_\ell\times\R^2,\ \theta\neq 0\}$:
\begin{equation}\label{eq:dirac_dirichlet}
	\int |\cD_{\T_{\ell},\alpha}^{(\min)}f_0|^2=\int|(\wt{\nabla}f_0)_{|_{\theta\neq 0}}|^2.
\end{equation}

We are interested in two particular extensions $\cD_{\T_{\ell},\alpha}^{(\pm)}$ with domains
\begin{align}\label{eq:form_domain_D_T}
	&\dom\big(\cD_{\T_{\ell},\alpha}^{(+)} \big):=\dom\big(\cD_{\T_{\ell},\alpha}^{(\min)} \big)
	\overset{\perp_{\cG}}{\oplus}
	\left\{ \frac{1}{\sqrt{2\pi\ell}} \sum_{j \in \T_{\ell}^*} \lambda_j e^{ijs}
	\begin{pmatrix}K_{1-\alpha}(\rho\langle j\rangle)e^{i(\alpha-1)\theta}\\0\end{pmatrix}: \ul \in\ell_2(\T_{\ell}^*)\right\}\nonumber\\
	&\dom\big(\cD_{\T_{\ell},\alpha}^{(-)} \big):=\dom\big(\cD_{\T_{\ell},\alpha}^{(\min)} \big)
	\overset{\perp_{\cG}}{\oplus}
	\left\{ \frac{1}{\sqrt{2\pi\ell}} \sum_{j \in \T_{\ell}^*} \lambda_j e^{ijs}
	\begin{pmatrix}0\\K_{\alpha}(\rho\langle j\rangle)e^{i\alpha\theta}\end{pmatrix}: \ul \in\ell_2(\T_{\ell}^*)\right\}.
\end{align}
Here, $\perp_{\cG}$ means the orthogonality with respect to the graph norm,
$K_\alpha$ and $K_{1-\alpha}$ denote modified Bessel functions 
of the second kind, the group
$\T_{\ell}^*:=\frac{2\pi}{\ell}\Z$ is the Pontrjagin dual of $\T_{\ell}$ 
and $\langle j\rangle:=\sqrt{1+j^2}$. The action is as follows: for all $f\in\dom\big(\cD_{\T_{\ell},\alpha}^{(\pm)} \big)$ we have:
\begin{equation}\label{eq:action_D_T}
  \cD_{\T_{\ell},\alpha}^{(\pm)}f:= \big(\wt{\bsigma}(-i\wt{\nabla})f\big)\big|_{\theta\neq 0}.
\end{equation}
These extensions should again be thought of as the situation when the singular part 
of the spinor ``aligns" with or against the magnetic field.
For any of these extensions we can decompose $f \in \dom\big(\cD_{\T_{\ell},\alpha}^{(\pm)}\big)$ 
with respect to the splittings \eqref{eq:form_domain_D_T} as 
\begin{align}\label{eq:model_decomp}
	f = f_0+f_{\sing}(\ul)
\end{align}
with $f_0 \in \dom\big(\cD_{\T_{\ell},\alpha}^{(\min)} \big)$.
Furthermore, there is a lower bound on the graph norm $\norm{\,\cdot\,}_{\T_{\ell}}$ (proved in \cite{dirac_s3_paper1})
\begin{equation}\label{eq:sing_contr}
	\norm{f}_{\T_{\ell}}^2=\norm{f_0}_{\T_{\ell}}^2+\norm{f_{\sing}}_{\T_{\ell}}^2
	\ge \norm{f_{\sing}}_{\T_{\ell}}^2\ge (C_\alpha+C_{1-\alpha}) \norm{\ul}_{\ell_2(d\mu_{\T_{\ell}^*})}^2,
\end{equation}
where $C_\alpha$ and $C_{1-\alpha}$ denote the integrals
\[
	\int_0^{+\infty}K_\beta(r)^2r\d r, \quad\beta=\alpha,1-\alpha.
\]
Observe that, due to the behaviour of $K_\alpha$ close to $0$, we have as $\alpha\to 1^-$
\begin{equation}\label{eq:est_C_alpha}
	C_\alpha+C_{1-\alpha}=\Theta\Big(\frac{1}{1-\alpha}\Big),
\end{equation}
that is, it behaves like $(1-\alpha)^{-1}$ as $\alpha\to 1^-$.
In other words, the graph norm  controls the $\ell_2$-norm of $\ul$. Observe also the following:
\begin{equation}\label{eq:sing_psi_la}
	\int|f_{\sing}(\ul)|^2=C_\alpha \sum_{j\in\T_\ell^*} \frac{|\lambda_j|^2}{1+j^2}.
\end{equation}

Furthermore, any $f \in \dom\big(\cD_{\T_{\ell},\alpha}^{(\pm)}\big)$ satisfies
\begin{align}\label{eq:dpm_decouple}
	\int |\cD_{\T_{\ell},\alpha}^{(\pm)}f|^2
	= \int|(\partial_s f)\big|_{\theta\neq 0}|^2 + \int |(\wt{\bsigma}(-i\wt{\nabla}_u)f)\big|_{\theta\neq 0}|^2,
\end{align}
which can be verified with the help of an explicit calculation.
Recall that $\norm{\cdot}_{\bA_k}$ and $\norm{\cdot}_{\T_{\ell}}$
denote the graph norms of $\cD_{\bA_k}$ and of the model operator $\cD_{\T_{\ell},\alpha}$
respectively. 

The next lemma (proved in \cite{dirac_s3_paper1}) compares the two graph norms.
\begin{lemma}\label{lem:control_Ngr}
	Let $\gamma_0$ be a knot of length $\ell$ with Seifert surface $S_0$ and
	$0<\alpha<1$. Given $\delta>0$ small enough and any 
	$\psi\in \dom\big(\mathcal{D}_{\bA_0}\big)$, $\bA_0=2\pi\alpha[S_0]$
	we write
	$$
		e^{i\zeta}\chi_{\delta,\gamma}\psi=f \cdot \xi = f_+\xi_++ f_{-}\xi_{-},
	$$ where the local gauge transformation $e^{i\zeta}$ shifts
                the phase jump $e^{-2i\pi\alpha}$ from the surface
                $S\cap B_{2\delta}[\gamma]$ to
                $\{(s,\rho,\theta),\,\rho\le 2\delta\ \&\ \theta=0\}$.
                Then $f=(f_+,f_-)\in
                \dom(\cD_{\T_{\ell},\alpha}^{(-)})$ and there exists a
                constant $C=C(\delta,\alpha,S)$ such that
	$$
		\norm{f}_{\T_{\ell}} 
		\leq C\norm{\chi_{\delta,\gamma}\psi}_{\bA_0}
		\leq C(1+\delta^{-1})\norm{\psi}_{\bA_0},
	$$
	and the bound $C(\delta,\alpha,S)$ depends continuously on $S\in\sS$ and $\alpha$.
\end{lemma}

\subsubsection{Proof of Theorem~\ref{thm:bump_cont_bdry_A}}\label{sec:start_proof}
	Since $\sS^{(K)}\times \Tf^{K}$ is metric the continuity
        in (S1) above is the same as sequential continuity.
	That (S2) implies (S1) is obvious. 
	
	\noindent {\bf (S1) implies (S2):} Assume (S1).
	We claim that there exists a sequence $(\uS^{(n)},\ua^{(n)})$ converging to $(\uS,\ua)$ 
	such that for all bump functions $\phi$ centered around $\lambda$ 
	there holds $\liminf_{n\to+\infty}\norm{\phi(\cD_{\bA^{(n}})-\phi(\cD_{\bA})}_{\cB}>0$.
	
	\paragraph{\textit{Proof of the claim}} 
	Let $\phi_s$ be a bump-function. We define a family of bump functions around $\lambda$ by setting
	$\phi_n(x):=\phi_s(n(x-\lambda))$, $n\ge 1$. Let $B_n:=B((\uS,\ua),2^{-n})$.
	By assumption, $\phi_n\circ \bbd_{|_{B_n}}$ is not continuous. Thus we can find a sequence
	 $(\uS^{(n,m)},\ua^{(n,m)})\underset{m\to+\infty}{\to}(\uS^{(n,\infty)},\ua^{(n,\infty)})$ \emph{in the ball} $B_n$
	such that $\lim_{m\to+\infty}\norm{\phi_n(\cD_{\bA^{(n,m)}})-\phi_n(\cD_{\bA^{(n,\infty)}})}_{\cB}= \eps_n>0$.
	Let $m(n)$ be a rank above which the norm is larger than $\eps_n/2$. We can assume that $m(n+1)>m(n)$,
	and for every $n_0$ the triangle inequality gives:
	\[
	 \liminf_{n\to+\infty}(\norm{\phi_{n_0}(\cD_{\bA^{(n,\infty)}})-\phi_{n_0}(\cD_{\bA})}_{\cB}+\norm{\phi_{n_0}(\cD_{\bA^{(n,m(n))}})-\phi_{n_0}(\cD_{\bA})}_{\cB})
	 \ge \tfrac{\eps_{n_0}}{2}.
	\]
	Thus one subsequence $(\uS^{(n_k,m(n_k))},\ua^{(n_k,m(n_k))})_k$ or $(\uS^{(n_k,+\infty)},\ua^{(n_k,+\infty)})_k$,
	which we rename $(\uS^{(k)},\ua^{(k)})_k$, satisfies the claim for the family $(\phi_{n_{k'}})_{k'\ge 1}$ of bump functions. For another bump function $\phi$
	around $\lambda$, for $k$ big enough we can factorize $\phi_{n_k}=f\circ \phi$ with $f$ continuous.

	\medskip
	We now consider this sequence $(\uS^{(n)},\ua^{(n)})_n$. As $\cD_{\bA}$ has discrete spectrum, we can isolate $\lambda$:
	there exists $\eta>0$ such that $\dist(\lambda\pm\eta,\spec(\cD_{\bA}))>0$ and 
	$\spec(\cD_{\bA})\cap[\lambda+\eta,\lambda-\eta]\subset\{\lambda\}$. For $n_0$ big enough, 
	the bump function $\phi_{n_0}$ has support in $(\lambda+\eta,\lambda-\eta)$.
	
	By Theorem~\ref{thm:strg_res_cont}, $\cD_{\bA^{(n)}}$ converges to $\cD_{\bA}$ in the strong resolvent sense.
	By Lemma~\ref{lem:char_sres_conv}, for every eigenfunction $\psi\in\dom(\cD_{\bA})$ with eigenvalue $\mu$, 
	there exists a sequence $(\psi_n)$ with $\psi_n\in\dom(\cD_{\bA^{(n)}})$ such that $(\psi_n,\cD_{\bA^{(n)}}\psi_n)$
	converges to $(\psi,\mu\psi)$ in norm. As the $\cD_{\bA^{(n)}}$'s also have discrete spectrum, we can assume $\cD_{\bA^{(n)}}\psi_n=\mu_n\psi_n$
	with $\mu_n\to \mu$.
	
	In particular, we can write 
	$
		\ran \mathds{1}_{[\lambda+\eta,\lambda-\eta]}\left(\cD_{\bA^{(n)}}\right) = V_{\eta}^{(n)} \overset{\bot}{\oplus}W_{\eta}^{(n)},
	$
	where $V_{\eta}^{(n)}, W_{\eta}^{(n)}$ are subspaces spanned by eigenfunctions of $\cD_{\bA^{(n)}}$ with
	$$
		\cD_{\bA^{(n)}}P_{W_{\eta}^{(n)}} \to \lambda\mathds{1}_{\{\lambda\}}(\cD_{\bA})\textrm{ in } \cB(L^2(\S^3)^2).
	$$
	Yet for all $n_0$, we have $\lim_{n\to\infty}\norm{\phi_{n_0}(\cD_{\bA^{(n)}})-\phi_{n_0}(\cD_{\bA})}_{\mathcal{B}}>0$. Thus for $n_0$ big enough we have
	$\supp\,\phi_{n_0}\subset (\lambda-\eta,\lambda+\eta)$ and
	$$
		\lim_{n\to+\infty}\norm{\phi_{n_0}(\cD_{\bA^{(n)}})P_{V_\eta^{(n)}}}_{\cB}
			=\lim_{n\to+\infty}\norm{\phi_{n_0}(\cD_{\bA^{(n)}})-\phi_{n_0}(\cD_{\bA})}_{\cB}>0.
	$$
	In particular $V_{\phi}^{(n)}\neq\{0\}$, and we can find a sequence $\psi^{(n)}\in V_{\eta}^{(n)}$ 
	of normalized eigenfunctions of $\cD_{\bA^{(n)}}$ with corresponding eigenvalues $\lambda^{(n)}$ all lying in $\supp\,\phi_{n_0}$.
	By another diagonal argument we can assume $\lambda^{(n)}\to\lambda$.
	
	At last, (S2) follows from the observation that any sequence satisfying:
	\[
		\psi^{(n)} \in V_{\eta}^{(n)}, \quad \norm{\psi^{(n)}}_{L^2} =1
	\]
	converges weakly to $0$: $\psi^{(n)}\rightharpoonup 0$. Indeed, pick 
	an accumulation point $\psi$ in the weak topology of $L^2(\S^3)^2$: let us show $\psi\in(\dom(\cD_{\bA}))^\perp=\{0\}$.
	As done above, let $\wt{\psi}$ be an eigenfunction $\cD_{\bA}$, approximated by eigenfunctions $\wt{\psi}^{(n)}$
	of $\cD_{\bA^{(n)}}$. If $\cD_{\bA}\wt{\psi}\neq \lambda\wt{\psi}$, then for $n$ big enough $\wt{\psi}^{(n)}\perp \psi^{(n)}$ as we have
	$|\mu_n-\lambda|>\eta$ , else $\cD_{\bA}\wt{\psi}= \lambda\wt{\psi}$ and we can choose $P_{W_{\eta}^{(n)}}\wt{\psi}^{(n)}=\wt{\psi}^{(n)}\perp \psi^{(n)}$.
	Anyway $\cip{\wt{\psi}}{\psi}_{L^2}=\lim_{n\to+\infty}\cip{\wt{\psi}^{(n)}}{\psi^{(n)}}_{L^2}=0$.
\medskip

	\noindent\textbf{(S2) implies (S3):}	
	We consider a vanishing sequence $(\psi^{(n)})_n$ along a sequence $(\uS^{(n)},\ua^{(n)})$ converging
	to $(\uS,\ua)$. For each flux $\alpha_k^{(n)},\alpha_k$ we consider the representative in $[0,1)$, that we write the same way.
	
	The sequence $(\uS^{(n)},\ua^{(n)})$ can be split into at most $2^{K}$ subsequences, indexed by the set
	$A\subset \{1,\cdots,K\}$ of indices $k$ for which $\alpha_{k}^{n}\to 1^-$ along the subsequence.

	By Theorem~\ref{thm:compactness}, there exists a subsequence with $A\neq\emptyset$, otherwise, the sequence $(\psi^{(n)})_n$
	cannot be vanishing. We will henceforth focus on one such subsequence associated to $A\neq\emptyset$, and furthermore denote by
	$$
		\gamma_{A} := \bigcup_{k \in A}\gamma_{k}.
	$$
	From the same Theorem~\ref{thm:compactness} it also follows that
	\begin{equation}\label{eq:psi_n_vanish}
		\psi^{(n)} \to 0 \textrm{ in } L_{\loc}^2(\S^3 \setminus \gamma_{A})^2.
	\end{equation}
	Thanks to \eqref{def:w-l_approx_seq} we also get that
	\begin{equation}\label{eq:Dpsi_n_vanish}
		\cD_{\bA^{(n)}} \psi^{(n)} \to 0 \textrm{ in } L_{\loc}^2(\S^3 \setminus \gamma_{A})^2.
	\end{equation}

	\paragraph{\textit{Localization of the sequence around the knots}}
	Using $\chi_{\delta,\gamma_k^{(n)}}$, we now localize around a knot $\gamma_k$ with $k\in A$,
	and drop the subscript $k$. To simplify notations, the relative torsion $\tau_{S_k^{(n)}}$ 
	will simply be written $\tau_k^{(n)}$ or just $\tau_n$.
	
	The phase function $E_k^{(n)}$ for $\gamma_k^{(n)}$ (see ~\eqref{def:curve_phasej_rem})
	is written $E^{(n)}$. We then decompose with respect to the trivialization $\xi^{(n)}=(\xi_+^{(n)},\xi_-^{(n)})$
	$$
		\chi_{\delta,\gamma^{(n)}}\psi^{(n)} = 
		E^{(n)}f^{(n)}\cdot \xi^{(n)}, f^{(n)}=(f_+^{(n)},f_-^{(n)}).
	$$
	As in Lemma~\ref{lem:control_Ngr}, up to a local gauge transformation 
	$e^{i\zeta_k^{(n)}}$ shifting the phase jump from $S_k^{(n)}$ to $\{\theta_n=0\}$,
	the spinor $f^{(n)}$ is in $\dom\big(\cD_{\T_{\ell_k^{(n)}},\alpha_k^{(n)}}^{(-)}\big)$ when 
	seen as a function of $(s_n,\rho_n,\theta_n)$. For any $n$, we perform this local gauge transformation
	around $\gamma_k^{(n)}$, and still write $f^{(n)}$ the new spinor. According to \eqref{eq:model_decomp} 
	we split $f^{(n)}$ into a regular part $\wt{f}_0^{(n)}$ and a singular part $\wt{f}_{\sing}^{(n)}$:
	$$
		E^{(n)}f^{(n)} = E^{(n)}\wt{f}_0^{(n)} + E^{(n)}\wt{f}_{\sing}^{(n)}.
	$$
	Furthermore, set 
	$$
		f_0^{(n)} := \chi(\tfrac{\rho}{2\delta})\wt{f}_0^{(n)},
		\quad f_{\sing}^{(n)} := \chi(\tfrac{\rho}{2\delta})\wt{f}_{\sing}^{(n)}.
	$$
	Observe that $f_0^{(n)} + f_{\sing}^{(n)} = \wt{f}_0^{(n)} + \wt{f}_{\sing}^{(n)}$, but that 
	$\wt{f}_0^{(n)},\wt{f}_{\sing}^{(n)}$ do not have compact support. 
	We also define
	\begin{align*}
		e^{i\zeta_k^{(n)}}\psi_0^{(n)}(f)(\bp) &:= f_0^{(n)}(s_n(\bp),\rho_n(\bp),\theta_n(\bp)) \cdot \xi^{(n)}(\bp)\in\dom(\cD_{\bA^{(n)}}^{(\min)}),\\
		e^{i\zeta_k^{(n)}}\psi_{\sing}^{(n)}(f)(\bp) &:= f_{\sing}^{(n)}(s_n(\bp),\rho_n(\bp),\theta_n(\bp)) \cdot \xi^{(n)}(\bp)\in\dom(\cD_{\bA^{(n)}}),\\
		e^{i\zeta_k^{(n)}}\psi^{(n)}(f)(\bp)&:=\psi(f)_0^{(n)}(\bp)+\psi_{\sing}^{(n)}(f)(\bp).
	\end{align*}
	We then have:
	\[
	 E^{(n)}\psi^{(n)}(f)=e^{i\zeta_k^{(n)}}\chi_{\delta,\gamma^{(n)}}\psi^{(n)}.
	\]
	\paragraph{\textit{Boundedness of the localized functions}}
	
	The form of $\wt{f}_{\sing}^{(n)}$ is somewhat easy to describe:
	\begin{equation}\label{eq:explicit_sing}
		\wt{f}_{\sing}^{(n)}(s,\rho,\theta)
		= \frac{1}{\sqrt{2\pi \ell^{(n)}}}\sum_{j \in (\T_{\ell_k^{(n)}})^*}e^{ijs} \lambda_j^{(n)}
		\begin{pmatrix} 0\\ K_{\alpha^{(n)}}(\rho \langle j\rangle)e^{i\alpha^{(n)}\theta} \end{pmatrix},
	\end{equation}
	(Recall that $\cD_{\bA}=\cD_{\bA}^{(-)}$ is related to the model operator $\cD_{\T_{\ell_k^{(n)}},\alpha_k^{(n)}}^{(-)}$.)
	
	Let $\bA_k^{(n)}$ be the magnetic potential:
	\[
	\bA_k^{(n)}:=2\pi\alpha_k^{(n)}[S_k^{(n)}]. 
	 \]

	 By \eqref{def:w-l_approx_seq} and \eqref{eq:gauge_E_k}:
	\[
	 \limsup_{n\to+\infty}\norm{\overline{E_k^{(n)}}\chi_{\delta,\gamma_k^{(n)}}\psi^{(n)}}_{\bA_k^{(n)}}<+\infty.
	\]

	By Lemma~\ref{lem:control_Ngr}, we have 
	\begin{equation}\label{eq:bound_graph_norm_f}
	\limsup_{n\to+\infty}\norm{\wt{f}^{(n)}}_{\T_n,\alpha^{(n)}}^2=
	  \limsup_{n\to+\infty}\big(  \norm{\wt{f}_0^{(n)}}_{\T_n}^2+\norm{\wt{f}_{\sing}^{(n)}}_{\T_n}^2\big)<+\infty.
	\end{equation}
	Using estimates \eqref{eq:sing_contr}-\eqref{eq:est_C_alpha}, this gives:
	\begin{equation}\label{eq:bound_graph_norm_sing}
	 \limsup_{n \to \infty} (1-\alpha^{(n)})^{-1}\sum_{j} |\lambda_j^{(n)}|^2 < \infty.
	\end{equation}
	By Lemma~\ref{lem:control_Ngr}, we also have 
	\begin{equation}\label{eq:bound_graph_norm_psi_f}
	 \limsup_n(\norm{\psi_0^{(n)}(f)}_{\bA_k^{(n)}}+\norm{\psi_{\sing}^{(n)}(f)}_{\bA_k^{(n)}})<+\infty.
	\end{equation}

	\paragraph{\textit{Careful study of $\cD_{\bA^{(n)}}\psi^{(n)}$ around the knots}}
	Recall \eqref{eq:D_max_close_knot}-\eqref{eq:D_gamma_in_T}; we introduce the operator
	$$
		\cQ^{(n)} := \nt + \cE_1^{(n)} + \cE_0^{(n)}+\frac{c_k^{(n)}}{h_n}\sigma_3,
	$$
	acting on $L^2(\T_n \times \R^2)^2$, where $\T_n := \R/(\ell_k^{(n)}\Z)$.
	Note that $h_n$ is the function $h$ in \eqref{eq:def_h} associated to the knot $\gamma_k^{(n)}$.
	We recall the correspondence:
	\begin{equation}\label{eq:rappel}
	 \big(e^{i\zeta_k^{(n)}}\overline{E}^{(n)}\cD_{\bA^{(n)}}\psi^{(n)}\big)(\bp)
	    =(\cQ^{(n)}f)|_{\theta=0}(s_n,\rho_n,\theta_n)(\bp)\cdot \xi^{(n)}(\bp),
	\end{equation}
	The part $\nt$ corresponds to $\cD_{\T_n,\alpha_k^{(n)}}$: see \eqref{eq:action_D_T}.
	We are now going to study
	\begin{align}
		\!\!\cQ^{(n)}(f_0^{(n)} + f_{\sing}^{(n)})
		&=\left(\cE_1^{(n)} - i\sigma_3\tfrac{\tau_n}{h_n}\partial_{\theta_n}\right)f_{\sing}^{(n)}\,_{\big|_{\theta\neq0}}\label{eq:Q_rest_1}\\
		&\quad + \Big(\cE_1^{(n)} + \cE_0^{(n)}+\frac{c_k^{(n)}}{h_n}\sigma_3\Big)f_0^{(n)}\,_{\big|_{\theta\neq0}}\label{eq:Q_rest_2}\\
		&\quad+\Big(\nt(f_0^{(n)}+f_{\sing}^{(n)}) + \big(\cE_0^{(n)}
		+\sigma_3(\tfrac{c_k^{(n)}}{h_n}+i\tfrac{\tau_n}{h_n}\partial_{\theta_n}\big)f_{\sing}^{(n)}\Big)_{\big|_{\theta\neq0}}\label{eq:Q_rest_3}.
	\end{align}
	We now determine the behavior of $f_0^{(n)}$ and
	$f_{\sing}^{(n)}$, then prove that the $L^2$-norm of \eqref{eq:Q_rest_1}-\eqref{eq:Q_rest_2} converge to $0$ 
	and study in details the term \eqref{eq:Q_rest_3}.

	\noindent\textbf{Convergence of $\wt{f}_0^{(n)}$ and concentration of $\wt{f}_{\sing}^{(n)}$:}

	Let us show that the following holds: first for any $\eps_1>0$
	\begin{equation}\label{eq:collapse_sing}
		m_{\eps_1}(\wt{f}_{\sing}^{(n)}):=
			\int_{\rho\ge \eps_1}(|\wt{f}_{\sing}^{(n)}|^2+|\cD_{\T_n,\alpha^{(n)}}\wt{f}_{\sing}^{(n)}|^2)\rho\d\rho\d s\d\theta\to_{n\to\infty} 0,
	\end{equation}
	and then
	\begin{equation}\label{eq:q_0_to_0}
		\lim_{n \to \infty}\norm{\wt{f}_0^{(n)}}_{L^2(\T_n \times \R^2)^2} = 0.
	\end{equation}
	
	 These results imply 
	$\psi_{\sing}^{(n)}(f)\to 0$ in $L^2_{\loc}(\S^3\setminus\gamma_k)^2$
	and $\psi_0^{(n)}(f)\to 0$ in $L^2(\S^3)^2$.

	Using the explicit form of $\wt{f}_{\sing}^{(n)}$ (\eqref{eq:explicit_sing})
	and \eqref{eq:bound_graph_norm_sing}, we get:
	\begin{align*}
	m_{\eps_1}(\wt{f}_{\sing}^{(n)})&=
		\sum_{j\in(\T_n)^*}|\lambda_j^{(n)}|^2\int_{\rho\ge \eps_1}(K_{\alpha^{(n)}}^2+K_{1-\alpha^{(n)}})(\rho\langle j\rangle)^2\rho\d\rho,\\
							&\le 2\sum_{j\in(\T_n)^*}\frac{|\lambda_j^{(n)}|^2}{1+j^2}\int_{\rho\ge \eps_1}K_{1}(\rho)^2\rho\d\rho,\\
							&\le C(1-\alpha^{(n)})\int_{\rho\ge \eps_1}K_{1}(\rho)^2\rho\d\rho\underset{n\to+\infty}{\to}0.
	\end{align*}

	We turn now to $\wt{f}_0^{(n)}$. By renormalizing the parameter $s$ with the transformation
	\[
		g(s,z)\in L^2(\T_n\times\R^2)^2\mapsto g_{\mathrm{sc}}(s',z)
		:=g\Big(\frac{\ell_k^{(n)}}{\ell_k}s',z\Big)\in L^2(\T_{\ell_k}\times\R^2)^2,
	\]
	we can work in the fixed Hilbert space $L^2(\T\times\R^2)^2$. We now consider the renormalized $\wt{f}_0^{(n)}$:
	\[
		\wt{g}_0^{(n)}(s',z)
		:=f_0^{(n)}\Big(\frac{\ell_k^{(n)}}{\ell_k}s',z\Big)\in \dom(\cD_{\T_{\ell_k,\alpha_k^{(n)}}}^{(\min)}).
	\]
	By \eqref{eq:bound_graph_norm_f}, $(\norm{\wt{f}_0^{(n)}}_{\T_n,\alpha_k^{(n)}})_n$ is bounded.
	The two facts~\eqref{eq:psi_n_vanish} and \eqref{eq:collapse_sing} imply that $\wt{g}_0^{(n)}\to 0$ in 
	$L^2_{\loc}(\T_{\ell_k}\times(\R^2\setminus\{0\}))^2$. The equality~\eqref{eq:dirac_dirichlet} 
	implies that the sequence $(\wt{g}_0^{(n)})_n$ is $H^1(\{\theta\neq 0\})^2$-bounded, 
	and thus converges up to extraction in $L^2_{\loc}(\T_{\ell_k}\times\R^2)^2$. 
	As $\wt{f}_{\sing}^{(n)}$ concentrates on $\T_n\times\{0\}$ and the sum of
	$\wt{f}_{\sing}^{(n)}$ and $\wt{f}_{0}^{(n)}$ has compact support, we obtain $\wt{g}_0^{(n)}\to 0$ in $L^2$, and
	\[
		\norm{\wt{f}_0^{(n)}}_{L^2(\T_n\times\R^2)^2}^2
		=\frac{\ell_k}{\ell_k^{(n)}}\norm{\wt{g}_0^{(n)}}_{L^2(\T_{\ell_k}\times\R^2)^2}^2\to 0.
	\]

	\smallskip
	\noindent\textbf{The term~\eqref{eq:Q_rest_1}:}
	Using \eqref{eq:explicit_sing} and \eqref{eq:collapse_sing}, one can show by direct computation that
	$$
		\lim_{n \to \infty}\norm{\rho_n (\partial_{s_n} f_{\sing}^{(n)})_{|_{\theta\neq 0}}}_{L^2(\T_n \times \R^2)^2}= 0,
		\quad \lim_{n \to \infty}\norm{\rho_n \partial_{\theta_n} f_{\sing}^{(n)}\,_{|_{\theta\neq 0}}}_{L^2(\T_n \times \R^2)^2}= 0.
	$$

	This shows that the norm of $(\cE_1^{(n)} - i\sigma_3\frac{\tau_n}{h_n}\partial_{\theta_n})f_{\sing}^{(n)}\,_{|_{\theta\neq 0}}$ tends to 
	zero as $n \to \infty$. Similarly, by direct computation thanks to \eqref{eq:collapse_sing} we have:
	\begin{equation}\label{eq:collapsa_rho_sing}
	  \norm{\rho f_{\sing}^{(n)}}_{L^2(\T_n\times\R^2)^2}+\norm{\rho \cD_{\T_n,\alpha_k^{(n)}}f_{\sing}^{(n)}}_{L^2(\T_n\times\R^2)^2}
	  \to 0.
	\end{equation}

	\smallskip
	\noindent\textbf{The term~\eqref{eq:Q_rest_2}:}
	The convergence \eqref{eq:q_0_to_0} together with the uniform boundedness of 
	$\cE_0^{(n)}$ in $n$ shows that $(\cE_0^{(n)}+\tfrac{c_k^{(n)}}{h_n}\sigma_3)f_0^{(n)}$ 
	converges to zero in norm.
	
	We now show that  $\cE_1^{(n)}f_0^{(n)}$ converges to zero. Below $\bG^{(n)}$ denotes the vector field
	defined in \eqref{def:bG} associated to the knot $\gamma_k^{(n)}$. 
	Using \eqref{def:bG} and \eqref{eq:rel_phase}, for $\bp\in B_{2\delta}[\gamma_k^{(n)}]\cap\{\theta_n\neq0\}$:  
	$(\tfrac{1}{h_n}-1)\partial_{s_n}-\tfrac{\tau_n}{h_n}\partial_{\theta_n}$ corresponds to $(1-h_n)\bT^{(n)}-\tau_n\sin(\rho_n)\bG^{(n)}$ 
	and the following equality holds:
	\begin{align}\label{eq:E1_comp_1}
		&-i\sigma_3\Big[(h_n^{-1}-1)\partial_{s_n} - \tfrac{\tau_n}{h_n}\partial_{\theta_n}\nn\\
		&\quad\quad\quad\quad\quad\quad+\big((1-h_n)M_{\xi^{(n)}}(\bT^{(n)})-\tau_n\sin\rho_n M_{\xi^{(n)}}(\bG^{(n)})\big)\Big]f_0^{(n)}\big|_{\theta\neq 0}\nn\\
		&\,=\!\begin{pmatrix}\cip{\xi_{+}^{(n)}}{-i\bsigma(\bT^{(n),\flat})[(1-h_n)\nabla_{\bT^{(n)}} 
		- \tau_n\sin\rho_n\nabla_{\bG^{(n)}}](e^{i\zeta_k^{(n)}}\psi_0^{(n)}(f))\big|_{\Omega_{\uS}}}\\
		\cip{\xi_{-}^{(n)}}{-i\bsigma(\bT^{(n),\flat})
		[(1-h_n)\nabla_{\bT^{(n)}} - \tau_n\sin\rho_n\nabla_{\bG^{(n)}}](e^{i\zeta_k^{(n)}}\psi_0^{(n)}(f))\big|_{\Omega_{\uS}}}
		\end{pmatrix}\!.
	\end{align}
	Similarly $\tfrac{\rho_n-\sin(\rho_n)}{\rho_n\sin(\rho_n)}\partial_{\theta_n}$ corresponds to $\tfrac{\rho_n-\sin(\rho_n)}{\rho_n}\bG^{(n)}$ and we have:
	\begin{multline}\label{eq:E1_comp_2}
		-i(-\sin\theta_n \sigma_1 + \cos\theta_n \sigma_2)\Big[\frac{\rho_n-\sin\rho_n}{\rho_n\sin\rho_n}\partial_{\theta_n}
		+\tfrac{\rho_n-\sin\rho_n}{\rho_n}M_{\xi^{(n)}}(\bG^{(n)})\Big]f_0^{(n)}\big|_{\theta\neq 0}\\
		=\begin{pmatrix}\cip{\xi_+^{(n)}}
		  {-i\bsigma(\bG^{(n),\flat})\tfrac{\rho_n-\sin\rho_n}{\rho_n}\nabla_{\bG^{(n)}}e^{i\zeta_k^{(n)}}\psi_0^{(n)}(f)\big|_{\Omega_{\uS}}}\\
		\cip{\xi_{-}^{(n)}}
		  {-i\bsigma(\bG^{(n),\flat})\tfrac{\rho_n-\sin\rho_n}{\rho_n}\nabla_{\bG^{(n)}}e^{i\zeta_k^{(n)}}\psi_0^{(n)}(f)\big|_{\Omega_{\uS}}}.
		\end{pmatrix}(\bp).
	\end{multline}
	As $\norm{\psi_0^{(n)}(f)}_{L^2}\to 0$, the zero-order terms with the $M_{\xi^{(n)}}$'s all converge to $0$ in the limit.
	Thus, it suffices to prove that the following three terms converge to $0$:
	\begin{align*}
		&\norm{(1-h_n)\nabla_{\bT^{(n)}} \psi_0^{(n)}(f)\big|_{\Omega_{\uS}}}_{L^2(\S^3)^2},
		\quad \norm{\tau_n \sin\rho_n\nabla_{\bG^{(n)}} \psi_0^{(n)}(f)\big|_{\Omega_{\uS}}}_{L^2(\S^3)^2},\\
		&\norm{\tfrac{\rho_n - \sin\rho_n}{\sin\rho_n}\nabla_{\bG^{(n)}} \psi_0^{(n)}(f)\big|_{\Omega_{\uS}}}_{L^2(\S^3)^2}.
	\end{align*}
	We only treat the first one, the calculation for the latter two is identical. As $(1-h_n) = \mathcal{O}(\rho_n)$, it suffices to estimate
	\begin{align*}
		\norm{\rho_n (\nabla_{\bT^{(n)}} \psi_0^{(n)}(f))\big|_{\Omega_{\uS}}}_{L^2(\S^3)^2} 
		&\leq \norm{\nabla_{\bT^{(n)}}(\rho_n\psi_0^{(n)}(f))\big|_{\Omega_{\uS}}}_{L^2(\S^3)^2}\\
		&\leq \norm{\cD_{\bA_k^{(n)}}^{(\min)}(\rho_n\psi_0^{(n)}(f))}_{L^2(\S^3)^2},
	\end{align*}
	where in the last step we used the Lichnerowicz formula that holds for the minimal Dirac operator, 
	see \cite[Section~3.2.1]{dirac_s3_paper1}. Recall $\bA_k^{(n)}=2\pi\alpha_k^{(n)}[S_k^{(n)}]$, 
	\cite{dirac_s3_paper1}*{Lemma 12} gives: if $\psi\in\dom(\cD_{\bA_k^{(n)}}^{(\max)})$ with
	$\supp\psi\subset B_{2\delta}[\gamma_k]$, then we have $\rho\psi\in \dom(\cD_{\bA_k^{(n)}}^{(\min)})$. Next,
	\begin{multline*}
		\norm{\cD_{\bA_k^{(n)}}(\rho_n\psi_0^{(n)}(f))}_{L^2(\S^3)^2}
		\leq \norm{\bsigma(\d\rho_n)\psi_0^{(n)}(f)}_{L^2(\S^3)^2} + \norm{\rho_n \cD_{\bA_k^{(n)}}\psi_0^{(n)}(f)}\\
		\leq \norm{\psi_0^{(n)}(f)}_{L^2(\S^3)^2} + \norm{\rho_n\cD_{\bA_k^{(n)}}(\psi^{(n)}(f))}_{L^2(\S^3)^2}
		+ \norm{\rho_n\cD_{\bA_k^{(n)}}\psi_{\sing}^{(n)}(f)}_{L^2(\S^3)^2}.
	\end{multline*}
	The first term tends to zero by \eqref{eq:q_0_to_0}. Recall the action of $\cD_{\bA_k^{(n)}}$ 
	in coordinates \eqref{eq:Q_rest_1}-\eqref{eq:Q_rest_3}.
	As in the part on \eqref{eq:Q_rest_1}, we get by direct computation that the last term converges to $0$ in $L^2$.
	We claim that the second term also vanishes. Indeed, by \eqref{eq:Dpsi_n_vanish},\eqref{def:w-l_approx_seq},
	we know that $\psi^{(n)},\cD_{\bA}\psi^{(n)}$ concentrates on $\gamma_k$. So by
	\eqref{eq:gauge_E_k}, $\cD_{\bA_k^{(n)}}(\psi^{(n)}(f))$ also concentrates on $\gamma_k$, hence 
	the claim follows.

	\smallskip
	\noindent\textbf{The term~\eqref{eq:Q_rest_3}:}
	It is through this expression that we will be able to deduce the sought-after condition.
	To ease the notational burden we will from now
	on drop all indices on the coordinates and simply write $(s,\rho,\theta)$.
	The condition $\int_{\S^3}|(\cD_{\bA^{(n)}}-\lambda)\psi_n|^2 \to 0$ implies that 
	$(\cQ^{(n)}-\lambda)f^{(n)}$ tends to $0$ in norm and both terms \eqref{eq:Q_rest_1} 
	\& \eqref{eq:Q_rest_2} tend to $0$ in norm. We obtain that the term
	\begin{equation}\label{eq:prep_reduc}
		\nt f_0^{(n)} + \big(\nt-\lambda + \cE_0^{(n)}
		+ \alpha_k^{(n)}\tfrac{\tau_n}{h_n}-\tfrac{c_k^{(n)}}{h_n}\big)f_{\sing}^{(n)}
	\end{equation}
	tends to zero in norm once we observe that: 1. $\wt{f}_0^{(n)}\to 0$ in norm, 
	2. the spinor $\sigma_3f_{\sing}^{(n)}$ equals $-f_{\sing}^{(n)}$ and 3. 
	$i\sigma_3 \tfrac{\tau_n}{h_n}\partial_{\theta} f_{\sing}^{(n)}\big|_{\theta\neq 0}
		= \alpha^{(n)}\tfrac{\tau_n}{h_n}f_{\sing}^{(n)}.$

	The connection form $\cE_0^{(n)}$ is a smooth matrix depending on the Seifert frame
	of $\gamma_k^{(n)}$ relatively to $S^{(n)}$. So we have:
	\[
	|\cE_0^{(n)}(s,\rho,\theta)-\cE_0^{(n)}(s,0,0)|\le \text{Cst}(k,n)\min(\rho,2\delta),
	\]
	where $\cE_0^{(n)}(s,0,0)$ denotes the value of $\cE_0^{(n)}$ at $\gamma_k^{(n)}(s)$.
	By geometric convergence of $\uS^{(n)}$ to $\uS$, we can assume that
	the constant $\text{Cst}(k,n)$ is uniform in $k$ and $n$. As we know that
	the $L^2$-norm of $\rho f_{\sing}^{(n)}$ tends to $0$, we can replace $\cE_0^{(n)}(s,\rho,\theta)$
	by $\cE_0^{(n)}(s)$. As the spinor $\wt{f}_{\sing}^{(n)}$ has no upper component, using \eqref{eq:con_form} 
	and \eqref{eq:def_xi}-\eqref{eq:rel_phase}, we can write:
	\begin{multline*}
	\cE_0^{(n)}(s,0,0)f_{\sing}^{(n)}(s,\rho,\theta)=\\i\big(\cip{\xi^{(n)}_-}{\nabla_{\bT^{(n)}}\xi_-^{(n)}}(s)
			-\cip{\xi^{(n)}_+}{\nabla_{\bT^{(n)}}\xi_-^{(n)}}(s)\sigma_1 \big)f_{\sing}^{(n)}(s,\rho,\theta),
	\end{multline*}
	Indeed, by the definition of the extension of the Seifert frame on $B_\eps[\gamma_k^{(n)}]$, 
	we have \emph{on the curve $\gamma_k^{(n)}$}:
	\[
	\nabla_{\bS^{(n)}}\bS^{(n)}=\nabla_{\bS^{(n)}}\bN^{(n)}=\nabla_{\bN^{(n)}}\bS^{(n)}=\nabla_{\bN^{(n)}}\bN^{(n)}=0.
	\]
	
	Similarly, as $|h_n^{-1}-1|\le \rho\times\mathrm{Cst}$, we can replace $h_n^{-1}$ by $1$ in \eqref{eq:prep_reduc}. 
	Thus, dropping all terms whose $L^2$-norm vanishes reduces \eqref{eq:prep_reduc} to
	\begin{equation}\label{eq:residuelle}
		\lVert\nt f_0^{(n)}
		+\big(-c_k^{(n)}+\cE_0^{(n)}(s) +\alpha_k^{(n)}\tau_n(s)-\lambda\big)
		f_{\sing}^{(n)}\rVert_{L^2}\to 0.
	\end{equation}
	
	We recall that $f^{(n)}=\chi(\tfrac{\rho}{2\delta})\wt{f}^{(n)}$.
	As $\norm{f_0^{(n)}}_{L^2}\to 0$ and $\wt{f}_{\sing}^{(n)},\cD_{\T_n,\alpha^{(n)}}\wt{f}_{\sing}^{(n)}$ 
	concentrate on $\T_n\times\{0\}$, then $\cD_{\T_n,\alpha^{(n)}}f_0^{(n)} $ concentrates
	on the same set, which implies that we can replace $f_0^{(n)}$ and $f_{\sing}^{(n)}$ by
	$\wt{f}_0^{(n)}$ and $\wt{f}_{\sing}^{(n)}$ in \eqref{eq:residuelle}.

	We are now left with the following problem: how can the term $\cD_{\T_n,\alpha^{(n)}}\wt{f}_0^{(n)}$ coming from an element $\wt{f}_0^{(n)}$ 
	\emph{in the minimal domain} approximate the expression~\eqref{eq:residuelle} (the 2nd line only)?
	
	We need to go into details. We denote by $f_{+}$ (resp. $f_{-}$) the spin up (resp. down) of a spinor $f$.
	For simplicity, we introduce the complex variable $z:=\rho e^{i\theta}$ in $\T_n\times\R^2$. 
	We rewrite \eqref{eq:residuelle} in components, recalling \eqref{eq:action_D_T} we have
	\begin{equation}\label{eq:residuelle_decouplee}
		\left\{
			\begin{array}{l}
				\Sup{\cD_{\T_n,\alpha^{(n)}}\wt{f}_0^{(n)}}
				+(F_{+}^{(n)}(s)-2i\partial_z)\Sdo{\wt{f}_{\sing}^{(n)}}\big|_{\theta\neq 0},\\
				\Sdo{\cD_{\T_n,\alpha^{(n)}}\wt{f}_0^{(n)}} +(F_{-}^{(n)}(s)
					+i\partial_s )\Sdo{\wt{f}_{\sing}^{(n)}}\big|_{\theta\neq 0}.
			\end{array}
		\right.
	\end{equation}
	where the $F_{\pm}^{(n)}(s)$ are defined as follows:
	\begin{equation*}
		\left\{
			\begin{array}{rcl}
				F_+^{(n)}(s)&:=&-i\cip{\xi_+^{(n)}}{\nabla_{\bT^{(n)}}\xi_-^{(n)}}(\gamma_k^{(n)}(s)),\\
				F_-^{(n)}(s)&:=&i\cip{\xi_-^{(n)}}{\nabla_{\bT^{(n)}}\xi_-^{(n)}}(\gamma_k^{(n)}(s))-c_k^{(n)}
				+\alpha_k^{(n)}\tau_k^{(n)}(s) -\lambda.
			\end{array}
		\right.
	\end{equation*}
	
	Our aim is to prove
	\begin{equation}\label{eq:cond_sur_s}
	\lim_{n\to+\infty}\norm{\big(F_{-}^{(n)}(s)+i\partial_s \big)\Sdo{\wt{f}_{\sing}^{(n)}}\big|_{\theta\neq 0}}=0,
	\end{equation}
	which we then show implies (S3).
	\subparagraph{\textit{Proof of \eqref{eq:cond_sur_s}}}
	We show that the inner product of $\Sdo{ \cD_{\T_n,\alpha^{(n)}}\wt{f}_0^{(n)}}$ with $\big(F_{-}^{(n)}\Sdo{\wt{f}_{\sing}^{(n)}}$
	resp. $\partial_s \Sdo{\wt{f}_{\sing}^{(n)}}\big|_{\theta\neq 0}$ tends to $0$: it implies \eqref{eq:cond_sur_s}. 
	For the first inner product we have:	
		\begin{equation}\label{eq:D_t_f_0_inner_F_-}
			\cip{\cD_{\T_n,\alpha^{(n)}}\wt{f}_0^{(n)}}{F_{-}^{(n)}\wt{f}_{\sing}^{(n)}}_{L^2}\!=\!
				-\cip{\wt{f}_0^{(n)}}{i(F_{-}^{(n)})'\wt{f}_{\sing}^{(n)}}_{L^2}
				+ \cip{\wt{f}_0^{(n)}}{F_{-}^{(n)}\cD_{\T_n,\alpha^{(n)}}\wt{f}_{\sing}^{(n)}}_{L^2}.
		\end{equation}
	As we know that $\norm{\wt{f}_0^{(n)}}_{L^2}\to 0$ and 
	$\norm{F_-^{(n)}}_{C^1(\T_n)}+\norm{\wt{f}_{\sing}^{(n)}}_{\bT_n}\le \mathrm{Cst}$,
	the above inner product converges to $0$.
	Now, by Cauchy-Schwarz inequality we have
	$$
		\lim_{n \to \infty}\big|\cip{\wt{f}_0^{(n)}}{\wt{f}_{\sing}^{(n)}}_{L^2}\big|=0.
	$$
	As $\wt{f}_0^{(n)}$ and $\wt{f}_{\sing}^{(n)}$ are $\norm{\,\cdot\,}_{\T_n}$-orthogonal, 
	the limit is the same for the scalar product 
	$$
		\cip{\cD_{\T_n,\alpha_k^{(n)}} \wt{f}_0^{(n)}}{\cD_{\T_n,\alpha_k^{(n)}} \wt{f}_{\sing}^{(n)}}_{L^2}.
	$$
	A direct calculation shows that
	\begin{equation}\label{eq:d_z_f_sing}
		\norm{2\partial_z\Sdo{\wt{f}_{\sing}^{(n)}}\big|_{\theta\neq 0}}_{L^2(\T_n\times\R^2)}^2
		=\left(\int_0^{\infty}K_{1-\alpha_{k}^{(n)}}(\rho)^2\rho\,\d\rho\right)\sum_{j\in \T_n^{*}}|\lambda_j^{(n)}|^2
		\underset{n\to \infty}{\longrightarrow}0,
	\end{equation}
	which implies
	\begin{equation}\label{eq:D_t_f_0_-avec_d_s_f_sing}
		\lim_{n \to \infty}
		\cip{\Sdo{ \cD_{\T_n,\alpha^{(n)}}\wt{f}_0^{(n)}}}
		  {i\partial_s \Sdo{\wt{f}_{\sing}^{(n)}}\big|_{\theta\neq 0}}_{L^2(\T_n \times\R^2)}
		=0.
	\end{equation}
	Hence \eqref{eq:cond_sur_s} holds. At last, we show the following.
	
	\subparagraph{\textit{\eqref{eq:cond_sur_s} implies (S3)}}
	
	We see $L^2(\T_n \times \R^2)^2$ as the tensor product 
	$L^2(\T_n)\otimes L^2(\R^2)^2$. The spectrum of the operator
	$i\partial_s +F_-^{(n)}(s)$, seen as an operator on $L^2(\T_n)$
	(or on $L^2(\T_n \times \R^2)^2$),
	is given by the \emph{discrete} set
	$$
		\frac{1}{\ell_k^{(n)}} \int_0^{\ell_k^{(n)}}F_-^{(n)}+\frac{2\pi}{\ell_k^{(n)}}\Z.
	$$
	From $\lim_{n\to\infty}\norm{\big(F_{-}^{(n)}(s)
		+i\partial_s \big)\Sdo{\wt{f}_{\sing}^{(n)}}}_{L^2}
		=0$ we know that
	$$
		\lim_{n \to \infty}\dist\left(0,\tfrac{1}{\ell_k^{(n)}} \int_0^{\ell}F_-^{(n)}+\tfrac{2\pi}{\ell_k^{(n)}}\Z\right)=0,
	$$
	or equivalently,
	$$
		\lim_{n \to \infty}\dist\big(\lambda,\spec\big(i\partial_s +i\cip{\xi_-^{(n)}}{\nabla_{\bT^{(n)}} \xi_-^{(n)}}-c_k^{(n)}
		+\alpha_k^{(n)}\tau_k^{(n)}\big)\big)=0.
	$$
	This condition implies that $\lambda$ is in the spectrum of the limit operator 
	$$
		i\partial_s +i\cip{\xi_-}{\nabla_{\bT} \xi_-}-c_k+\tau_{S_k},
	$$
	acting on the limit space $L^2(\T\times\R^2)^2$ with domain $H^1(\T)\otimes L^2(\R^2)^2$, 
	or on $L^2(\T)$ with domain $H^1(\T)$.
	Reintroducing the phase-jumps through the functions $E_k^{(n)}$, the condition is equivalent to 
	$$
		\lambda\in \spec\big(-D_{k,\bA} +i\cip{\xi_-}{\nabla_{\bT} \xi_-}+\tau_{S_k}\big)
	$$ 
	on $L^2(\T)$, that is 
	$\lambda\in \spec\big(\cT_{k,\bA} +\tau_{S_k}\big)$.

	\medskip
	\noindent\textbf{(S3) implies (S2):}
	We pick $\alpha_k^{(n)}\to 1^-$ and keep the Seifert surfaces and the other fluxes fixed. 
	In particular, we have
	$$
		\lim_{n \to \infty}\dist(\lambda,\spec(\cT_{k,\bA} +\alpha_k^{(n)}\tau_{S_k})) = 0.
	$$
	Let $(\wt{e}^{(n)}(s))_n$ be the normalized eigenfunction of $\cT_{k,\bA} +\alpha_k^{(n)}\tau_{S_k}$
	with eigenvalues approximating $\lambda$. Writing $z_n:=\rho_ne^{i\theta_n}$, the Ansatz
	$$
		\wt{e}^{(n)}(s_n)\chi\big(\frac{\rho_n}{2\delta}\big)\sqrt{\frac{1-\alpha_k^{(n)}}{2\pi\ell_k^{(n)}}}
		  \Big(1-\frac{i}{2}F_+^{(n)}(s_n)z_n\Big)\overline{z}_n^{-\alpha_k^{(n)}}\xi_-^{(n)}
	$$
	defines (up to normalization) the desired concentrating sequence of (S2): the function 
	$F_+^{(n)}(s_n)=-i\cip{\xi_+^{(n)}}{\nabla_{\bT^{(n)}}\xi_-^{(n)}}(\gamma_k^{(n)}(s_n))$
	is $C^1$-bounded by geometrical convergence of $\uS^{(n)}$. 
\qed

\subsubsection{Proof of Theorem~\ref{thm:bump_cont_bdry_B}}
	Recall that $\cD_{\wt{\bA}}$ and the induced operators $\cT_{\wt{\gamma}_k,\wt{\bA}}+\tau_{\wt{S}_{k}}$ 
	(with $k\in R$) have discrete spectrum.
	So we can find a common $2\eta>0$ that isolates $\lambda$ from the rest of the spectrum
	for \emph{all of them}: that is we can find $\eta$ with
	\begin{equation}\label{eq:split_spec}
	[\lambda-2\eta,\lambda+2\eta]\cap\spec\big(\wt{O}\big)\subset\{\lambda\},
	\end{equation}
	where $\wt{O}$ is the operator $\cD_{\wt{\bA}}$ or one of the $\cT_{\wt{\gamma}_k,\wt{\bA}}+\tau_{\wt{S}_{k}},k\in R$.
	And we write $I:=[\lambda-\eta,\lambda+\eta]$.
	
	We split the small ball $B_\eps[\wt{\uS}]\times B_\eps[\wt{\ua}]$ (with $\eps>0$ to be chosen),
	into the $2^{|R|}$ parts $C_\eps(\wt{\ua},R'),\,R'\subset R$.
	Up to taking $\eps>0$ small enough, we first assume that for all 
	$(\uS,\ua)\in B_\eps[\wt{\uS}]\times B_\eps[\wt{\ua}]$, $\bA:=\sum_k 2\pi\alpha_k[S_k]$, we have
	\begin{equation}\label{eq:split_spec2}
	[\lambda-2\eta,\lambda+2\eta]\cap\spec\big(O\big)\subset[\lambda-\eta,\lambda+\eta],
	\end{equation}
	where $O$ denotes either $\cD_{\bA}$ or $\cT_{k,\bA}+\tau_{S_k}$ for $k\in R$ such that $\alpha_k=0$.
	Such an assumption is possible thanks to Theorem~\ref{thm:bump_cont_bdry_A} and Proposition~\ref{prop:calcul_spectre}.
	
	We now study the continuity in the region $B_\eps[\wt{\uS}]\times C_\eps(\wt{\ua},R')$, 
	on which the dimension $d(\uS,\ua,\lambda)$ of $V$ (in Theorem~\ref{thm:bump_cont_bdry_B}(1)) is constant. 
	By strong resolvent continuity, the defect of bump-continuity is only due to the vanishing of eigenfunctions.
	We need the following auxiliary result, which is proved at the end.
		
	\begin{lemma}\label{lem:aux_result}
		Let $(\uS^{(n)},\ua^{(n)})_{n\ge 1}$ be a sequence of $B_\eps[\wt{\uS}]\times C_\eps(\wt{\ua},R')$ 
		converging to $(\uS^{(\infty)},\ua^{(\infty)})$, with 
		$\ua^{(\infty)}\in\Gamma(R_0)\cap \rT_b(K)$ and $R_0\subset R$.
		Let $d(R',R,\lambda)$ be the dimension $d(\uS,\ua,\lambda)$ on
		$B_\eps[\wt{\uS}]\times C_\eps(\wt{\ua},R')$. There exists a subspace $V_n$ 
		of eigenfunctions of $\cD_{\bA^{(n)}}$ with dimension $d(R',R,\lambda)-d(\uS^{(\infty)},\ua^{(\infty)},\lambda)$, 
		such that
		\[
		\mathds{1}_I(\cD_{\bA^{(n)}}) -P_{V_n},\ \cD_{\bA^{(n)}}\big(\mathds{1}_I(\cD_{\bA^{(n)}}) -P_{V_n}\big)
		\]
		converge in norm to $\mathds{1}_{[\lambda-\eta,\lambda+\eta]}(\cD_{\bA^{(\infty)}})$ and 
		$\cD_{\bA^{(\infty)}}\mathds{1}_{[\lambda-\eta,\lambda+\eta]}(\cD_{\bA^{(\infty)}})$ respectively.
		Furthermore the eigenvalues of the eigenfunctions of $V_n$ converge to 
		the $\lambda_k^{(\infty)}$ of Theorem~\ref{thm:bump_cont_bdry_B}(4), and $P_{V_n}$
		converges to $0$ in the strong operator topology.
	\end{lemma}

	\paragraph{\textit{Definition of $P_W(\bA,\lambda)$}}
	
	Let $R_0\subset R$, we will call $B_\eps[\wt{\uS}]\times(B_\eps[\wt{\ua}]\cap\Gamma(R_0))$ 
	the $\Gamma(R_0)$-boundary (in $T_b(K)\cap\Gamma(R)$). Pick $(\uS,\ua)\in B_\eps[\wt{\uS}]\times C_\eps(\wt{\ua},R')$.
	According to Lemma~\ref{lem:aux_result}, we try to define the projection $P_{R_0}(\bA,\lambda)$ which corresponds 
	to the part of $\mathds{1}_I(\cD_{\bA})$ which converges in norm as $(\uS,\ua)$ tends to the $\Gamma(R_0)$-boundary.
	
	We use the projection onto a closed convex set. For $(\uS,\ua)$ in a small ball $B_\eps[(\wt{\uS},\wt{\ua})]$
	we define the ``trace on the $\Gamma(R_0)$-boundary''
	\begin{equation}\label{eq:def_trace}
		P_{R_0}(\bA,\lambda):=\mathds{1}_{I}\big[\cD_{\bA_0}\big],
		\quad \bA_0:=\sum_{k\notin R_0}2\pi\alpha_k [S_k].
	\end{equation}
	This operator is norm-continuous for $\eps$ small enough: it is ensured by the strong 
	resolvent-continuity of $\cD_{\bA_0}$ and the fact that the rank of 
	$P_{R_0}(\bA,\lambda)$ is constant as $\eta>0$ 
	satisfies \eqref{eq:split_spec}-\eqref{eq:split_spec2}.
	
	Let $\sC_{\bA}(R_0)$ be the \emph{compact convex set} (in the Hilbert-Schmidt norm):
	\[
	 \sC_{\bA}(R_0):=\Big\{\omega,\ 0\le \omega\le \mathds{1}_{I}(\cD_{\bA}), [\cD_{\bA},\omega]=0\ \&\ \tr(\omega)=\tr(P_{R_0}(\bA,\lambda))\Big\}.
	\]
	We project $P_{R_0}(\bA,\lambda)$ onto $\sC_{\bA}(R_0)$ w.r.t. the Hilbert-Schmidt norm $\norm{\,\cdot\,}_{\mathfrak{S}_2}$:
	\begin{equation}\label{eq:def_proj}
	      \omega_{R_0}(\bA,\lambda):=\mathrm{Proj}_{\sC_{\bA}(R_0)}(P_{R_0}(\bA,\lambda))
	      	=\underset{\omega\in \sC_{\bA}(R_0)}{\mathrm{argmin}}\,\tr\,(\omega-P_{R_0}(\bA,\lambda))^2.
	\end{equation}
	
	Let us now freeze $\uS$ and take the limit $\ua\to \ua'\in\Gamma(R_0)$: the point $(\uS,\ua')$ defines the potential $\bA'$
	and we write $P':=\mathds{1}_I(\cD_{\bA'})$. By Theorem~\ref{thm:strg_res_cont} and Lemma~\ref{lem:char_sres_conv},
	$\omega_{R_0}(\bA,\lambda)$ converges to $P'$: for any convergent sequence $\ua^{(n)}\to \ua'$, 
	we can construct a \emph{projector} $P^{(n)}\in \sC_{\bA^{(n)}}(R_0)$ converging to $P'$. 
	We can assume $[P^{(n)},\cD_{\bA^{(n)}}]=0$
	as the Dirac operators have discrete spectrum. In particular, we have: $\lim_{\ua\to \ua'}\tr\{\omega_{R_0}(\bA,\lambda)(1-\omega_{R_0}(\bA,\lambda))\}\to 0$.
	
	Thus there exists $0<\delta_1<2^{-1}$ such that for $\dist(\ua,\ua_0)<\delta_1$, the projector 
	$\mathds{1}_{[3/4,1]}(\omega_{R_0}(\bA,\lambda))=:P_{W_{R_0}}(\bA,\lambda)$ has rank $\tr(P')$ 
	and converges to $P'$ as $\ua\to \ua'$. As $\ran\,\omega_{R_0}(\bA,\lambda)$ is spanned by eigenfunctions of $\cD_{\bA}$, 
	by Theorem~\ref{thm:strg_res_cont} and Lemma~\ref{lem:char_sres_conv}, the range $W_{R_0}(\bA,\lambda)$ of $P_{W_{R_0}}(\bA,\lambda)$
	is also spanned by eigenfunctions. For $R=R_0$, this defines the subspace $W(\bA,\lambda)$. 
	The subspace $V(\bA,\lambda)$ is defined through its projection by:
	\[
	 P_{V(\bA,\lambda)}:=\mathds{1}_I(\cD_{\bA})-P_{W(\bA,\lambda)}.
	\]
	We now study the continuity of these operators.
	
	\smallskip
	
	\paragraph{\textit{Continuity of $\omega_{R_0}$ and $P_{W}$}}
	\subparagraph{\textit{Continuity of $\omega_{R_0}$ in the bulk}}
	We check sequential continuity. Theorem~\ref{thm:strg_res_cont} and Lemma~\ref{lem:char_sres_conv}
	ensure that for any sequence $(\uS^{(n)},\ua^{(n)})\to (\uS',\ua')$ where $(\uS',\ua')$ is in the bulk $B_\eps[\wt{\uS}]\times(B_\eps[\wt{\ua}]\setminus T_b(K))$,
	we can find a sequence $\omega^{(n)}\in\sC_{\bA^{(n)}}(R_0)$ with $\tr(\omega_{R_0}(\bA',\lambda)-\omega^{(n)})^2\to 0$.
	By norm-continuity of $P_{R_0}(\cdot,\lambda)$ in a neighborhood of $(\uS',\ua')$, 
	we have continuity of $\omega_{R_0}(\cdot,\lambda)$ at $(\uS',\ua')$.

	\subparagraph{\textit{Continuity of $\omega_{R_0}$ on the $\Gamma(R_0)$-boundary}}
	We recall that for $(\uS',\ua')$ on the $\Gamma(R_0)$-boundary: as $(\uS,\ua)\to(\uS',\ua')$,
	the operator $\omega_{R_0}(\bA,\lambda)$ converges to $\mathds{1}_I(\cD_{\bA'})$. 
	Now, let $Q_{R_0}$ be a compact set in $B_\eps[\wt{\uS}]\times (B_\eps[\wt{\ua}]\cap\Gamma(R_0))$.
	With the help of Lemma~\ref{lem:aux_result} and standard technique we easily get the following uniform convergence result.
	For any $a>0$ there exists $0<r<\eps$ such that:
	for any $(\uS,\ua)$ in the neighborhood $(B_{\eps}[\wt{\uS}]\times B_{\eps}[\wt{\ua}])\cap B_{r}[Q_{R_0}]$ of $Q_{R_0}$, we have:
	\begin{equation}\label{eq:unif_cont}
	 \norm{\omega_{R_0}(\bA,\lambda)-\mathds{1}_{I}(\cD_{\bA_0})}_{\mathfrak{S}_2}<a,\ \bA_0:=\sum_{k\notin R_0}2\pi\alpha_k[S_k].
	\end{equation}

	\subparagraph{\textit{Continuity of $P_{W_{R_0}}$ near $Q_{R_0}$}}
	Let $Q_{R_0}\subset B_\eps[\wt{\uS}]\times (B_\eps[\wt{\ua}]\cap\Gamma(R_0))$ be a \emph{compact} set.
	By the uniform convergence of $\omega_{R_0}(\bA,\lambda)$ on some $B_{r}[Q_{R_0}]$, 
	the positive number $0<\delta_1<2^{-1}$ used to define $P_{W_{R_0}}$ can be chosen uniformly on 
	$B_{r}[Q_{R_0}]$, and the $\norm{\,\cdot\,}_{\mathfrak{S}_2}$-continuity of $P_{W_{R_0}}$ on this set follows from that of
	$\omega_{R_0}(\bA,\lambda)$. For $R_0=R$ and $Q_{R_0}=\{(\wt{S},\wt{\ua})\}$, we obtain continuity of 
	$P_{W(\bA,\lambda)}$ in some ball $B_{\eps}[(\wt{S},\wt{\ua})]$ up to taking $\eps>0$ small enough.
	
	\subparagraph{\textit{Continuity of the other operators}}
	
	Theorem~\ref{thm:strg_res_cont} and $\mathfrak{S}_2$-continuity of $P_{W}$ ensure the
	norm-continuity of $\cD_{A} P_W$ (the rank of $P_W$ is bounded). 
	By Theorem~\ref{thm:strg_res_cont}, Lemma~\ref{lem:char_sres_conv} and our choice of $\eta>0$ (\eqref{eq:split_spec}-\eqref{eq:split_spec2}), 
	we have strong continuity of $\mathds{1}_{I}(\cD_{\bA})$ in a small ball $B_\eps[(\wt{\uS},\wt{\ua})]$,
	hence we have strong continuity of $P_{V(\bA,\lambda)}=\mathds{1}_{I}(\cD_{\bA})-P_W$. 
	As $V(\bA,\lambda)$ is spanned by eigenfunctions of $\cD_{\bA}$ with eigenvalues in $I=\lambda+[-\eta,\eta]$,
	Theorem~\ref{thm:strg_res_cont} gives strong continuity of $\cD_{\bA}P_{V(\bA,\lambda)}$ (using
	the discreteness of the spectrum): it suffices to check sequential continuity.

	By Theorem~\ref{thm:bump_cont_bulk}, norm-continuity of $P_{W(\bA,\lambda)}$ and \eqref{eq:split_spec}-\eqref{eq:split_spec2},
	the dimension of $V(\bA,\lambda)$ is constant on each connected subset $B_\eps[\wt{\uS}]\times C_\eps(R')$: 
	by Lemma~\ref{lem:aux_result}, it is equal to $d(R',R,\lambda)$. 
	For $R_0\subset R$, the continuity of the corresponding eigenvalues 
	on the $\Gamma(R_0)$-boundary is given by the following argument.
	First, when $(\uS,\ua)$ converges to a point on the $\Gamma(R_0)$-boundary, 
	the part $P_{W_{R_0}}(\bA,\lambda)$ converges in norm.
	The remaining part converges weakly to $0$ by Lemma~\ref{lem:aux_result}.
	Using sequential continuity together with the spectral study of Theorem~\ref{thm:bump_cont_bdry_A} and Lemma~\ref{lem:aux_result}, 
	we easily obtain the convergence of the eigenvalues of the vanishing part $\cD_{\bA}(P_{V(\bA,\lambda)}-P_{W_{R_0}}(\bA,\lambda))$.
	Indeed, these eigenvalues cannot escape $I=[\lambda-\eta,\lambda+\eta]$ by assumption, and in the limit, each one must be
	an eigenvalue of the induced operator of one of the knots whose flux has converged to $1^{-}$.
	The computation of the spectrum of the induced operators in Proposition~\ref{prop:calcul_spectre} ensures the continuity
	of their limits on the $\Gamma(R_0)$-boundary.	
	The continuity of the family 
	\[
		G(\bA,\lambda):=\cD_{\bA}(1-P_{V(\bA,\lambda)})+ \mu P_{V(\bA,\lambda)}-\lambda
	\]
	in the neighborhood of $(\wt{\uS},\wt{\ua})$ follows easily. The strong continuity of $G(\bA,\lambda)$
	follows from that of $\cD_{\bA},P_V$ and $\cD_{\bA}P_V$. 
	As we have shifted the eigenvalues carried by the elements in $V$, this shows that 
	$\mathds{1}_{[-\eta,\eta]}\big(G(\bA,\lambda)\big)$ is norm-continuous in a neighborhood of $(\wt{\uS},\wt{\ua})$. 
	Up to picking a smaller $\eps$, this proves the bump-continuity of $G(\bA,\lambda)$. 
	There only remains to prove Lemma~\ref{lem:aux_result}
	to end the proof of Theorem~\ref{thm:bump_cont_bdry_B}.\hfill\qed

	\medskip

	\paragraph{\textit{Proof of Lemma~\ref{lem:aux_result}}}\ 
	
	We first assume $(\uS^{(\infty)},\ua^{(\infty)})=(\wt{\uS},\wt{\ua})$, 
	the general case is similar.

	\subparagraph{\textit{Ansatz for vanishing approximate eigenfunctions}}
	Let $(\psi^{(n)})_n$ be a vanishing sequence along $(\uS^{(n)},\ua^{(n)})_n$ (see \eqref{def:w-l_approx_seq}).
	It exists as long as $d(R',R,\lambda)>0$ thanks to the spectral study of Theorem~\ref{thm:bump_cont_bdry_A},
	and the equivalence $(S1)\iff (S2)$ in its proof see \eqref{def:w-l_approx_seq}. 
	We localize around a knot $\gamma_k^{(n)}$ for which the flux
	$2\pi\alpha_k^{(n)}$ tends to $2\pi$. We use the same notations as the proof of Theorem~\ref{thm:bump_cont_bdry_A}.
	We introduce the same Ansatz as in the part (S3) implies (S2).
	
	The function $E_k^{(n)}$ is the phase jump function for $\gamma_k^{(n)}$ (Proposition~\ref{prop:sa_link}).
	The gauge $e^{i\zeta_k^{(n)}}$ locally around $\gamma_k^{(n)}$ shifts the phase jump across $S_k^{(n)}$
	from this surface to $\{\theta_n=0\}\cap B_{\eps}[\gamma_k^{(n)}]$ (see Lemma~\ref{lem:control_Ngr}).
	We consider $e_0^{(n)}(s)$ the normalized approximate zero mode corresponding to the operator 
	$i\partial_s +F_{-}^{(n)}(s)$ on $L^2(\T_n)$. The function $E_k^{(n)}(\gamma_k^{(n)}(s))e_0^{(n)}(s)$ 
	is the corresponding $\lambda$-approximate eigenfunction corresponding to the following operator defined on $\dom(D_{k,\bA^{(n)}})$: 
	$
		i\partial_s +i\cip{\xi_-^{(n)}}{\nabla_{\bT^{(n)}} \xi_-^{(n)}}(\gamma_k^{(n)}(s))+\alpha_k^{(n)}\tau_k^{(n)}(s)
	$
	(see Section~\ref{sec:effective_operator} for the definition of $D_{k,\bA^{(n)}}$). We define the Ansatz
	\begin{equation}\label{eq:ansatz_eig}
		\wt{\psi}_{k,\textrm{ans}}^{(n)}(\bp):=E_k^{(n)}e^{i\zeta_k^{(n)}}f_{-,\textrm{ans}}^{(n)}\big((s_n,\rho_n,\theta_n)(\bp)\big)\xi_-^{(n)}(\bp),	
	\end{equation}
	where, writing $z_n:=\rho_n e^{i\theta_n}$, 
	$f_{-,\textrm{ans}}^{(n)}$ is defined as follows:
	\begin{multline*}
		f_{-,\textrm{ans}}^{(n)}(s_n,\rho_n,\theta_n)
		:=e_0^{(n)}(s_n)\chi\big(\frac{\rho_n}{2\delta}\big)
		\times\sqrt{\frac{1-\alpha_k^{(n)}}{2\pi\ell_k^{(n)}}}\Big(1-\frac{i}{2}F_+^{(n)}(s_n)z_n\Big)\overline{z}_n^{-\alpha_k^{(n)}}.
	\end{multline*}
	By construction we have $\psi_{k,\textrm{ans}}^{(n)}\in\dom(\cD_{\bA^{(n)}})$, and we can check using \eqref{eq:rappel} that it 
	concentrates on $\gamma_k$ as $n$ tends to $0$, and that $\norm{(\cD_{\bA^{(n)}}-\lambda)\psi_{k,\textrm{ans}}^{(n)}}_{L^2}$ tends to $0$.
	In other words, it satisfies \eqref{def:w-l_approx_seq} up to normalization.
	
	\smallskip
	
	\subparagraph{\textit{Maximality of the dimension of $V_n$}}
	Let us prove: if $\psi^{(n)}$ concentrates around $\gamma_k^{(n)}$ 
	with $\norm{(\cD_{\bA^{(n)}}-\lambda)\psi^{(n)}}_{L^2}\to 0$, 
	then the vanishing part of its localization around this knot 
	is essentially colinear to the above Ansatz.
	
	We consider $(\chi_{\delta,\gamma_k^{(n)}}\psi^{(n)})_n$ and define  the 
	functions $\wt{f}_0^{(n)}$ and $\wt{f}_{\sing}^{(n)}(\ul^{(n)})$ in $\dom(\cD_{\T_n,\alpha^{(n)}})$,
	as in (S2) implies (S3) in the proof of Theorem~\ref{thm:bump_cont_bdry_A}. 
	We also call $f_0^{(n)}$ and $f_{\sing}^{(n)}$ their multiplication by $\chi(\tfrac{\rho_k^{(n)}}{2\delta})$
	(recall also \eqref{eq:rappel} and \eqref{eq:Q_rest_1}-\eqref{eq:Q_rest_3}).
	We claim that \eqref{eq:cond_sur_s} also holds here:
	\begin{equation}\label{eq:similarly_cond_s}
		\lim_{n\to+\infty}\norm{(i\partial_s +F_{-}^{(n)}(s))\Sdo{f_{\sing}^{(n)}}}_{L^2}=0.
	\end{equation}
	
	This claim is proved at the end. We first show the colinearity result.
	
	Let $\cD_{0,\alpha^{(n)}}$ be the 2D-Dirac operator with Dirac-point magnetic field $2\pi\alpha^{(n)}\delta_0$ 
	(see \cite{Persson06} for instance). It is the self-adjoint operator on $L^2(\R^2)^2$ acting on spinors with the phase jump
	$e^{-2i\pi\alpha^{(n)}}$ across $\{\theta=0\}$, acting like $-i\sigma_1\partial_{u_1}-i\sigma_2\partial_{u_2}$
	on $\{\theta\neq 0\}$, and whose domain contains $w_{\sing}^{(n)}:=(0, K_{\alpha^{(n)}}(\rho)e^{i\alpha^{(n)} \theta})^\rT$.
	This last element spans the singular domain of $\cD_{0,\alpha^{(n)}}$. On its graph-norm complement, the minimal domain,
	the Lichnerowicz formula is satisfied: for $w\in\dom(\cD_{0,\alpha^{(n)}}^{(\min)})$, we have:
	\[
	\int|\cD_{0,\alpha^{(n)}}w|^2=\int|(\nabla w)_{|_{\theta\neq 0}}|^2.
	\]
	Using \eqref{eq:dpm_decouple}, it is easy to see that 
	$H^1(\T_n)\otimes\dom(\cD_{0,\alpha^{(n)}})$ is dense in the domain $\dom(\cD_{\T_n,\alpha^{(n)}})$.
	On this latter domain, the graph norm $\norm{\,\cdot\,}_{0,\alpha^{(n)}}$ of $\cD_{0,\alpha^{(n)}}$
	defines an intermediate norm between $\norm{\,\cdot\,}_{L^2}$ and  $\norm{\,\cdot\,}_{\T_n}$:
	\[
	\norm{f}_{0,\alpha^{(n)}}^2:=\int|f|^2+\int|(\wt{\sigma}\cdot\wt{\nabla}_u f)_{|_{\theta\neq 0}}|^2
		=\norm{f}_{\T_n}^2-\int |(\partial_s f)_{|_{\theta\neq 0}}|^2.
	\]
	We introduce the $\norm{\,\cdot\,}_{0,\alpha^{(n)}}$-Hilbert basis
	$$
		(e_{j_1}^{(n)}(s)w_{j_2}^{(n)}(\rho,\theta))_{j_1,j_2},
	$$ 
	where $(e_{j_1}^{(n)})_{j_1\in \T_n^*}$ is an eigenbasis for $i\partial_s +F_-^{(n)}(s)$ 
	on $L^2(\T_n)$, $(w_{j_2})_{j_2\ge 0}$ is a Hilbert basis for $\norm{\,\cdot\,}_{0,\alpha^{(n)}}$ on 
	$L^2(\R^2)^2$, where $w_0$ is co-linear to $w_{\sing}^{(n)}$.
	
	We decompose $f_{\sing}^{(n)}$ with respect to this Hilbert basis. 
	The operator $i\partial_s +F_-^{(n)}(s)$ has discrete spectrum which is uniformly spaced in $n$,
	and $(w_{j_2})_{j_2\ge 1}$ is a Hilbert basis for the minimal domain of $\cD_{0,\alpha^{(n)}}$.
	We thus obtain the existence of a sequence $(\mu_n)_n$ such that
	$$
		\lim_{n\to\infty}\norm{f_{\sing}^{(n)}-\mu_nf_{-,\textrm{ans}}^{(n)}\begin{pmatrix}0&  1\end{pmatrix}^{\rT}}_{\T_n}=0
	$$
	with 
	$$
		0<\liminf_{n\to\infty}|\mu_n|\le \limsup_{n\to+\infty}|\mu_n|<+\infty.
	$$
	We have used the fact that $H^1(\T_n)\otimes \dom(\cD_{0,\alpha_k^{(n)}})$ is $\norm{\,\cdot\,}_{\T_n}$-dense.
	Thus the rest $\chi(\delta^{-1}\rho_k^{(n)})\psi^{(n)}-\mu_nE_k^{(n)}\psi_{k,\sing}^{(n)}(f_{\textrm{ans}}^{(n)})$ does not collapse
	onto the knot $\gamma_k^{(\infty)}$, or equivalently, Theorem~\ref{thm:compactness} applies to it and along any converging
	subsequence, there is no loss of mass.
	
	Thus, for each $\gamma_k^{(n)}$ with $\alpha_k^{(n)}\to 1^-$, we can remove the collapsing part, and the rest 
	converges (up to a subsequence) without loss of mass.
	
	\smallskip
	
	\noindent\underline{\textit{Proof of \eqref{eq:similarly_cond_s}}}
	By Lemma~\ref{lem:control_Ngr}, $\wt{f}_{0}^{(n)},\wt{f}_{\sing}^{(n)}$ are $\norm{\,\cdot\,}_{\T_n}$-bounded. So, first by 
	\eqref{eq:dirac_dirichlet}, $\wt{f}_0^{(n)}(\ell_k^{(n)}(\ell_k^{(\infty)})^{-1}s,\rho,\theta)$ converges in $L^2(\T_{\infty}\times\R^2)^2$
	up to extraction of a subsequence. Then we have: $\sum_{j\in\T_n^*}|\lambda_j^{(n)}|^2=\mathcal{O}(1-\alpha_k^{(n)})$.
	As in \eqref{eq:collapse_sing}, by direct computation we obtain that both $\wt{f}_{\sing}^{(n)}$ and 
	$\cD_{\T_n,\alpha_k^{(n)}}\wt{f}_{\sing}^{(n)}$ collapse onto $\T_n\times\{0\}$.
	Consider now the splitting \eqref{eq:Q_rest_1}-\eqref{eq:Q_rest_3} relative to $\cD_{\bA^{(n)}}\psi^{(n)}$.
	The term \eqref{eq:Q_rest_1} tends to $0$ by direct computation.
	We now use the collapse of the singular part: it implies that of $\cD_{\T_n,\alpha_k^{(n)}}f_{\sing}^{(n)}$.
	Furthermore the inner product of \eqref{eq:Q_rest_2} with \eqref{eq:Q_rest_3} minus 
	$\cD_{\T_n,\alpha_k^{(n)}}\wt{f}_0^{(n)}$ tends to $0$: $(\nabla f_0^{(n)})_{|_{\theta\neq 0}}$ is $L^2$-bounded and
	by direct computation: $\norm{\rho \big( |f_{\sing}^{(n)}|+|\cD_{\T_n,\alpha_k^{(n)}}f_{\sing}^{(n)}|\big)}_{L^2}\to 0.$
	
	Using that \eqref{eq:D_t_f_0_inner_F_-} converges to $0$, \eqref{eq:d_z_f_sing} and \eqref{eq:D_t_f_0_-avec_d_s_f_sing}
	in \eqref{eq:Q_rest_3} we get \eqref{eq:similarly_cond_s}.
	
	\smallskip

	\subparagraph{\textit{Conclusion}}
	From the Ansatz, we get that there is at least a $d(R',R,\lambda)$-dimensional plane
	$V_n'$ of $\lambda$-approximate eigenfunctions vanishing in the limit $n\to+\infty$ 
	(one dimension for each of the knots $\gamma_k$, $k\in R'$ with $\lambda\in\spec(\cT_{k,\wt{\bA}}+\tau_{\wt{S}_k})$). 
	As the operators $\cD_{\bA^{(n)}}$ have discrete spectrum,
	up to projecting onto the eigenspaces, we thus obtain a $d(R',R,\lambda)$-dimensional subspace $V_n$
	of eigenfunctions, with eigenvalues converging to $\lambda$ in the limit.
	The fact that this dimension is maximal follows from the results on the approximate collinearity of the vanishing
	sequence to the Ansatz. This maximality implies the norm-convergence of $(\mathds{1}_I(\cD_{\bA^{(n)}})-P_{V_n})$
	and that of $\big(\cD_{\bA^{(n)}}(\mathds{1}_I(\cD_{\bA^{(n)}})-P_{V_n})\big)$.
	The fact that $P_{V_n}$ converges to $0$ is due to the fact that the functions of its range
	collapse onto the knots, weakly converging to $0$ in $L^2$ (and that $\dim V_n=d(R',R,\lambda)$ is finite).
	
	\subparagraph{\textit{Extension to other points $(\uS^{(\infty)},\ua^{(\infty)})$}}
	We now assume that the sequence converges to another point $(\uS^{(\infty)},\ua^{(\infty)})$.
	We can adapt the previous parts. For all knots $\gamma_{k}^{(\infty)}$ such that $\alpha_{k}^{(\infty)}=0$ and
	\[
	[\lambda-\eta,\lambda+\eta]\cap\spec\big(\cT_{k,\bA^{(\infty)}}+\tau_{S_{k}^{(\infty)}}\big)=\{\lambda_k^{(\infty)}\}\neq \emptyset,
	\]
	we pick a similar Ansatz as in \eqref{eq:ansatz_eig}, but where
	the function $e_0^{(n)}(s_n)$ is defined from the eigenfunction of $\cT_{k,\bA^{(\infty)}}+\tau_{S_{k}^{(\infty)}}$ associated to
	$\lambda_k^{(\infty)}$. There are $d(R',R,\lambda)-d(\uS^{(\infty)},\ua^{(\infty)},\lambda)$ of such knots, which gives a vanishing subspace $V^{(n)}$
	of the same dimension, spanned by approximate eigenfunctions (associated to the $\lambda_k^{(\infty)}$). Once again, we can assume
	that they are eigenfunctions of $\cD_{\bA^{(n)}}$ by discreteness of the spectrum. The maximality 
	(in the spectral region $[\lambda-\eta,\lambda+\eta]$) is proved in a similar manner.
\qed

\appendix
\section{}

\subsection{Integrated relative torsion, normal holonomy \& writhe of a knot}\label{sec:techn}
We define and study in \ref{sec:integrated_relative_torsion} the integrated relative torsion of a knot $\gamma$,
and in \ref{sec:link_writhe} we show that it is equal to $-2\pi$ times its writhe. In \ref{sec:varying_I_tau}
we show how we can vary these numbers by adding twists to $\gamma$.
\subsubsection{Trivial relative degree of Seifert frames}
In this part we consider a knot $\gamma:(\R/\ell\Z)\to \S^3$ with Seifert surface $S$ and Seifert frame $(\bT,\bS,\bN)$, framing $\gamma$.
Note that for $0<\eps\ll 1$, the parallel curve $\mathrm{exp}_{\gamma(s)}^{S}(\eps\bS(s))$ to the knot $\gamma$ has trivial homology in $H^1(\S^3\setminus\gamma)$
since it lies in the Seifert surface $S$ ($\mathrm{exp}^{S}$ denotes the exponential map with respect to the induced metric on $S$). 
This holds for all Seifert framing of $\gamma$. This implies that the two Seifert frames $(\bT,\bS,\bN)$ and $(\bT,\bS',\bN')$ of $\gamma$
have trivial relative degree. That is, writing
\[
	g_3(\bS,\bS')+ig_3(\bS,\bN')=e^{i\vartheta},
\]
$\vartheta$ has degree zero as a function from $\R/\ell\Z\simeq \S^1$ to $\R/(2\pi\Z)\simeq \S^1$ (using for instance the fact that $\link(\cdot,\gamma)$
defines an isomorphism of $H^1(\S^3\setminus\gamma)$ onto $\Z$).
Another proof can be easily derived from \cite{MR916076}.

\subsubsection{Integrated relative torsion of a knot}\label{sec:integrated_relative_torsion}
In this section we recall the definition of the normal holonomy of a knot $\gamma\subset \S^3$ and show 
that the integrated (or total) relative torsion associated to a Seifert frame
$$
	\int_{\gamma}\tau\,\bT^{\flat}
	=I_\tau(\gamma) \in \R
$$ 
does not depend on the choice of the Seifert surface. We recall that given a Seifert surface $S$ for $\gamma$,
the relative torsion is $\tau=\cip{\nabla_{\bT}\bS}{\bN}$, where $(\bT,\bS,\bN)$ is the Darboux frame of $\gamma$ 
on $S$ (see Section~\ref{sec:seif_fram_loc_coord}). As above, we consider
two Seifert surfaces $S,S'$ for $\gamma$ (with basepoint $\bp_0=\gamma(0)$). 
The Levi-Civita connection induces a canonical 
connection
$\nabla^{\cN}$ on the normal bundle $\cN \gamma$, defined by
$$
	\nabla^{\cN}_X Y:= (1-P_{\bT}) \nabla_X Y,
$$
where $P_{\bT}$ is the (pointwise) projection onto $\R\,\bT$. 
We now pick any unit vector $X(0)$ in $\bT(\bp_0)^{\perp}\subset \rT_{\bp_0} \S^3$ 
(say $\bS(0)$) and write $X(s)$ for the parallel transport of $X(0)$ 
along $\gamma$ in $\cN \gamma$ with respect to $\nabla^{\cN}$. 
After one full turn, the angle $\theta_{\hol}:=\measuredangle (X(0),X(\ell))$ is the normal holonomy 
$\mathrm{Hol}(\gamma)$ \cite{HebdaTsau}. As we deal with the parallel transport, 
it does not depend on the choice of $X(0)$. 
If we choose $X(0)=\bS(0)$, we obtain
$$
	X(\ell)=\cos(\theta_{\hol})\bS(0)+\sin(\theta_{\hol})\bN(0).
$$
More precisely, the normal holonomy is the isometry (on the fiber of $\cN_{\gamma}$ at point $\bp_0$) 
obtained after the parallel transport along the loop $\gamma$. 
Since this fiber is a real plane, we can see the holonomy as the angle of the rotation in $\SO(2)$.

We can now interpret the torsion geometrically; it is simply the derivative of the angle 
$\phi_{S}(s):=\measuredangle (X(s),\bS(s))\in \R/(2\pi\Z)$ (in the \emph{oriented} fiber):
$$
	\d \phi_{S}=\tau(s)\,\d s.
$$
The angle $\phi_{S}$ depends on $X$ but its derivative does not. If we consider another 
Seifert surface $S'$ we have
\[
	\phi_{S'}=\phi_S+\vartheta_{S,S'},
\]
and thus, writing $\tau_S$ and $\tau_{S'}$ the relative torsions with respect to $S$ and $S'$
\[
	\int_{\gamma}\tau_{S'}\,\bT^{\flat}
	=\int_{\gamma}\tau_S\,\bT^{\flat}+\int_{\gamma}\d\vartheta_{S,S'}
	=\int_{\gamma}\tau_S\,\bT^{\flat}.
\]
In particular we have:
\begin{equation}
	\mathrm{Hol}(\gamma)=-I_\tau(\gamma)\mod 2\pi.
\end{equation}
Observe the following:
\begin{enumerate}
	\item Assume that the Frenet frame $(\bT,\bN_{\gamma},\bB_{\gamma})$ of $\gamma$ is 
	defined everywhere on $\gamma$. We recall that the Frenet frame is defined at $p\in\gamma$
	if $\big(\nabla_{\bT}\bT\big)(p)\neq 0$. At this point the frame is defined by
	$$
		\nabla_{\bT}\begin{pmatrix} \bT\\ \bN_{\gamma}\\\bB_{\gamma}\end{pmatrix}=
		\begin{pmatrix}0& \kappa_\gamma& 0 \\ -\kappa_\gamma &0 & \tau_\gamma \\ 0 & -\tau_\gamma & 0 \end{pmatrix}  
		\begin{pmatrix} \bT\\ \bN_{\gamma}\\\bB_{\gamma}\end{pmatrix},
	$$
	where $\kappa_\gamma\ge 0$ and $\tau_\gamma$ are called respectively the curvature and the torsion of $\gamma$.
	The torsion $\tau_{\gamma}$ corresponds to the derivative of the angle $\measuredangle (X(s),\bN_\gamma(s))$. 
	We thus recover the fact that the total torsion (\emph{i.e.} the integral of $\tau_{\gamma}$ along the curve) 
	is equal to the integrated relative torsion modulo $2\pi$:
	\[
		\int_0^{\ell}\tau_\gamma(s)\d s
		=\int_{\gamma}\tau\,\bT^{\flat}\!\!\mod 2\pi.
	\]

	\item The integrated relative torsion is conformally invariant in the following sense. 
	Say we conformally change the metric in a tubular neighborhood $B_\eps[\gamma]$, 
	the new metric being $g_\Omega=\Omega^2g_3$ with conformal factor $\Omega>0$.
	By \cite[section 4]{MR1860416}, the new Levi-Civita connection $\nabla^{(\Omega)}$ on $\rT B_\eps[\gamma]$ is:
	\[
		\nabla^{(\Omega)}_XY=\nabla_X Y+\frac{1}{\Omega}
		\Big(X[\Omega]Y+Y[\Omega]X-g_3(X,Y)(\d \Omega)^{\sharp} \Big).
	\]
	In this new metric we easily get that the relative torsion $\tau^{(\Omega)}$ at point $\gamma(s)$ is:
	\[
		\tau^{(\Omega)}(\gamma(s))=\frac{\tau_S(\gamma(s))}{\Omega(\gamma(s))},
	\]
	thus the integrated relative torsion for the two metric coincides. We emphasize that the relative 
	torsion is then defined for the renormalized Seifert frame $(\Omega^{-1}\bT,\Omega^{-1}\bS,\Omega^{-1}\bN)$.
	In particular through a stereographic projection, 
	we can study the integrated relative torsion in $\R^3$ with the flat metric $g_{\R^3}$. 

	Let us say then that the $g_{\R^3}$-Frenet frame is defined in $\R^3$. By \cite{total_torsion}, we know that 
	\begin{itemize}
		\item[-] the $g_{\R^3}$-total torsion is $0\!\!\mod \pi$ if and only if the curve is a line of curvature of a surface,
		\item[-] the $g_{\R^3}$-total torsion is $0\!\!\mod 2\pi$ if and only if the curve is a line of curvature 
			of a close oriented surface of genus $1$.
	\end{itemize}
\end{enumerate}

\subsubsection{Link with the writhe of a curve}\label{sec:link_writhe}
We have seen above that the integrated relative torsion of a knot was a conformal invariant.
So for a knot $\gamma\subset\S^3$, we consider its image $\gamma_0$ through a stereographic
projection, and identify this curve with its arclength parametrization in $\R^3$. 
Consider now a Seifert surface for $\gamma_0$. This induced a copy of $\gamma_0$ 
through $\wt{\gamma}_0(s):=\gamma_0(s)+\varepsilon \bn(s)\in\R^3$ where $(\bt,\bs,\bn)$ is the Seifert
frame and $\varepsilon>0$ a small number. By C\u{a}lug\u{a}reanu-White-Fuller Theorem (see \cite{Calu61}*{p. 613} and \cite{White69}) 
we have:
\[
\link(\wt{\gamma}_0,\gamma_0)=\mathrm{Tw}(\bn)+\Wr(\gamma_0),
\]
where the twist $\mathrm{Tw}(\bn)$ of $\bn$ along $\gamma_0$ is:
\[
\mathrm{Tw}(\bn):=\frac{1}{2\pi}\int_{0}^{\ell(\gamma_0)}\cip{\tfrac{\d\bn}{\d s}(s)}{\bt(s)\times\bn(s)}_{\R^3}\d s,
\]
and the writhe of $\gamma_0$ is:
\[
\Wr(\gamma_0):=\frac{1}{4\pi}\iint_{[0,\ell(\gamma_0)]^2} 
	\cip{\bt(s_1)\times\bt(s_2)}{\frac{\gamma_0(s_1)-\gamma_0(s_2)}{|\gamma_0(s_1)-\gamma_0(s_2)|^3}}_{\R^3}\d s_1\d s_2.
\]
In the special case of $\bn$, we have $\mathrm{Tw}(\bn)=\frac{1}{2\pi}I_\tau(\gamma)$ and $\link(\wt{\gamma}_0,\gamma_0)=0$, hence
the integrated relative torsion is $-2\pi$ times the writhe. As this result holds for any stereographic projection, we define the writhe
of $\gamma$ on $\S^3$ as the writhe of any of its stereographic projection on $\R^3$.

\subsubsection{Varying the integrated relative torsion}\label{sec:varying_I_tau}
Let $\gamma$ be a knot. In this section we give an explicit procedure to deform
a knot which changes $I_\tau(\gamma)$ while staying in the same isotopy class.

In the previous part we have shown that $I_\tau(\gamma)$ is a conformal invariant, so it suffices
to prove the result in $\R^3$, and we consider an arc length parametrization $\gamma:\T_\ell\to \R^3$,
where $\ell$ is the $\R^3$-length of $\gamma$. Let us consider a Seifert surface $S$ with Seifert frame
$(\bt,\bs,\bn)$, with geodesic and normal curvatures $\kappa_g,\kappa_n$ and relative torsion $\tau$
defined in the metric of $\R^3$.

We pick $r_0>0$ small enough\footnote{smaller than the distance of $\gamma$ from its cut locus, 
that is, the set of points having at least two geodesics minimizing the distance from $\gamma$.} 
and $n \in \N$. We introduce $j:=\frac{2\pi}{\ell}n$ and $a:=r j$.
We define two families $(\gamma_{r,n,\mp})_{0\le r\le r_0}$ of deformations of $\gamma$ by the formulas
\begin{equation}\label{def_deform}
	\begin{array}{rcl}
	\gamma_{r,n,\mp}(s)&:=&\exp_{\gamma(s)}\big\{r(\cos(js)\bs(s)\mp \sin(js)\bn(s))\big\}\\
					&=&\gamma(s)+r(\cos(js)\bs(s)\mp \sin(js)\bn(s)).
	\end{array}
\end{equation}

\begin{proposition}\label{prop:varying_I_tau}
	Let $a_0>0$ and let $\gamma_{r,n,\mp}$ be defined as in \eqref{def_deform}. 
	Then as $n\to +\infty$, we have in the regime $0< a\le a_0$
	\[
		I_\tau(\gamma_{r,n,\mp})=\pm 2\pi n\Big(1-\frac{1}{\sqrt{1+a^2}}\Big)+\underset{n\to+\infty}{\mathcal{O}}(1).
	\]
\end{proposition}
\begin{proof}
	We only deal with the case of $\gamma_{r,n,-}$ and drop the subscripts $n,-$ for short 
	(the case $\gamma_{r,n+}$ is dealt with in a smilar way). 
	Furthermore, the dot above function denotes the derivative with respect to $s\in\T_\ell$. 
	We will proceed as follows. We introduce a smooth one-parameter family 
	$(\mathbf{W}(r))_{0\le r\le r_0}$ of direct orthonormal frames associated to $(\gamma_r)_{r\le r_0}$, 
	which converge to the Seifert frame $(\bt,\bs,\bn)$ as $r$ tends to $0$, say, in the norm
	topology of $C^1(\T_\ell;\S^2)^3$, and where 
	$$
		\mathbf{W}(r)=\big(W_0(r):=\norm{\dot{\gamma}_r}^{-1}\dot{\gamma}_r,W_1(r),W_2(r)\big).
	$$
	By continuity, this implies that for all $r$, the relative degree of $\mathbf{W}(r)$ with the Seifert frames of $\gamma_r$ is $0$.
	
	Using the same method as in the proof of the independence of $I_\tau(\gamma)$ with respect to the Seifert surface, 
	defining the relative torsion of $\mathbf{W}(r)$
	\[
		\tau_{\mathbf{W}(r)}:=\cip{\dot{W}_1(r)}{W_2(r)}_{\R^3}=-\cip{W_1(r)}{\dot{W}_2(r)}_{\R^3},
	\]
	we obtain
	\[
		\int_0^\ell \tau_{\mathbf{W}(r)}(s)\norm{\dot{\gamma}_r(s)}ds
		=\mathrm{relative\,deg}\,(\mathbf{W}(r),\mathrm{Seifert\ frame})+I_\tau(\gamma_r)=I_\tau(\gamma_r).
	\]
	
	Setting
	\[
		V_2(r,s):=\dot{\gamma}_r(s)\times \bs(s)\quad\&\quad V_1(r,s):=V_2(r,s)\times \dot{\gamma}(s),
	\]
	we define $W_1(r,s):=\norm{V_1(r,s)}^{-1}V_1(r,s)$, $W_2(r,s):=\norm{V_2(r,s)}^{-1}V_2(r,s)$.	
	A computation shows:
	\begin{multline*}
		\dot{\gamma}(s)=\big(1-\eps(\cos(js)\kappa_g(s)-\sin(js)\kappa_n(s))\big)\bt(s)\\
		-(a-\eps\tau_r(s))\big(\sin(js)\bs(s)+\cos(js)\bn(s)\big).
	\end{multline*}
	Writing 
	\[
		h(r,s):=1-r(\cos(js)\kappa_g(s)-\sin(js)\kappa_n(s)), \quad m(r,s):=a-r\tau(s),
	\] 
	we also have
	\begin{equation*}
		\left\{
		\begin{array}{rcl}
			V_2(r,s)&=&m(r,s)\cos(js)\bt(s)+h(r,s)\bn(s),\\
			V_1(r,s)&=&h(r,s)m(r,s)\sin(js)\bt(s)\\
			&&\ +\big(h(r,s)^2+\cos(js)^2m(r,s)^2\big)\bs(s)-m(r,s)^2\frac{\sin(2js)}{2}\bn(s).
		\end{array}
		\right.
	\end{equation*}
	The product of the norms of the $V$'s is
	\[
		\norm{V_1(r,s)}\norm{V_2(r,s)}=(h(r,s)^2+m(r,s)^2\cos(js)^2)\sqrt{h(r,s)^2+m(r,s)^2},
	\]
	and their scalar product computes to
	\begin{multline*}
	-\cip{V_1}{\dot{V}_2}_{\R^3}=jhm^2\sin(js)^2+hm\sin(js)\big[h\kappa_n(s)+r\dot{\tau}(s)\cos(js)\big]\\
			+(\tau(s) h-m\kappa_g(s))(h^2+m^2\cos(js)^2)+m^2\frac{\sin(2js)}{2}(\dot{h}+\kappa_n(s) m).
	\end{multline*}
	The formula follows easily; in the regime $0< a\le a_0$, we have
	\begin{align*}
		\int_0^\ell \tau_{\mathbf{W}(r)}(s)\norm{\dot{\gamma}_r(s)}ds
		&=-\int_0^\ell\frac{\cip{V_1(r,s)}{\dot{V}_2(r,s)}_{\R^3}}{\norm{V_1(r,s)}\norm{V_2(r,s)}}ds,\\
		&=\frac{ja^2}{\sqrt{1+a^2}}\int_0^\ell \frac{\sin(js)^2}{1+a^2\cos(js)^2}ds+\underset{n\to+\infty}{\mathcal{O}}(1),\\
		&=2\pi n\Big(1-\frac{1}{\sqrt{1+a^2}}\Big)+\underset{n\to+\infty}{\mathcal{O}}(1).
	\end{align*}
\end{proof}

\begin{rem}[Examples]
	We apply Proposition~\ref{prop:varying_I_tau} for the unit circle. 
	For $0<r<1$ and $n\in\mathbb{N}$, let $\gamma_{r,n,\mp}:(\R/(2\pi\Z))\to\R^3$ be the knots
	\[
		\gamma_{r,n,\mp}(s):=\begin{pmatrix}([1+r\cos(ns)]\cos(s)& [1+ r\cos(ns)]\sin(s) & \pm r\sin(ns) \end{pmatrix}^T.
	\]
	For the unit circle ($r=0$), we have $I_\tau=0$.
	We then obtain the following formula, where $a$ denotes $rn$,
	\begin{align*}
		I_\tau(\gamma_{r,n,\mp})&=\pm\int_0^{2\pi}\frac{d t}{\sqrt{a^2+(1+r\cos(t))^2}}
		\bigg[\frac{a^2n\sin(t)^2}{a^2\cos(t)^2+(1+r\cos(t))^2}+a\cos(t)\bigg],\\
		&=\pm 2\pi n\Big(1-\frac{1}{\sqrt{1+a^2}}\Big)+\mathcal{O}\big(\frac{r}{\sqrt{1+a^2}}\big).
	\end{align*}
	Here the $\mathcal{O}(r/(1+a^2)^{-1/2})$ means that the error is bounded by $\tfrac{Cr}{\sqrt{1+a^2}}$, 
	where $C$ depends neither on $r$ nor on $a$.
\end{rem}

\subsection{Computation of $\int_{\gamma}\omega_{\xi}$}\label{app:comp_omega_xi}
In Section~\ref{sec:dirac_op_def}, we have introduced (on the vector bundle $B_\eps[\gamma]\times\C^2$)
two smooth sections $\xi_{\pm}$ with unit length
of the line bundles corresponding to the two eigenspaces
$$
	\ker(\bsigma(\bT^{\flat})\mp 1).
$$
Their relative phase is fixed by the pointwise condition
\[
	\omega(\bS)+i\omega(\bN)=\cip{\xi_-}{\bsigma(\omega)\xi_+}, \quad \forall\,\omega\in\Omega^1(B_{\eps}[\gamma]).
\]
Furthermore, $\omega_{\xi}$ is the one form
\[
	\omega_{\xi}(X):=\frac{i}{2}\big(\cip{\xi_+}{\nabla_X \xi_+} +\cip{\xi_-}{\nabla_X \xi_-} \big).
\]
We will now prove that 
$$
	\alpha_\xi:=\int_\gamma \omega_\xi = \pi\!\!\mod 2\pi.
$$
First observe the following remarks:
\begin{enumerate}
	\item The class of $\alpha_\xi$ does not depend on the choice of the sections $\xi_{\pm}$. Indeed, 
	if we change the gauge
	$$
		\wt{\xi}_{\pm}:=e^{i\phi}\xi_{\pm},\quad \phi\in C^\infty(\gamma),
	$$
	we have
	$$
		\int_\gamma \omega_{\wt{\xi}}
		=-\int_{\gamma}\d\phi+\int_\gamma \omega_{\xi}
		=\alpha_\xi\!\!\mod 2\pi.
	$$
		
	\item The class of $\alpha_\xi$ does not depend on the choice of the Seifert surface; let us consider another
		Seifert surface $S'$ with Seifert frame $(\bT,\bS',\bN')$ which makes an angle $\vartheta$ with 
		$(\bT,\bS,\bN)$.
		The corresponding unital sections $\xi_+'=e^{i\phi_+}\xi_+$ and $\xi_-'=e^{i\phi_-}\xi_-$ 
		must satisfy
		$$
			\phi_--\phi_+=\vartheta\!\!\mod 2\pi.
		$$
		As we have $\int_{\gamma}\d\vartheta=0$, the following holds:
		\begin{align*}
			\int_\gamma \omega_{\xi'}
			=-\frac{1}{2}\int_\gamma \d\vartheta+\int_\gamma \d\phi_+
			+\int_\gamma \omega_\xi
			=\alpha_\xi \!\! \mod2\pi.
		\end{align*}
	
	\item Using \cite[Section 4]{MR1860416}, we see that the class of $\alpha_\xi$ is a 
		conformal invariant in the following sense. 
		We conformally change the metric in a tubular neighborhood 
		$B_\eps[\gamma]$ of the knot $\gamma$ to
		$$
			g'=\Omega^2g_3,
		$$
		where $\Omega$ is a smooth conformal factor.
		The new connection $\nabla^{(\Omega)}$ on the spin spinor bundle $\Psi$ is given by
		\[
			\nabla_X^{(\Omega)}=\nabla_X+\frac{1}{4\Omega}\big[\bsigma(X^{\flat}),\bsigma(\d\Omega)\big],
		\]
		and we write $\omega_\xi^{(\Omega)}$ the new $1$-form. 
		As
		\begin{align*}
			\big[\bsigma(\bT^{\flat}),\bsigma(\d\Omega) \big]
			&=\big[\bsigma(\bT^{\flat}),\bsigma(\bT^{\flat})\bT[\Omega]\big]
			+\big[\bsigma(\bT^{\flat}),\bsigma(\bS^{\flat})\bS[\Omega] \big]\\
			&\quad+\big[\bsigma(\bT^{\flat}),\bsigma(\bN^{\flat})\bN[\OmegaØ \big]\\
			&=i\big(\bsigma(\bN^{\flat})\bS[\Omega]- \bsigma(\bS^{\flat})\bN[\Omega]\big),
		\end{align*}
		and $\cip{\xi_{\pm}}{\bsigma(\bS^{\flat})\xi_{\pm}}=\cip{\xi_{\pm}}{\bsigma(\bN^{\flat})\xi_{\pm}}=0$,
		we obtain
		\[
			\int_{\gamma}\omega_\xi^{(\Omega)}=\int_{\gamma}\omega_\xi.
		\]
\end{enumerate}

Since the Seifert surface does not cover all $\S^3$, there exists a point $\bp\in \S^3\setminus S$, 
and we can work in $\R^3$, seen as a chart of $\S^3\setminus \{\bp\}$ through the stereographic 
projection $\st$ with respect to $\bp$. Thanks to the last remark we can work in $\R^3\times\C^2$ with 
the flat metric $g_{\R^3}$ of $\R^3$ and the canonical connection 
$\nabla^{\R^3}$ on $\R^3\times\C^2$ (with Clifford map $\bsigma_{\R^3}$). 
Up to a change of sections, we can assume that the isometry
$\bsigma_{\R^3}$ is given by
\[
	\bsigma_{\R^3}(\omega)=\bsigma\cdot \big(g_{\R^3}(\d x_j,\omega)\big)_{1\le j\le 3}, \quad \forall \omega \in \Omega^1(\R^3).
\]
In other words, after identifying vectors and $1$-forms in $\R^3$, we can assume that for all vector fields $X=X_i\partial_{x_i}$, we have the usual
formula:
\[
	\bsigma_{\R^3}(X)=\sigma\cdot X=\sigma_1 X_1+\sigma_2 X_2+\sigma_3 X_3.
\]

We write $c:\R/\ell_{\R^3}\Z\to \R^3$ the arclength parametrization of the loop $\st\circ\gamma$ in $\R^3$. 
Furthermore $(\bT_{\Omega_3},\bS_{\Omega_3},\bN_{\Omega_3})$
denotes  the push forward of the Seifert frame by $\st$ and $(\mathbf{t},\mathbf{s},\mathbf{n})$ the normalized frame in the $g_{\R^3}$-metric.
For simplicity we write $\xi_\pm(r):=\xi_{\pm}(\st^{-1}(c(r)))$.
Up to a global change of sections $U_0\in \mathbf{SU}(2)$ ($(\bp,\psi)\in\S^3\times\C^2\mapsto (\bp,U_0\psi)$) 
we can assume that $|t_3(\bx)|\neq 1$ for all $\bx\in\st (B_\eps[\gamma])$.  
Then there exists $z_+,z_-:\R/\ell_{\R^3}\Z\to \S^1$ such that:
\[
	\xi_+=\frac{z_+}{\sqrt{2(1-t_3)}}\begin{pmatrix} t_1-it_2 \\ 1-t_3\end{pmatrix},\quad\xi_-
	=\frac{z_-}{\sqrt{2(1-t_3)}}\begin{pmatrix} t_3-1\\ t_1+it_2\end{pmatrix}.
\]
If we see $(\xi_+,\xi_-)=U_{\xi}$ as a $2\times 2$-matrix, the matrix $U_\xi$ is in $\mathbf{U}(2)$ and its associated rotation $R_\xi$ defined by
\[
	U_{\xi}^*(\bsigma\cdot v)U_{\xi}=\bsigma\cdot(R_\xi v),\quad v\in\R^3
\]
maps the Seifert frame $(\mathbf{s},\mathbf{n},\mathbf{t})$ onto the canonical basis of $\R^3$. Remark the following:
\begin{equation}
	M(r):=\frac{1}{\sqrt{2(1-t_3)}}\begin{pmatrix} t_1-it_2 & t_3-1\\   1-t_3 &  t_1+it_2 \end{pmatrix}(c(r))\in \mathbf{SU}(2).
\end{equation}

\begin{lemma}\label{lem:loop_not_trivial}
	The loop 
	$$
		r\in\R/\ell_{\R^3}\Z \mapsto R_\xi(c(r))\in \SO(3)
	$$ 
	is not trivial in the fundamental group $\pi_1(\SO(3))\simeq\Z/(2\Z)$. 
\end{lemma}
This lemma is proved below. For each $r\in [0,\ell_{\R^3})$, 
up to multiplying $U_{\xi}(c(r))$ by a phase $z(r)\in \S^1$, we can assume that 
$\wt{U}_{\xi}(c(r))=z(r)U_{\xi}(c(r))$ is in $\mathbf{SU}(2)$.
We can also assume that $z$ is smooth up to $\ell_{\R^3}^{-}$, but by Lemma~\ref{lem:loop_not_trivial}, 
we have $z(\ell_{\R^3})\overline{z(0)}=-1$.
For each $r\in [0,\ell_{\R^3}]$, there exists a phase $w(r)\in\S^1$ such that
\[
	\wt{U}_\xi(r)=\frac{1}{\sqrt{2(1-t_3)}}\begin{pmatrix} w(t_1-it_2) & \overline{w}(t_3-1)\\
	w(1-t_3) &  \overline{w}(t_1+it_2) \end{pmatrix}(c(r)).
\]
We can assume that $w:[0,\ell_{\R^3}]\to \S^1$ is smooth but by the 
Lemma~\ref{lem:loop_not_trivial} we have $w(\ell_{\R^3})\overline{w(0)}=-1$. 
We thus obtain:
\[
	z_+(r)=w(r)\overline{z(r)},\quad z_-(r)=\overline{w(r)}\overline{z(r)}.
\]
Writing $z(r)=\exp(i\phi_z(r))$ and $w(r)=\exp(i\phi_w(r))$ we easily get:
\begin{align*}
	\frac{i}{2}\int_0^{\ell_{\R^3}}\!\big[\cip{\xi_+}{\frac{\d}{\d r}\xi_+}(r)+\cip{\xi_-}{\frac{\d}{\d r}\xi_-}(r)\big]\d r
	&=i\int_0^{\ell_{\R^3}}\!\big[z\overline{z}'(r)+\tfrac{1}{2}(|w|^2)'\big]d r\\
	&=i\int_0^{\ell_{\R^3}}z\!\overline{z}'(r)d r=\pi\!\!\mod 2\pi.
\end{align*}

\subsubsection*{Proof of Lemma~\ref{lem:loop_not_trivial}}

As we know that two Seifert frames have trivial relative degree, we can assume that the Seifert surface
under consideration can be obtained from the Seifert algorithm. Following the construction,
we have $S$ Seifert circles and $C$ crossings. We can assume that the diagram is drawn on the projection onto $x_3=0$.
The formula of the genus of the Seifert surface of a knot ($g=1-\frac{1}{2}(S-C+1)$) implies that $S+C$ is odd (see \cite{Rolfsen}*{Part 5.A}). 
Let us show that this parity gives the class of the loop $R_\xi^{-1}$ in $\mathbf{SO}(3)$. 
We define an auxiliary (direct and orthonormal) moving frame $(\bT,\wt{\mathbf{e}}_2,\wt{\mathbf{e}}_3)$,
in which we set 
\[
\wt{\mathbf{e}}_3:=\frac{\mathbf{e}_3-\cip{\bT}{\mathbf{e}_3}_{\R^3}}{\norm{\mathbf{e}_3-\cip{\bT}{\mathbf{e}_3}_{\R^3}}}.
\] 
It is well-defined up to a continuous deformation of $\gamma$ and it gives rise to another loop $s\mapsto \wt{R}(s)$ in $\mathbf{SO}(3)$ 
(mapping the canonical basis $(\mathbf{e}_1,\mathbf{e}_2,\mathbf{e}_3)$ of $\R^3$ onto the auxiliary frame). 
Here $\bT$ denotes the tangent vector of $\gamma$. 
The parity of $S$ corresponds to the class of the loop $s\mapsto \wt{R}(s)$.
It is not difficult to see that the parity of $C$ gives the class of the loop $s\mapsto R_\xi^{-1}(s)\wt{R}(s)^{-1}$ 
(these last rotations are the rotations that map the auxiliary moving frame onto the Seifert frame). This proves that
the class of the loop $R_\xi^{-1}$ is given by the parity of $S+C$, which is odd.
\qed


\end{document}